\documentclass[
paper=letter,%
numbers=noendperiod,%
captions=nooneline,%
abstracton,%
DIV=10%
]{scrartcl}

\usepackage[T1]{fontenc}%
\usepackage{lmodern}%
\usepackage[american]{babel}%
\usepackage{microtype}%

\usepackage[
hyperref,%
table%
]{xcolor}%

\usepackage{scrlayer-scrpage}%
\usepackage{eso-pic}%
\usepackage{rotating}%
\usepackage{amsmath}%
\usepackage{mathtools}%
\usepackage{amssymb}%
\usepackage{amsthm}%
\usepackage{thmtools}%
\usepackage{etoolbox}%
\usepackage{bm}%
\usepackage{bbm}%
\usepackage{enumitem}%
\usepackage{graphicx}%
\usepackage{grffile}%
\usepackage{tikz}%
\usepackage{wrapfig}%
\usepackage{tabularx}%
\usepackage{siunitx}%
\usepackage{booktabs}%
\usepackage{multirow}%
\usepackage{vruler}%
\usepackage{fancyvrb}%
\usepackage{listings}%
\usepackage{csquotes}%
\usepackage[
style=authoryear,%
dashed=false,%
hyperref=true,%
useprefix=true,%
maxcitenames=2,%
maxbibnames=6%
]{biblatex}%
\usepackage[
hypertexnames=false,%
setpagesize=false,%
pdfborder={0 0 0},%
pdfstartview=Fit,%
bookmarksopen=true,%
bookmarksnumbered=true%
]{hyperref}%
\xdefinecolor{lightgray}{RGB}{247, 247, 247}%
\xdefinecolor{semilightgray}{RGB}{240, 240, 240}%
\xdefinecolor{middlegray}{RGB}{127, 127, 127}%
\xdefinecolor{blue}{RGB}{58, 95, 205}%
\xdefinecolor{deepskyblue}{RGB}{0, 154, 205}%
\xdefinecolor{chocolate}{RGB}{205, 102, 29}%

\setkomafont{pageheadfoot}{\normalfont\normalcolor\sffamily}%
\setkomafont{pagenumber}{\normalfont\normalcolor\sffamily}%
\automark{section}%
\setkomafont{captionlabel}{\normalfont\normalcolor\sffamily\bfseries}%

\setlist{%
  align=left,%
  labelsep=*,%
  leftmargin=*,%
  topsep=1mm,%
  itemsep=0mm%
}
\newcommand*{\mysquare}{\rule[0.18em]{0.36em}{0.36em}}
\newcommand*{\mytriangle}{\raisebox{0.12em}{\resizebox{0.48em}{0.48em}{$\blacktriangleright$}}}
\newcommand*{\mybar}{\rule[0.32em]{0.62em}{0.08em}}
\newcommand*{\mydot}{\raisebox{0.14em}{\resizebox{0.44em}{!}{$\bullet$}}}
\setlist[itemize,1]{label={\mysquare\ }}%
\setlist[itemize,2]{label={\mytriangle\ }}%
\setlist[itemize,3]{label={\mybar\ }}%
\setlist[itemize,4]{label={\mydot\ }}%
\setlist[enumerate,1]{label=\arabic*)}%
\setlist[enumerate,2]{label=\arabic{enumi}.\arabic*)}%
\setlist[enumerate,3]{label=\arabic{enumi}.\arabic{enumii}.\arabic*)}%

\makeatletter
\newcommand\myisodate{\number\year-\ifcase\month\or 01\or 02\or 03\or 04\or 05\or 06\or 07\or 08\or 09\or 10\or 11\or 12\fi-\ifcase\day\or 01\or 02\or 03\or 04\or 05\or 06\or 07\or 08\or 09\or 10\or 11\or 12\or 13\or 14\or 15\or 16\or 17\or 18\or 19\or 20\or 21\or 22\or 23\or 24\or 25\or 26\or 27\or 28\or 29\or 30\or 31\fi}%
\makeatother
\newcommand*{\abstractnoindent}{}%
\let\abstractnoindent\abstract
\renewcommand*{\abstract}{\let\quotation\quote\let\endquotation\endquote
  \abstractnoindent}
\deffootnote[1em]{1em}{1em}{\textsuperscript{\thefootnotemark}}%
\lstdefinestyle{input}{
  backgroundcolor=\color{semilightgray},%
  commentstyle=\itshape\color{black},%
  keywordstyle=\color{black},%
  emphstyle=\color{black},%
  stringstyle=\color{black},%
  numbers=left,%
  numbersep=4.8pt,%
  numberstyle=\color{darkgray!80}\tiny%
}
\lstdefinestyle{output}{
  backgroundcolor=\color{lightgray}%
}

\lstdefinestyle{Rstyle}{
  language=R,%
  keywords={function, if, else, switch, repeat, while, for, in, next, break},%
  otherkeywords={},%
  emph={TRUE, FALSE, NULL, NA, NaN, Inf}%
}
\expandafter\let\csname Sinput\endcsname\relax
\expandafter\let\csname endSinput\endcsname\relax
\expandafter\let\csname Soutput\endcsname\relax
\expandafter\let\csname endSoutput\endcsname\relax
\lstnewenvironment{Sinput}[1][]{%
  \lstset{style=input, style=Rstyle}
  #1%
}{\vspace{-0.25\baselineskip}}%
\lstnewenvironment{Soutput}[1][]{%
  \lstset{style=output, style=Rstyle}
  #1%
}{\vspace{-0.25\baselineskip}}%

\lstdefinestyle{LaTeXstyle}{
  language=[LaTeX]TeX,%
  texcs={},%
  otherkeywords={}%
}
\lstnewenvironment{LaTeXinput}[1][]{%
  \lstset{style=input, style=LaTeXstyle}
  #1%
}{\vspace{-0.25\baselineskip}}%
\lstnewenvironment{LaTeXoutput}[1][]{%
  \lstset{style=output, style=LaTeXstyle}
  #1%
}{\vspace{-0.25\baselineskip}}%

\lstdefinestyle{otherstyle}{
  language={},%
  otherkeywords={},%
  upquote=true%
}
\lstnewenvironment{otherinput}[1][]{%
  \lstset{style=input, style=otherstyle}
  #1%
}{\vspace{-0.25\baselineskip}}%
\lstnewenvironment{otheroutput}[1][]{%
  \lstset{style=output, style=otherstyle}
  #1%
}{\vspace{-0.25\baselineskip}}%
\newcommand*{\code}{\lstinline[
  basicstyle=\upshape\ttfamily,
  style=otherstyle,
  literate={~}{{$\sim$}}1
  ]}

\DeclareNameAlias{sortname}{family-given}%
\DefineBibliographyExtras{american}{\DeclareQuotePunctuation{}}%
\renewbibmacro*{volume+number+eid}{%
  \setunit*{\addcomma\space}%
  \printfield{volume}%
  \printfield{number}}
\DeclareFieldFormat*{number}{(#1)}
\DeclareFieldFormat*{title}{#1}%
\DeclareFieldFormat{doi}{%
  \ifhyperref
    {\href{http://dx.doi.org/#1}{\nolinkurl{doi:#1}}}%
    {\nolinkurl{doi:#1}}}%
\renewbibmacro*{in:}{}%
\DeclareFieldFormat{isbn}{ISBN #1}%
\DeclareFieldFormat{pages}{#1}%
\DeclareFieldFormat{url}{\url{#1}}%
\DeclareFieldFormat{urldate}{\mkbibparens{#1}}%
\renewcommand*{\cite}[2][]{\textcite[#1]{#2}}%

\newif\ifstarttheorem
\declaretheoremstyle[%
  spaceabove=0.5em,
  spacebelow=0.5em,
  headfont=\sffamily\bfseries\global\starttheoremtrue,
  notefont=\sffamily\bfseries,
  notebraces={(}{)},
  headpunct={},
  bodyfont=\normalfont,
  postheadspace=\newline%
]{myMainStyle}
\declaretheorem[style=myMainStyle, numberwithin=section]{definition}%

\declaretheorem[style=myMainStyle, sibling=definition]{lemma}

\declaretheorem[style=myMainStyle, sibling=definition]{remark}

\declaretheorem[style=myMainStyle, sibling=definition]{algorithm}

\makeatletter
\preto\itemize{%
  \if@inlabel
    \ifstarttheorem
      \mbox{}\par\nobreak\vskip\glueexpr-\parskip-\baselineskip+0.25em\relax\hrule\@height\z@
    \fi%
  \fi%
  \global\starttheoremfalse%
 \def\tempa{proof}%
 \ifx\tempa\mycurrenvir
    \ifstarttheorem
      \mbox{}\par\nobreak\vskip\glueexpr-\parskip-\baselineskip+0.25em\relax\hrule\@height\z@
    \fi%
 \fi%
 \global\starttheoremfalse%
}
\preto\enditemize{\global\starttheoremfalse}
\makeatother

\makeatletter
\preto\enumerate{%
  \if@inlabel
    \ifstarttheorem
      \mbox{}\par\nobreak\vskip\glueexpr-\parskip-\baselineskip+0.25em\relax\hrule\@height\z@
    \fi%
  \fi%
  \global\starttheoremfalse%
 \def\tempa{proof}%
 \ifx\tempa\mycurrenvir
    \ifstarttheorem
      \mbox{}\par\nobreak\vskip\glueexpr-\parskip-\baselineskip+0.25em\relax\hrule\@height\z@
    \fi%
 \fi%
 \global\starttheoremfalse%
}
\preto\endenumerate{\global\starttheoremfalse}
\makeatother

\newcommand*{\omu}[3]{\underset{#3}{\overset{#1}{#2}}}
\newcommand*{\T}{^{\top}}
\renewcommand*{\i}{\leftarrow}
\newcommand*{\isim}{\omu{\text{\tiny{ind.}}}{\sim}{}}

\newcommand*{\IN}{\mathbb{N}}

\newcommand*{\IR}{\mathbb{R}}

\newcommand*{\PNVM}{\operatorname{PNVM}}
\newcommand*{\MVT}{\operatorname{MVT}}

\newcommand*{\Par}{\operatorname{Par}}

\newcommand*{\NVM}{\operatorname{NVM}}
\newcommand*{\U}{\operatorname{U}}

\newcommand*{\N}{\operatorname{N}}

\newcommand*{\rd}{\mathrm{d}}

\newcommand*{\argmax}{\operatorname*{argmax}}
\newcommand*{\argmin}{\operatorname*{argmin}}

\newcommand*{\LSE}{\operatorname{LSE}}
\renewcommand*{\P}{\mathbb{P}}
\newcommand*{\E}{\mathbb{E}}

\newcommand*{\var}{\operatorname{var}}
\newcommand*{\cov}{\operatorname{cov}}
\newcommand*{\corr}{\operatorname{corr}}

\newcommand*{\R}{\textsf{R}}
\newcommand*{\eps}{\varepsilon}
\newcommand*{\btheta}{\bm{\theta}}
\newcommand*{\ba}{\bm{a}}
\newcommand*{\bnu}{\bm{\nu}}
\newcommand*{\bb}{\bm{b}}
\newcommand*{\bx}{\bm{x}}
\newcommand*{\by}{\bm{y}}

\newcommand*{\bu}{\bm{u}}
\newcommand*{\bU}{\bm{U}}
\newcommand*{\bv}{\bm{v}}

\newcommand*{\bZ}{\bm{Z}}
\newcommand*{\bY}{\bm{Y}}
\newcommand*{\bX}{\bm{X}}

\newcommand*{\bzero}{\bm{0}}
\newcommand*{\bone}{\bm{1}}
\newcommand*{\bmu}{\bm{\mu}}

\newcommand*{\ta}{\tilde{a}}
\newcommand*{\tb}{\tilde{b}}

\newcommand*{\hd}{\hat{d}}
\newcommand*{\he}{\hat{e}}
\newcommand*{\hmuMCn}{\hat{\mu}^\text{\tiny{MC}}_n}
\newcommand*{\hsigma}{\hat{\sigma}}
\newcommand*{\hmurqmc}{\hat{\mu}^\text{\tiny{RQMC}}}
\newcommand*{\hmurqmcn}{\hat{\mu}^\text{\tiny{RQMC}}_{n}}
\newcommand*{\imax}{i_{\max}}
\begin{document}
\thispagestyle{plain}
\begin{center}
  \sffamily
  {\bfseries\LARGE Normal variance mixtures: Distribution,
density and parameter estimation
\par}
  \bigskip\smallskip
  {\Large
    Erik Hintz\footnote{Department of Statistics and Actuarial Science, University of
    Waterloo, 200 University Avenue West, Waterloo, ON, N2L
    3G1,
    \href{mailto:erik.hintz@uwaterloo.ca}{\nolinkurl{erik.hintz@uwaterloo.ca}}.},
    Marius Hofert\footnote{Department of Statistics and Actuarial Science, University of
    Waterloo, 200 University Avenue West, Waterloo, ON, N2L
    3G1,
    \href{mailto:marius.hofert@uwaterloo.ca}{\nolinkurl{marius.hofert@uwaterloo.ca}}. The
    author would like to thank NSERC for financial support for this work through Discovery
    Grant RGPIN-5010-2015.},
    Christiane Lemieux\footnote{Department of Statistics and Actuarial Science, University of
    Waterloo, 200 University Avenue West, Waterloo, ON, N2L
    3G1,
    \href{mailto:clemieux@uwaterloo.ca}{\nolinkurl{clemieux@uwaterloo.ca}}. The
    author would like to thank NSERC for financial support for this work through Discovery
    Grant RGPIN-238959.}
    \par
    \bigskip
    \myisodate\par}
\end{center}
\par\smallskip
\begin{abstract}
Normal variance mixtures are a class of multivariate distributions that generalize the multivariate normal by randomizing (or mixing) the covariance matrix via multiplication by a non-negative random variable $W$. The multivariate $t$ distribution is an example of such mixture, where $W$ has an inverse-gamma distribution. Algorithms to compute the joint distribution function and perform parameter estimation for the multivariate normal and $t$ (with integer degrees of freedom) can be found in the literature and are implemented in, e.g., the \R\ package \texttt{mvtnorm}. In this paper, efficient algorithms to perform these tasks in the general case of a normal variance mixture are proposed. In addition to the above two tasks, the evaluation of the joint (logarithmic) density function  of a general normal variance mixture  is tackled as well, as it is needed for parameter estimation and does not always exist in closed form in this more general setup. For the evaluation of the joint distribution function, the proposed algorithms apply randomized quasi-Monte Carlo (RQMC) methods in a way that improves upon existing methods proposed for the multivariate normal and $t$ distributions. An adaptive RQMC algorithm that similarly exploits the superior convergence properties of RQMC methods is presented for the task of evaluating the joint log-density function. In turn, this allows  the parameter estimation task to be accomplished via an expectation-maximization-like algorithm where all weights and log-densities are numerically estimated.
It is demonstrated through numerical examples that the suggested algorithms are quite fast; even for high dimensions around 1000 the distribution function can be estimated with moderate accuracy using only a few seconds of run time. Even log-densities around $-100$ can be estimated accurately and quickly. An implementation of all algorithms presented in this work is available in the \R\ package \texttt{nvmix} (version $\ge$ 0.0.4).
\end{abstract}
\minisec{Keywords}
Multivariate normal variance mixtures, distribution functions, densities, Student $t$, quasi-random number sequences.
\minisec{MSC2010}
62H99, 65C60%

\section{Introduction}
The multivariate normal
and (Student) $t$
distributions are among the most widely used multivariate distributions within
applications in statistics, finance, insurance and risk management. A simple way to create a much larger range of distributions than the multivariate normal, with different (joint and marginal) tail behavior
including tail dependence, is by
randomizing (mixing) the covariance matrix of a multivariate normal distribution.
This makes normal variance mixtures better suited,
for example, for
log-return distributions, while keeping many of the advantages of multivariate
normal distributions such as closedness with respect to linear combinations; see
\cite[Section~6.2]{mcneilfreyembrechts2015} for more details.

Formally speaking, we say that a random vector $\bX=(X_1,\dots,X_d)$ follows a \emph{normal variance
  mixture}, denoted $\bX\sim \NVM_d(\bmu,\Sigma,F_W)$, if, in
distribution,
\begin{align}
    \bX=\bmu+\sqrt{W}A\bm{Z},\label{eq:nvm:stoch:rep}
\end{align}
where $\bmu\in\IR^d$ denotes the \emph{location (vector)}, $\Sigma=AA\T$ for
$A\in\IR^{d\times k}$ is the \emph{scale (matrix)} (a covariance matrix), and
$W\sim F_W$ is a non-negative random variable independent of
$\bm{Z}\sim\N_k(\bm{0},I_k)$ (where $I_k\in\IR^{k\times k}$ denotes the identity
matrix), which we can think of as the mixing variable; see, for example, \cite[Section~6.2]{mcneilfreyembrechts2015}. Note
that $(\bX\,|\,W)\sim\N_d(\bmu,W\Sigma)$, hence the name of this class of
distributions. This implies that if $\E(\sqrt{W})<\infty$, then
$\E(\bX)=\bmu$, and if $\E(W)<\infty$, then
$\cov(\bX)=\E(W)\Sigma$ and $\corr(\bX)=P$ (the correlation matrix
corresponding to $\Sigma$). Furthermore, note that in the latter case with
$A=I_d$ (so when the components of $\bX$ are uncorrelated) the
components of $\bX$ are
independent if and only if $W$ is constant almost surely and thus $\bX$ is
multivariate normal; see \cite[Lemma~6.5]{mcneilfreyembrechts2015}.
The multivariate $t$ distribution is obtained by letting $W$ have an inverse-gamma distribution.
In what
follows we focus on the case $k=d$ in which $A$ is typically the Cholesky factor
computed from a given $\Sigma$; other decompositions of $\Sigma$ into $AA\T$ for
some $A\in\IR^{d\times d}$ can be obtained from the eigendecomposition or
singular-value decomposition.

Working with normal variance mixtures (as with any other multivariate
distribution) often involves four tasks: sampling, computing the joint
distribution function, computing the joint density function as well as parameter
estimation.
Sampling is straightforward via~\eqref{eq:nvm:stoch:rep} based on
the Cholesky factor $A$ of $\Sigma$.

In contrast,
evaluating multivariate distribution functions (such as the normal and the $t$) is a difficult, yet important problem that has gained much attention in the last couple of decades; see, for instance, \cite{genz1992}, \cite{hickernellhong1997}, \cite{genzbretz1999}, \cite{genzbretz2002}, \cite{genzbretz2009} as well as references therein for a discussion of the estimation of multivariate normal and $t$ probabilities and recent work in \cite{botevlecuyer2015} for the evaluation of truncated multivariate $t$ distributions.
To further illustrate how challenging this problem is, we note that  the \R\ package \texttt{mvtnorm} (one of the most
widely used packages according to reverse depends, see \cite{eddelbuettel2012})
and other \R\ packages do not even provide functionality for evaluating the
distribution function of the well-known multivariate $t$ distribution for
non-integer degrees of freedom $\nu>0$.

In this paper, we propose efficient algorithms for computing the joint distribution function and joint density function of a normal variance mixture, and also for estimating its parameters. The only requirement we have for the normal variance mixture is that we must have access to a (possibly numerical) procedure to evaluate the quantile function of $W$.
The assumption that such ``black-box'' procedure is available to evaluate the quantile is something we refer to as having a {\em computationally tractable} quantile function for $W$. Providing algorithms for the above tasks for a more general family of distributions than what currently exists in the literature is one of the main contributions of this work.

The algorithm we propose to efficiently evaluate the joint distribution function of a normal variance mixture (including the case when $\Sigma$ is singular)  is obtained by generalizing
methods by A. Genz and F. Bretz to evaluate the distribution function of the multivariate normal and $t$ distribution.
In particular, we generalize a variable reordering algorithm originally suggested by~\cite{gibsonglasbeyelston1994} and adapted by~\cite{genzbretz2002} which significantly reduces the variance of the integrand yielding fast convergence of our estimators. We also propose a different approach for using RQMC methods within the integration routine required to evaluate the joint distribution function. Our approach better leverages the improved convergence properties of these methods compared to Monte Carlo sampling. In addition, we explore the synergy between these methods and the variable reordering algorithm using the concept of Sobol' indices and effective dimension, thus providing new insight on why the reordering algorithm works so well. Sections \ref{sec:mcrqmc} and \ref{sec:pnvmix} respectively include the discussion of RQMC methods and the tasks of evaluating the joint distribution function.

Regarding the joint density function of $\bX$, when going from a simple case such as the multivariate normal to a general normal variance mixture, it can go from being available in closed form to requiring the numerical evaluation of an intractable one-dimensional integral. An example of the latter situation is when $W$ follows an inverse-Burr distribution. Since our goal is to provide algorithms that work for any normal variance mixture, an efficient algorithm to approximate the joint (log)-density function of $\bX$  is needed.
We tackle this
 by proposing in Section \ref{sec:dnvmix} an adaptive RQMC algorithm that mostly samples  in certain important subdomains of the range of the mixing variable to efficiently estimate the log-density of a multivariate normal variance mixture. Even log-densities around $-100$ can be estimated efficiently.

This flexible algorithm turns out to be a key ingredient for the task of parameter estimation, which we again propose in enough generality to handle any normal variance mixture, as explained in Section \ref{sec:fitnvmix}. More precisely,
we employ an ECME (``Expectation/Conditional Maximization Either'') algorithm, which is a likelihood-based fitting procedure developed in~\cite{liurubin1994}. This procedure requires repeated evaluations of the log-density function of $\bX$, which is one of the reasons why efficient algorithms for the latter are important when this density does not have a closed form.

An extensive numerical study for all proposed
algorithms is included in Section~\ref{sec:numericalexamples}. This section also includes a detailed investigation of why the reordering algorithm works well with RQMC methods, as well as a data analysis with real-world financial data.

All presented algorithms are available in our \R\ package
\texttt{nvmix} (in particular, via \texttt{rnvmix()}, \texttt{pnvmix()}, \texttt{dnvmix()} and \texttt{fitnvmix()}; see also \texttt{vignette(nvmix\_functionality)}) and the conducted simulations are reproducible with the demo
\texttt{numerical\_experiments}; see \cite{hoferthintzlemieux2020}.

To the best of our knowledge, none of the four  aforementioned tasks have
been discussed in the literature in such generality where the only requirement is to have a computationally tractable quantile
function for the mixing variable $W$. By specifying the latter,
methods developed in this paper (and the implementation in \texttt{nvmix}) can
be used to perform standard modeling tasks for multivariate normal variance
mixtures well beyond the case of a multivariate $t$ distribution. To demonstrate this,
a real financial data set is analyzed using an inverse-gamma, a Pareto and an inverse-Burr mixture at the end of Section~\ref{sec:numericalexamples}.

\section{Normal variance mixture distribution function and density}%
We assume that $\Sigma$ has full rank so that the density of $\bX\sim\NVM_d(\bmu, \Sigma, F_W)$ exists.
 Denote by $D^2(\bx;\bmu,\Sigma)=(\bx-\bmu)\T\Sigma^{-1}(\bx-\bmu)$ the (squared) Mahalanobis distance of $\bx\in\mathbb{R}^d$ from $\bmu$ with respect to (wrt) $\Sigma$. By conditioning
on $W$ and substituting $w=F_W^\i(u)$ (where
$F_W^\i(u)=\inf\{w\in[0,\infty):F_W(w)\ge u\}$, $u\in(0,1)$, denotes the
quantile function of $F_W$),
the density of $\bX$ can then be written as
\begin{align}
    f_{\bX}(\bx) &=\int_0^\infty f_{\bX|W}(\bx\,|\,w)\,\rd F_W(w)
    =\int_0^{\infty} \frac{1}{\sqrt{ (2\pi w)^d |\Sigma|}} \exp\left(-\frac{D^2(\bx;\bmu,\Sigma)}{2w}\right)\,\rd F_W(w)\label{eq:densityX:raw}\\
    &=\int_0^{1} \frac{1}{\sqrt{ (2\pi F_W^\i(u) )^d |\Sigma|}} \exp\left(-\frac{D^2(\bx;\bmu,\Sigma)}{2F_W^\i(u) }\right)\,\rd u.\label{eq:densityX}
\end{align}
Note that this representation holds for the case when $W$ is absolutely continuous, discrete or of mixed type.
In the former case, (\ref{eq:densityX:raw}) equals
\begin{align}\label{eq:densityX:fw}
    f_{\bX}(\bx) = \int_0^{\infty} \frac{1}{\sqrt{ (2\pi w)^d |\Sigma|}} \exp\left(-\frac{D^2(\bx;\bmu,\Sigma)}{2w}\right)f_W(w) \,\rd w,
\end{align}
where $f_W$ denotes the density of $W$.

Furthermore, note that $f_{\bX}(\bx)$ is decreasing in the Mahalanobis distance $D^2(\bx;\bmu,\Sigma)$. Thus
$$ f_{\bX}(\bx) \leq f_{\bX}(\bmu) = \frac{1}{\sqrt{(2\pi)^d |\Sigma|}} \E\left(\frac{1}{W^{d/2}}\right); \quad \bx\in\mathbb{R}^d$$
so that $ f_{\bX}(\bx)$ is bounded if and only if $\E(W^{-d/2}) < \infty$.

Let $F_{\bX}(\ba,\bb)$ denote the probability
that $\bX$ falls into the hyperrectangle spanned by the lower-left endpoint
$\ba$ and upper-right endpoint $\bb$, where $\ba,\bb\in\bar{\mathbb{R}}^d$ for $\bar{\mathbb{R}}=\mathbb{R} \cup\{-\infty, \infty\}$ and $\ba < \bb$
(interpreted componentwise), where we interpret non-finite components as the corresponding limits. Note that the joint distribution function
of $\bX$ is a special case of $F_{\bX}(\ba,\bb)$ since $F_{\bX}(\bx):=\P(\bX \leq \bx) = F_{\bX}(\ba,\bx)$ for $\ba =(-\infty,\dots,-\infty)$. In what follows we write $F(\ba,\bb)$ instead of $F_{\bX}(\ba,\bb)$
to simplify notation.
For computing $F(\ba,\bb)$ assume (potentially after adjusting $\ba,\bb$) that
$\bmu=\bm{0}$ and that $\Sigma$ has full rank (the singular case will be discussed in Section~\ref{sec:app:singular}).
By conditioning and the substitution $w=F_W^\i(u)$ we obtain that
\begin{align}
  F_{\bX}(\ba,\bb)&=\P(\ba < \bX \le \bb)= \P(\ba < \sqrt{W} A \bZ \le \bb) = \E\left(\P(\ba/\sqrt{W} < A\bZ \le\bb/\sqrt{W}\;|\, W)\right)\notag\\
  &=\E\left(\Phi_{\Sigma}(\ba/\sqrt{W}, \bb/\sqrt{W})\right) = \int_0^{\infty}\Phi_{\Sigma}(\ba/\sqrt{w}, \bb/\sqrt{w})\,\rd F_W(w)\notag\\
  &=\int_0^1\Phi_{\Sigma}\left(\ba/\sqrt{F_W^\i(u)}, \bb/\sqrt{F_W^\i(u)}\right)\,\rd u,\label{eq:cdfX}
\end{align}
where $\Phi_{\Sigma}(\ba,\bb)=\P(\ba<\bY\leq \bb)$ for $\bY\sim \N_d(\bm{0},\Sigma)$.

\section{Monte Carlo and (randomized) quasi-Monte Carlo methods}\label{sec:mcrqmc}
Quantities of interest in this paper, such as the distribution function of a normal variance mixture, are (after a suitable transformation) expressed as intractable integrals over the unit hypercube $(0,1)^d$ for some $d\in\IN$, i.e.,
\begin{align}\label{eq:mu}
\mu = \int_{(0,1)^d} g(\bu)\,\rd\bu,
\end{align}
where $g:(0,1)^d\rightarrow\IR$ is integrable. Monte Carlo (MC) methods approximate $\mu$ in~\eqref{eq:mu} by the arithmetic average $\hmuMCn = (1/n) \sum_{i=1}^n g(\bU_i)$ where $\bU_1,\dots,\bU_n\isim\U(0,1)^d$. An asymptotic $(1-\alpha)$-confidence interval (CI) can be approximated for sufficiently large $n$ by
\begin{align*}
\left[\hmuMCn - z_{1-\alpha/2} \hsigma_g/\sqrt{n},\ \hmuMCn + z_{1-\alpha/2} \hsigma_g/\sqrt{n}\right],
\end{align*}
where $z_{\alpha}=\Phi^{-1}(\alpha)$ and $\hsigma_g^2 = \widehat{\var}(g(\bU)=\sum_{i=1}^n(g(\bU_i)-\hmuMCn)^2/(n-1)$. One can choose $n$ so that the length of this CI does not exceed a pre-determined absolute error tolerance.

Replacing the (pseudo-random) evaluation points $\bU_1,\dots,\bU_n$ by a deterministic low-discrepancy point set which aims at filling the unit hypercube in a more homogeneous way, say $P_n=\{\bv_1,\dots,\bv_n\}\subset [0,1)^d$, leads to a \emph{quasi-Monte Carlo} (QMC) estimator for $\mu$. QMC methods often provide better estimators than classical MC methods, the deterministic nature of the points in $P_n$ however does not allow for simple error estimation via CIs as was done for the MC estimator $\hmuMCn$. To overcome this, one can randomize the point set $P_n$ in a way
such that the points in the resulting point set, say $\tilde{P}_n$, are uniformly distributed over
$(0,1)^d$ without losing the low-discrepancy of the point set overall. This leads to randomized QMC (RQMC) methods. In our algorithms, we use a digitally-shifted Sobol' sequence (\cite{sobol1967}) as implemented in the function \texttt{sobol(, randomize = "digital.shift")} of the \R\ package \texttt{qrng}; see \cite{hofertlemieux2019}. We remark that generating $\tilde{P}_n$ is slightly faster than generating $\bU_1,\dots,\bU_n\isim\U(0,1)^d$ using \R's default (pseudo-)random number generator, the Mersenne Twister.

Given $B$ independently randomized copies of $P_n$, say
$\tilde{P}_{n,b}=\{\bu_{1,b},\dots,\bu_{n,b}\}$ for $b=1,\dots,B$, one can construct
$B$ independent RQMC estimators of the form
\begin{align}\label{eq:hmurqmcbn}
  \hmurqmc_{b,n} = \frac{1}{n} \sum_{i=1}^n g(\bu_{i,b}), \quad b=1,\dots,B,
\end{align}
and combine them to the RQMC estimator
\begin{align}\label{eq:hmurqmcn}
  \hmurqmcn = \frac{1}{B} \sum_{b=1}^B \hmurqmc_{b,n}
\end{align}
of $\mu$. An approximate $(1-\alpha)$-CI for $\mu$ can be estimated as
\begin{align}\label{eq:CI:rqmc}
 \left[\hmurqmcn - z_{1-\alpha/2}\hsigma_{\hmurqmc}/\sqrt{B}, \hmurqmcn + z_{1-\alpha/2}\hsigma_{\hmurqmcn}/\sqrt{B}\right],
 \end{align}
where
\begin{align}\label{eq:sigma:rqmc}
\hsigma_{\hmurqmcn} = \sqrt{\frac{1}{B-1} \sum_{i=1}^B (\hmurqmc_{b,n}-\hmurqmcn)^2}.
\end{align}
One can compute $\hmurqmcn$ from~\eqref{eq:hmurqmcn} for some initial sample size $n$ (e.g., $n=2^7$) and iteratively increase the sample size of each $\hmurqmc_{b,n}$ in~\eqref{eq:hmurqmcbn} until the length of the CI in~\eqref{eq:CI:rqmc} satisfies a pre-specified error tolerance. In our implementations, we use $B=15$, an absolute default error tolerance $\varepsilon=0.001$ (which can be changed by the user) and $z_{1-\alpha/2}=3.5$ (so $\alpha\approx0.00047$). By using $\hmurqmcn$ as approximation for the true value of $\mu$, one can also consider relative errors instead of absolute errors.

Function evaluations from iterations that did not meet the tolerance can be recycled as
follows. Let $P_{n_1, n_2} = \{\bv_{n_1+1},\dots,\bv_{n_1+n_2}\}$ be the point
set consisting of the $n_2$ low-discrepancy points after skipping the first
$n_1$-many points. Furthermore, let
$\tilde{P}_{n_1, n_2, b} = \{\bu_{n_1+1,b},\dots,\bu_{n_1+n_2,b}\}$ be the
$b$th randomly shifted version of $P_{n_1,n_2}$ and let
\begin{align*}
  \hmurqmc_{b,n_1, n_2}= \frac{1}{n_2} \sum_{\bu\in\tilde{P}_{n_1, n_2, b}}g(\bu),\quad b=1,\dots,B.
\end{align*}
If $\hmurqmc_{n_1}$ does not meet the error tolerance, an estimator based
on $n_1+n_2$ points can be calculated using only $n_2$ additional function
evaluations based on
\begin{align*}
\hmurqmc_{b,0, n_2}= \frac{n_1 \times \hmurqmc_{b,0,n_1} + n_2 \times \hmurqmc_{b,n_1,n_2}}{n_1+n_2},\quad b=1,\dots,B.
\end{align*}

In iteration $i$ this update is being done with $n_1=in_0$ and $n_2=n_1+n_0$ in Step~\ref{step:update:rqmc} of our Algorithm~\ref{alg:RQMC:mu} to estimate $\mu$ from~\eqref{eq:mu}. That is, we start with initial sample size $n_0$ and add another $n_0$ points in each iteration. We highlight that this update can be easily implemented for a Sobol' sequence, as one can generate $P_{n_1,n_2}$ efficiently without having to generate $P_{0, n_2}$; in \R, this can be achieved by calling \texttt{sobol(, skip = n1)}. We do not lose any low-discrepancy properties of the randomized Sobol' sequences as the resulting estimator is mathematically equivalent to $\hmurqmc_{n^*} = (1/B) \sum_{b=1}^B \hmurqmc_{b,0,n^*}$ where $n^*$ is the total number of function evaluations in each randomization. We therefore leverage convergence properties in $n$ of Sobol' sequence based estimators.
It is important to point out that the reason why we can add points in this way without discarding previous function evaluations is because the Sobol' sequence is extensible in $n$. That is, it is constructed as a sequence in such a way that the first $n$ points can be used as a low-discrepancy point set $P_n$ for any $n$, with additional uniformity properties when $n$ is a  power of 2 (or a multiple of a power of 2).

The update in our algorithm is conceptually different from updates in RQMC methods suggested in our references: For instance, the RQMC algorithm proposed in \cite{genzbretz2002} to estimate the distribution function of a multivariate $t$ distribution, therein referred to as QRSVN algorithm, is based on a randomized Korobov rule (which belong to the wider class of
lattice rules; see \cite{keast1973} and \cite{cranleypatterson1976}). The QRSVN algorithm also iteratively evaluates the integrand at low-discrepancy points until the estimated error is small enough; however, it does not move along the same sequence of low-discrepancy points from one iteration to another. In iteration $i$, their method computes an estimator based on a lattice of size $p_i$ (a prime), and estimators from different iterations are combined as a variance-weighted average. Ultimately, the QRSVN algorithm outputs a weighted average of $B\cdot i^*$ different RQMC estimators based on different sample sizes (where $i^*$ denotes the number of iterations needed until termination), whereas our algorithm outputs the average of $B$ digitally-shifted RQMC estimators based on the first $n^*$ points of a Sobol' sequence. Hence, our methods leverage properties of the Sobol' sequence with growing $n$ rather than combining more and more RQMC estimators of different sample sizes.
Our proposed approach is thus superior because the variance of RQMC estimators can be shown to be in $O(n^{-\delta})$ with $\delta >1$ (and the smoother $f$ is, the larger $\delta$ is). Hence for a given fixed computing budget of $Bn$ function evaluations that must be split between $B$ and the size $n$ for the point set $P_n$, it is best to try to take $B$ just large enough so that we get a reasonable variance estimate, and then set $n$ as large as possible in order to further reduce the variance thanks to its $O(n^{-\delta})$ behavior: this is precisely what our approach does.
Numerical results in Section \ref{subsec:pnvmixexamples} illustrate how this leads to improved efficiency compared to the QRSVN algorithm.

Finally, our way of updating merely requires $B n_0$ additional function evaluations in each iteration, rather than $B p_{i+1}$. This typically leads to a smaller run-time, as only as many function evaluations as needed are computed.

\begin{algorithm}[RQMC Algorithm to estimate $\mu=\int_{(0,1)^d} g(\bu)\;\rd \bu$.]\label{alg:RQMC:mu}
    Given $\eps$, $B$, $n_0$, $\imax$, estimate $\mu=\int_{(0,1)^d} g(\bu)\;\rd \bu$ via:
    \begin{enumerate}
        \item Set $n=n_0$, $i=1$, and compute $\hmurqmc_{b,n}=\hmurqmc_{b,0,n_0}$ for $b=1,\dots,B$ and $\hmurqmc_n$ from~\eqref{eq:hmurqmcbn} and~\eqref{eq:hmurqmcn}.
        \item Set $\hat{\eps}=3.5 \hsigma_{\hmurqmc_n}$ with $\hsigma_{\hmurqmc_n}$ from~\eqref{eq:sigma:rqmc}.
        \item While $\hat\eps>\eps$ and $i\leq \imax$ do:
        \begin{enumerate}
            \item\label{step:update:rqmc} Set $n=n+n_0$, compute $\hmurqmc_{b,in_0,(i+1)n_0}$, $b=1,\dots,B$ and set $\hmurqmc_{b,n}=(i\hmurqmc_{b,n}+\hmurqmc_{b,in_0,(i+1)n_0})/(i+1)$.
            \item Update $\hmurqmc_n=(1/B)\sum_{i=1}^B\hmurqmc_{b,n}$ and update $\hat{\eps}=3.5 \hsigma_{\hmurqmc_n}$ with $\hsigma_{\hmurqmc_n}$ from~\eqref{eq:sigma:rqmc}.
            \item Set $i=i+1$.
        \end{enumerate}
        \item Return $\hmurqmc_n$.
    \end{enumerate}
\end{algorithm}

\label{page:properlog}
Sometimes it is necessary to estimate $\log \mu$ rather than $\mu$; in particular, when $\mu$ is small. For instance, if $\mu = f(\bx)$ where $f(\bx)$ is the density function of $\bX\sim\NVM_d(\bmu, \Sigma, F_W)$ evaluated at $\bx\in\IR^d$, interest may lie in $\log(\mu)=\log f(\bx)$ as this quantity is needed to compute the log-likelihood of a random sample (which then may be optimized over some parameter space).
When $\mu$ is small, using $\log\mu \approx \log(\hmurqmcn)$ directly should be avoided. One should instead compute a numerically more robust estimator for $\log \mu$, a \emph{proper logarithm}. To this end, define the function $\LSE$ (for Logarithmic Sum of Exponentials) as

\begin{align*}
 \LSE(c_1,\dots,c_n) = \log\left(\sum_{i=1}^n \exp(c_i)\right) = c_{\text{max}} + \log\left(\sum_{i=1}^n \exp(c_i - c_{\text{max}})\right),
 \end{align*}

\noindent
where $c_1,\dots,c_n\in\IR$ and $c_{\max} = \max\{c_1,\dots,c_n\}$. The right-hand side of this equation is numerically more stable than the left-hand side as the the sum inside the logarithm is bounded between 1 and $n$.

Let $c_{i,b}=\log g(u_{i,b})$ for $i=1,\dots,n$ and $b=1,\dots,B$. An estimator numerically superior (but mathematically equivalent) to $\log(\hmurqmcn)$ is given by
\begin{align}\label{eq:hmurqmcn:log}
\hmurqmc_{n, \log} = -\log(B) + \LSE(\hmurqmc_{1, n, \log} ,\dots,\hmurqmc_{B, n, \log}),
\end{align}
where
\begin{align}\label{eq:hmurqmcbn:log}
\hmurqmc_{b, n, \log} = -\log(n) + \LSE(c_{1,b},\dots, c_{n,b}),\quad b=1,\dots,B.
\end{align}
The standard deviation of $\hmurqmc_{n, \log} $ is estimated in the usual way by computing the sample standard deviation of $\hmurqmc_{1, n, \log} ,\dots,\hmurqmc_{B, n, \log}$ so that, as before, the integration error can be estimated from the length of the CI in~\eqref{eq:CI:rqmc}. A summary of the procedure to estimate $\log\mu$ with a proper logarithm via RQMC is given in Algorithm~\ref{alg:RQMC:logmu} in the appendix.
Note that despite the fact that the problem under study here is a one-dimensional integral, we refer to our algorithm as being in the RQMC family. We do so because although the distinctive features of RQMC mostly have to do with how they design low-discrepancy point sets in dimension larger than 1, another distinctive feature they have is to make use of low-discrepancy sequences that are extensible in $n$, which is precisely what we are exploiting for this algorithm.

For more information about RQMC methods and their applications in the financial literature, see, e.g., \cite{niederreiter1992}, \cite{lemieux2009} and \cite{glasserman2013}.

\section{Computing the distribution function}\label{sec:pnvmix}

As mentioned in the introduction,
throughout this paper we assume that the quantile function
$F_W^\i$ of $W$ is computationally tractable (possibly through an
approximation). Assume furthermore that the scale matrix $\Sigma$ has full rank;
the evaluation of singular normal variance mixtures is discussed in~\ref{sec:app:singular}.

One might be tempted to sample $U_i\isim\U(0,1)$, $i=1,\dots,n$, and then
approximate the integral in (\ref{eq:cdfX}) by the conditional Monte Carlo estimator
\begin{align*}
  F(\ba,\bb) \approx \hat{\mu}_{F}^{\text{CMC}} = \frac{1}{n} \sum_{i=1}^n \Phi_{\Sigma}\left(\ba/\sqrt{F_W^\i(U_i)}, \bb/\sqrt{F_W^\i(U_i)}\right).
\end{align*}

However, $\Phi_{\Sigma}$ itself is a $d$-dimensional integral typically evaluated by
RQMC methods, so this approach
would be time-consuming. Hence, the first step should be to approximate
$\Phi_{\Sigma}$. To this end, we follow \cite{genz1992} and start by expressing $\Phi_{\Sigma}$ (and then $F(\ba,\bb)$) as
integrals over the unit hypercube. In the second part of this section, we derive an efficient RQMC
algorithm to approximate $F(\ba,\bb)$ based on Algorithm \ref{alg:RQMC:mu}. In particular, it
details how a significant variance reduction (and hence decrease in run time)
can be achieved through a variable reordering following an approach originally
suggested by \cite{gibsonglasbeyelston1994} for multivariate normal
probabilities and later adapted by \cite{genzbretz2002} to work for multivariate
$t$ probabilities.

The novelty of our approach for this problem is three-fold: first, our algorithm applies to any normal variance mixture; second, it uses RQMC methods in a way that better leverages their convergence properties, compared to previous work done for the multivariate normal and $t$ distributions, and third, we include a detailed analysis (with our numerical results, in Section \ref{subsubsec:VarReorder}) of why the reordering algorithm works well with RQMC methods.

\subsection{Reformulation of the integral}\label{sec:reformulationFab}
We now address $\Phi_{\Sigma}$. Let $C=(C_{ij})_{i,j=1}^d$  be the Cholesky factor of $\Sigma$, i.e., a lower triangular matrix satisfying $C C\T = \Sigma$. Denote by $C_k\T$ the $k$th row of $C$ for $k=1,\dots,d$. \cite{genz1992} (see also \cite{genzbretz1999}, \cite{genzbretz2002} and
\cite{genzbretz2009}) uses a series of transformations that rely on the lower triangular structure of $C$ to produce a separation of variables as follows:
\begin{align}\label{eq:normalcdf}
    \Phi_{\Sigma}(\ba,\bb)&= \int_{a_1}^{b_1}\dots \int_{a_d}^{b_d} \frac{1}{\sqrt{(2\pi)^d |\Sigma|}}\exp\left(-\frac{\bx\T \Sigma^{-1} \bx}{2}\right)\rd \bx \nonumber \\
    &= (\he_1-\hd_1) \int_0^1 (\he_2-\hd_2)\ \dots \int_0^1 (\he_d-\hd_d)\, \rd u_{d-1}\dots \rd u_1,
\end{align}
where the $\hd_i$ and $\he_i$ are recursively defined via
\begin{align*}
\he_1 = \Phi\left(\frac{b_1}{C_{11}}\right), \he_i = \he_i(u_1,\dots,u_{i-1})=\Phi\left(\frac{ b_i - \sum_{j=1}^{i-1} C_{ij}\Phi^{-1}\left(\hd_j + u_j(\he_j-\hd_j)\right)}{C_{ii}}\right),
\end{align*}
and $\hd_i$ is $\he_i$ with $b_i$ replaced by $a_i$ for $i=1,\dots,d$. Note that the final integral in~\eqref{eq:normalcdf} is $(d-1)$-dimensional.

With this at hand, we can write~\eqref{eq:cdfX} as
\begin{align}\label{eq:cdfXtransformed2}
 F(\ba,\bb) =\int\limits_{(0,1)^d} g(\bu)\,\rd \bu =  \int\limits_0^1 g_1(u_0) \int\limits_0^1 g_2 (u_0,u_1) \dots \int\limits_0^1 g_d(u_0,\dots,u_{d-1})\, \rd u_{d-1}\dots \rd u_0,
 \end{align}
 where
\begin{align}\label{eq:gi}
g(\bu)=\prod_{i=1}^d g_i(u_0,\dots, u_{i-1}),\quad g_i(u_0,\dots,u_{i-1})= e_i - d_i,\quad i=1,\dots,d,
\end{align}
for $\bu=(u_0,u_1,\dots,u_{d-1})\in(0,1)^d$. The $e_i$ are recursively defined by
\begin{align}\label{eq:ei}
e_1 &= e_1(u_0) = \Phi\left(\frac{b_1}{C_{11}\sqrt{F_{W}^\i(u_0)}}\right),\nonumber \\
e_i &= e_i(u_0, \dots, u_{i-1}) = \Phi\left( \frac{1}{C_{ii}}\left( \frac{b_i}{\sqrt{F_{W}^\i(u_0)}}  - \sum_{j=1}^{i-1} C_{ij}\Phi^{-1}(d_j+u_{j}(e_j-d_j))\right)\right),
\end{align}
for $i=2,\dots,d$ and the $d_i$ are $e_i$ with $b_i$ replaced by $a_i$ for $i=1,\dots,d$. We remark that there is a typo (wrong bracket) in the corresponding formula for the special case of a multivariate $t$ distribution in \cite[p. 958]{genzbretz2002}.

Summarizing, the original $(d+1)$ dimensional integral is reduced to
 $$ F(\ba,\bb) = \int_{(0,1)^d} g(\bu)\;\rd\bu,$$
with the function $g$ defined in~\eqref{eq:gi} so that RQMC methods from Section~\ref{sec:mcrqmc} could be applied directly to the problem in this form to estimate $F(\ba,\bb)$. As pointed out in \cite{genzbretz2009}, the transformations undertaken in this section to produce a separation of variables essentially describe a Rosenblatt transform; see \cite{rosenblatt1952}.

\subsection{Variable reordering and RQMC estimation}\label{subsec:precond}
Inspecting~\eqref{eq:cdfXtransformed2} and~\eqref{eq:ei}, we see that the sampled component $u_j$ of $\bu$ in the $j$th integral affects the ranges of all $g_k$
with $k>j$. Observe that permuting
the order in $\ba$, $\bb$ and $\Sigma$ does not affect the value of
$F(\ba,\bb)$ as long as the same permutation is applied to $\ba$, $\bb$
and to both the rows and columns of $\Sigma$. It therefore seems to be a fruitful approach to
choose a permutation of $\ba$, $\bb$ and $\Sigma$ such that $g_2$ has, on average, the smallest range; $g_3$ the second smallest, and so on. This has been observed in
\cite{gibsonglasbeyelston1994} in the context of calculating multivariate normal
probabilities and has been adapted by \cite{genzbretz2002} to handle multivariate $t$
integrals. As in the latter reference, one can sort the integration limits a
priori according to their expected length of integration limits. This is more complicated
than just ordering $\ba$, $\bb$ and $\Sigma$ according to the lengths $b_{i}-a_{i}$ (assuming all of them are finite) as the latter does not take into account the dependence of the components in $\bX$. We generalize the Gibson, Glasbey and Elston method for reordering according to expected ranges to work for normal variance mixture distribution functions in Algorithm~\ref{alg:precond} in the appendix.

From a simulation point of view, the particular value of $u_1$ will
affect the ranges of all the remaining $d - 2$ integrals. Indeed, each input $\bu=(u_{0},\dots,u_{d-1})\sim\U(0,1)^{d}$ is transformed to a product of conditional probabilities: The first component, $u_{0}$, is used to sample from the mixing variable via inversion; $g_{1}(u_{0})$ is then the conditional probability of the first component of the random vector $\bX$  falling into $(a_{1}, b_{1})$ given that $W = F_W^\i(u_0)$, that is $g_1(u_0)=\P(X_1 \in (a_{1}, b_{1}) \mid W = F_W^\i(u_0))$.
Next, $u_{1}$ is transformed to $y_{1} = \Phi^{-1}(d_1+u_{1}(e_1-d_1))$, which is a realization of the random variable $\left(X_{1}\mid X_{1} \in (a_{1}, b_{1}), W = F_W^\i(u_0)\right)$. Then, $g_{2}(u_{0}, u_{1}) = \P(X_{2}\in(a_{2}, b_{2}) \mid X_{1} = y_{1}, W = F_W^\i(u_0))$ and so on and so forth. As we are conditioning on events of the form $\{X_{1} = y_{1},\dots, X_{l} = y_l, W = F_W^\i(u_0)\}$ for all subsequent probabilities, this also explains why variable reordering can help decrease the variance: It is designed in a way so that $X_1$ has smallest (expected) range, $X_2$ second smallest and so on. In the explanation above, if $b_{1}-a_{1}$ is small, there is only little variability in $y_{1}$ so that $g(u_{0}, u_{1})$ should be close to $\P(X_{2}\in(a_{2}, b_{2}) \mid X_{1} \in(a_{1}, b_{1}), W = F_W^\i(u_0))$. We point out that if $F_W^\i(u)$ is a non-zero constant for all $u\in(0,1)$ (corresponding to $\bX$ being multivariate normal), this is the original derivation in \cite{gibsonglasbeyelston1994} who independently developed a Monte Carlo procedure to approximate multivariate normal probabilities similar to \cite{genz1992}.

Algorithm~\ref{alg:precond} is a greedy procedure that only reorders $\ba,\bb$, $\Sigma$ (and updates the Cholesky factor $C$ accordingly). Changing the order in $\ba$, $\bb$ and $\Sigma$ does not introduce any bias so that one can use a rather crude approximation for $\mu_{\sqrt{W}}$ for $\E(\sqrt{W})$ if the true mean is not known. Note also that variable reordering needs to be performed only once before applying RQMC to the integrand $g$ in~\eqref{eq:gi} so that the cost of reordering is low compared to the overall cost of evaluating $F(\ba, \bb)$.

Our method to estimate $F(\ba,\bb)$ is summarized in Algorithm~\ref{alg:RQMC:F:a:b}.

\begin{algorithm}\label{alg:RQMC:F:a:b}
Given $\ba,\bb,\Sigma$, $\eps$, $B$, $n_0$, $\imax$, estimate $F(\ba,\bb)$ as follows:
\begin{enumerate}
\item Apply the reordering Algorithm~\ref{alg:precond} to the inputs $\ba,\bb,\Sigma$.
\item Apply Algorithm~\ref{alg:RQMC:mu} on the integrand $(g(\bu)+g(\bone-\bu))/2$ with $g$ from~\eqref{eq:gi} and reordered inputs.
\end{enumerate}
\end{algorithm}

In Section~\ref{subsec:pnvmixexamples} it is shown through a simulation study that this (rather cheap) variable reordering can yield a great variance reduction for the RQMC algorithm,
Algorithm~\ref{alg:RQMC:F:a:b}. A detailed study as to why this works so well is included in Section \ref{subsubsec:VarReorder}.

Note that parallelization of our methods, i.e., estimation of $F(\ba_i, \bb_i)$, $i=1,\dots,n$, simultaneously is difficult for two reasons: Reordering needs to be performed for each input $\ba_i,\bb_i$ separately so that Algorithm~\ref{alg:precond} needs to be called $n$ times. Furthermore, the structure of the integrand $g$ from~\eqref{eq:gi} (see also~\eqref{eq:ei}) does not allow for an efficient implementation of common random numbers as all quantile evaluations $\Phi^{-1}(\cdot)$ depend on the limits $\ba$, $\bb$ so that they cannot be recycled.

\section{Computing the (logarithmic) density}\label{sec:dnvmix}

We now turn to the task of computing the (logarithmic) density function of a normal
variance mixture.
Let us first point out that the main reason why we need to be able to evaluate the density function is for the fitting procedure, which is  likelihood-based and is explained in detail in Section \ref{sec:fitnvmix}. Now, since our goal is to be able to cover all normal variance mixtures, we cannot assume that the density function of $\bX$ is available in closed form. Indeed, a closed form for $f_{\bX}(\bx)$ exists in some cases (e.g., when $W$ is an inverse-gamma or Pareto), but not in all cases (e.g., when $W$ follows an inverse-Burr distribution, a model actually used with success in Section \ref{subsec:pnvmixexamples}). For those latter cases, an efficient approximation is needed, as there is likely to be  a repeated need for evaluating the density (or log-density) within the fitting procedure. This also means that fitting algorithms proposed for the multivariate $t$ cannot be directly adapted for the general normal variance mixture case, as they would not include functionalities able to deal with a density that does not exist in closed form. Below we propose an adaptive RQMC algorithm to deal with those cases, which is based on the ideas presented in Section \ref{sec:mcrqmc}.

From~\eqref{eq:densityX} it follows that computing the density requires the evaluation of the univariate integral
$ \mu:=f_{\bX}(\bx) = \int_0^1 h(u)\,\rd u, $
where
\begin{align}\label{eq:integrandh}
 h(u) = \frac{1}{\sqrt{ (2\pi F_W^\i(u) )^d |\Sigma|}} \exp\left(-\frac{D^2(\bx;\bmu,\Sigma)}{2F_W^\i(u) }\right),\quad u\in(0,1).
 \end{align}
To simplify notation, we write $f(\bx)$ instead of $f_{\bX}(\bx)$ whenever confusion is not possible.

For likelihood-based methods one should compute the logarithmic density (or
\emph{log-density}) rather than the density. Since $f(\bx)$ is expressed as a univariate integral over $(0,1)$, Algorithm~\ref{alg:RQMC:logmu}, that is, RQMC methods combined with a proper logarithm as described at the end of Section~\ref{sec:mcrqmc} on Page~\pageref{page:properlog}, can be applied directly to estimate $\log(\mu)=\log f(\bx)$ via RQMC. In fact, given inputs $\bx_1,\dots,\bx_N$, the log-densities $\log f(\bx_1),\dots,\log f(\bx_N)$ can be estimated simultaneously by using the same realizations of $W$, i.e., using the same $F_W^\i(u_{i,b})$ for all inputs $\bx_k$, $k=1,\dots,N$, until the precision is reached for all inputs. This procedure, i.e. estimating $\log \mu$ directly based on~\eqref{eq:integrandh} via RQMC, will be referred to as the \emph{crude procedure}.

\begin{figure}[t]
\minipage{0.32\textwidth}
  \includegraphics[width=\linewidth]{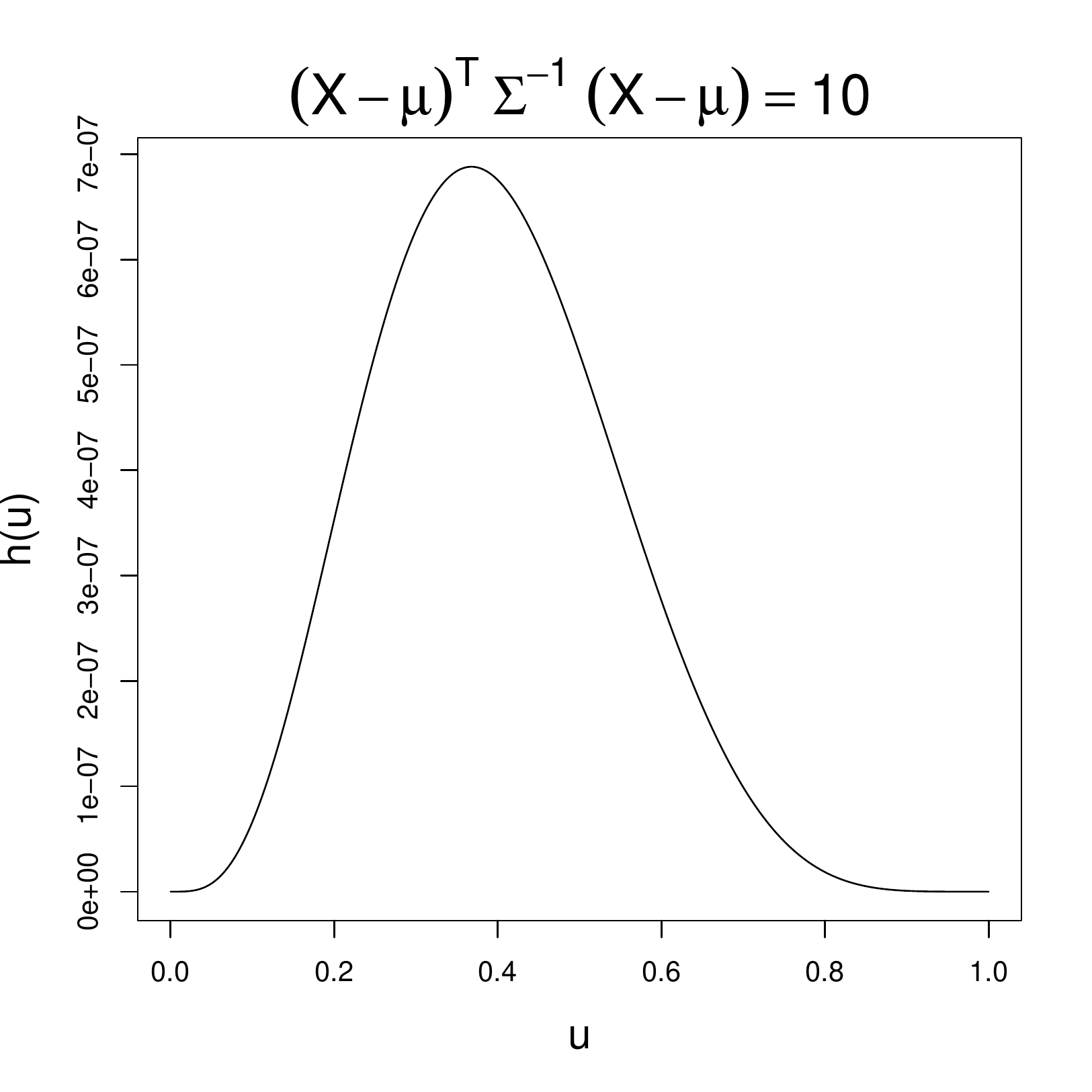}
\endminipage\hfill
\minipage{0.32\textwidth}
  \includegraphics[width=\linewidth]{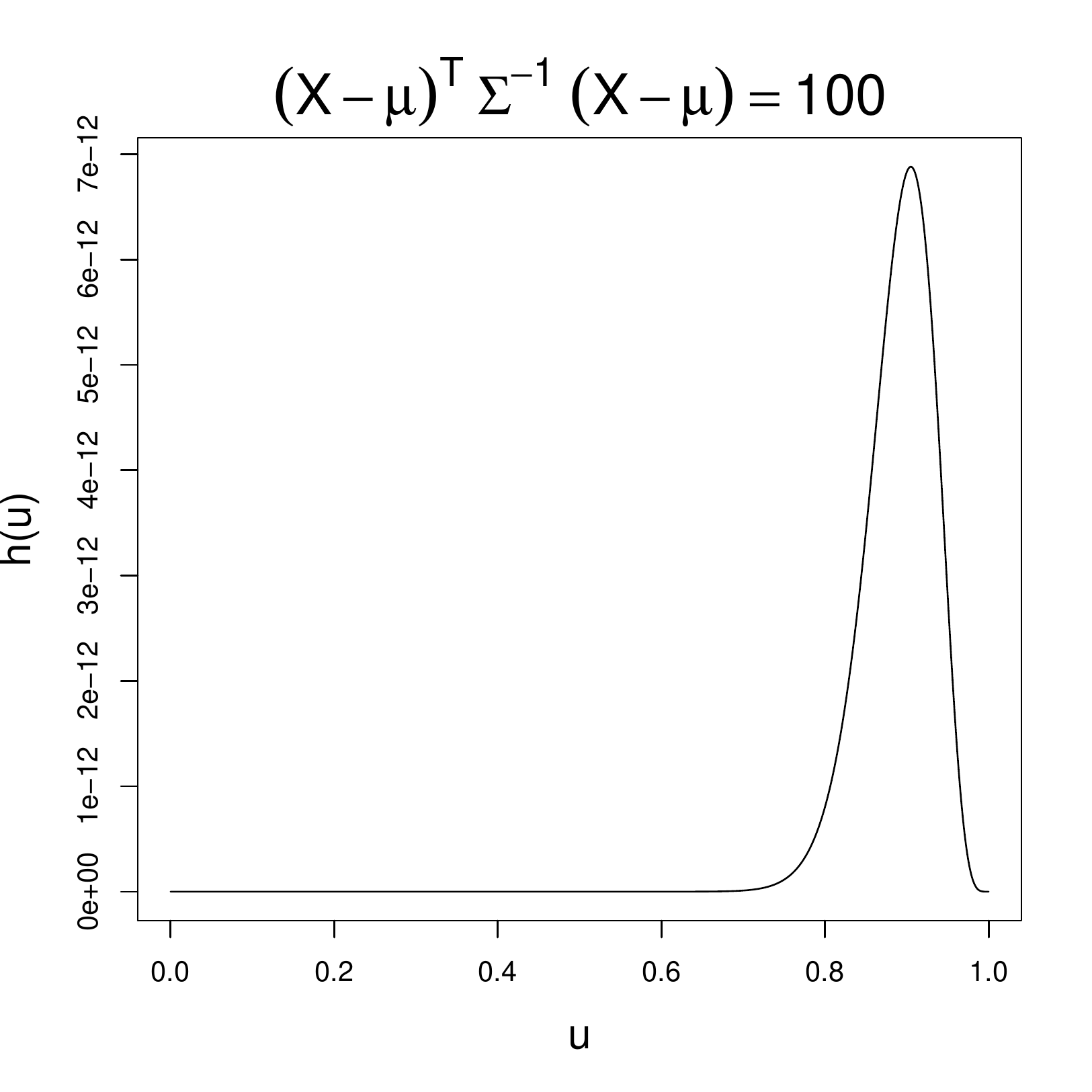}
\endminipage\hfill
\minipage{0.32\textwidth}
  \includegraphics[width=\linewidth]{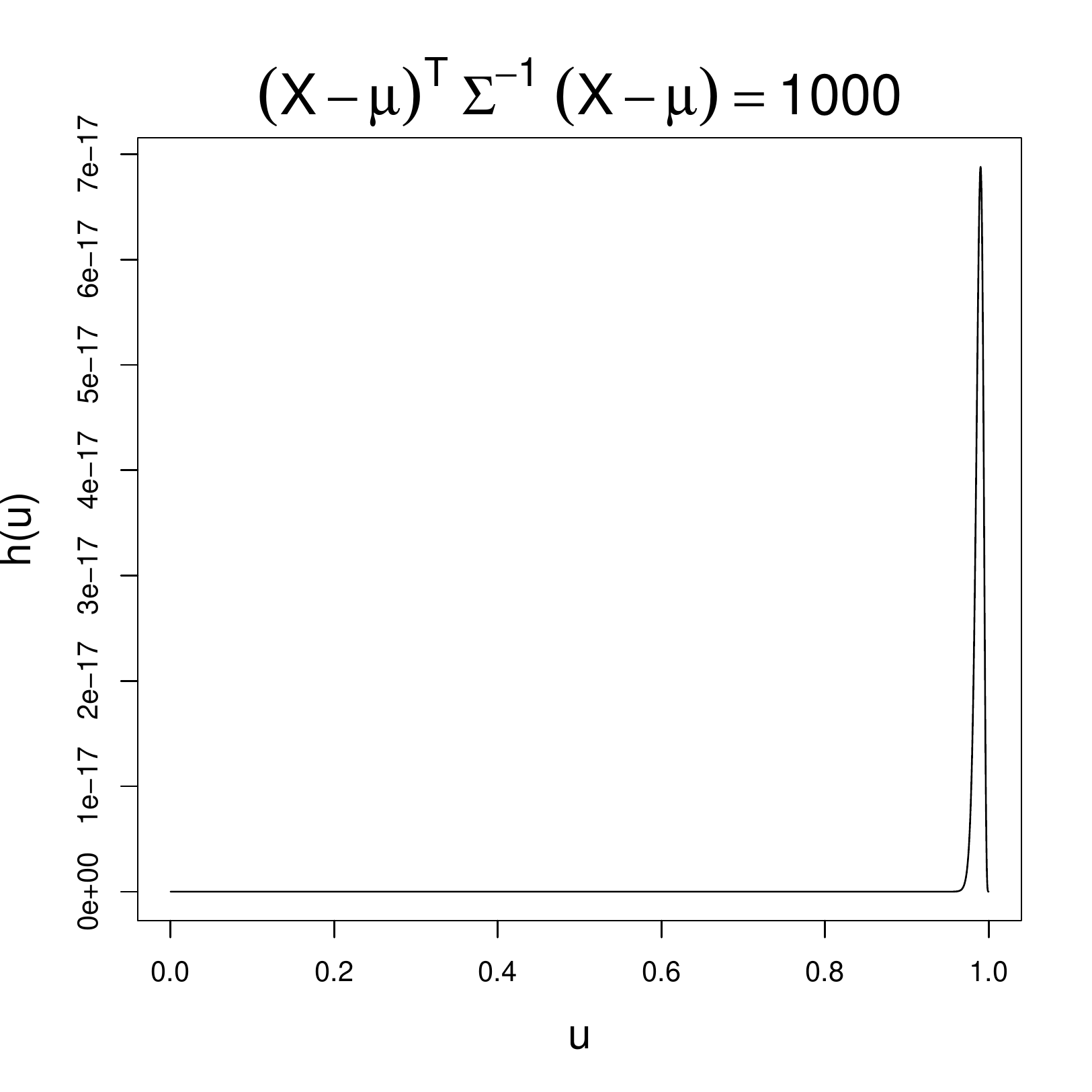}
\endminipage\hfill
\caption{Integrand $h$ for a 10-dimensional $t$ distribution with 2 degrees of freedom.}
\label{fig:density_integrand}
\end{figure}

It turns out that the crude procedure works sufficiently well for inputs $\bx$ with small to moderate Mahalanobis distances, but deteriorates for larger Mahalanobis distances. The reason is that the overall shape of the integrand $h$ is heavily influenced by $D^2(\bx; \bmu, \Sigma)$ and for large values, most of the mass is concentrated in a small domain of $(0,1)$. This is illustrated in Figure~\ref{fig:density_integrand} where the integrand $h(u)$ is plotted against $u$ in the special case where $\bX$ follows a multivariate $t$ distribution in dimension 10 with 2 degrees of freedom. For instance, in the right-most plot, most of the mass is concentrated near 1. It thus seems to be a fruitful approach to tailor the integration routine in a way so that it samples mostly in this relevant domain around the maximum, giving rise to an \emph{adaptive algorithm}. To this end, we summarize some properties of the function $h$ in the following lemma which can be shown using elementary calculus.

\begin{lemma}\label{lemma:propertiesh}
Let $W$ have a continuous distribution supported on the whole positive real line. Then, the function $h$ from Equation~\eqref{eq:integrandh} is continuous on $(0,1)$, satisfies $h(0)=h(1)=0$ and $h(u)>0$ for $u\in(0,1)$. Furthermore, the maximum
value of $h$ on $(0,1)$, i.e., $h_{\max} = \max\{h(u): u\in [0,1]\}$ is attained in the interior of $(0,1)$. The maximum is attained at
\begin{align}\label{eq:hmax}
u^* = F_W\left( \frac{D^2(\bx; \bmu, \Sigma)}{d}\right) \text{ with } h(u^*) = \left(\frac{2\pi|\Sigma|^{1/d}\cdot D^2(\bx; \bmu, \Sigma) }{d}\right)^{-d/2} \exp\left(-\frac{d}{2}\right)
\end{align}
so that $h_{\max}$ is independent of the distribution of $W$. Finally, $h$ is strictly increasing on $(0,u^*)$ and strictly decreasing on $(u^*, 1)$.
\end{lemma}

Equation~\eqref{eq:hmax} is crucial for the adaptive algorithm we propose: The value $h_{\max}$, i.e., the height of the peak of the integrand $h$, is independent of the distribution of $W$ as long as $W$ is continuous and supported on the whole positive real line. If $W$ is continuous but has bounded support, $h_{\max}$ may need to be replaced by $h(0)$ or $h(1)$. If $W$ is discrete, the problem becomes trivial as an analytical formula for the density is available in this case.

The idea is now to apply RQMC to a relevant region around $u^*$ from~\eqref{eq:hmax}, which can be done as follows: Given a threshold $\eps_{\text{th}}$ with $0 < \eps_{\text{th}} \ll h_{\max}$, the structure of the integrand $h$ guarantees the existence of $u_l$ and $u_r$ ($l$ for ``left'' and $r$ for ``right'') with $0< u_l < u^* < u_r < 1$ so that $h(u) > \eps_{\text{th}}$ if and only if $u \in (u_l, u_r)$. For instance, take
\begin{align}\label{eq:epsth}
\eps_{\text{th}} = 10^{\log(h_{\max})/\log(10) - k_{\text{th}}}
\end{align}
with $k_{\text{th}} = 10$ so that $\eps_{\text{th}}$ is 10 orders smaller than $h_{\max}$. RQMC can then be used in the region $(u_l, u_r)$ by replacing every number $v\in(0,1)$ by $v' = u_l + (u_r - u_l) v\in(u_l, u_r)$ yielding an estimate for $\log \int_{u_l}^{u_r} h(u)\;\rd u$. For the remaining regions $(0, u_l)$ and $(u_r, 1)$ we suggest using a crude trapezoidal rule: If $\eps_{\text{th}} \ll h_{\max}$ those regions do not significantly contribute to the overall integral anyway, so a rather cheap and quick procedure is recommended here.

It remains to discuss how the numbers $u_l, u^*, u_r$ can be computed. Recall that the only information available about $W$ is its quantile function $F_W^\i$ in form of a ``black box'' so that $u^*$ from~\eqref{eq:hmax} cannot be computed directly. We suggest using a bisection algorithm to solve the equivalent equation $F_W^\i(u) = \bx\T \Sigma^{-1} \bx / d$. Starting values can be found using a small number of pilot runs. Similarly, there is no direct formula for $u_l$ and $u_r$. While those can be expressed using Lambert's $W$ function, the lack of information about $W$ does not allow a direct computation. A bisection can be used here as well. Clearly, all pilot runs and all quantile evaluations performed in the bisections should be stored so that those expensive evaluations can be re-used.

It is clear from Figure~\ref{fig:density_integrand} that the shape of the integrand heavily depends on $\bx$ through its Mahalanobis distance, and this holds true for $u_l, u^*, u_r$ as well. As such, the adaptive procedure just described does not allow for simultaneous estimation of $\log f(\bx_1),\dots,\log f(\bx_N)$ directly, as the regions to which RQMC is applied differ from one input to another one. In order to reduce run time, we suggest using the crude procedure on all inputs $\bx_1,\dots,\bx_N$ with a small number of iterations (say, $\imax = 4$) first and use the adaptive procedure only for those inputs $\bx_j$ whose error estimates did not reach the tolerance. The advantage is that only little run time is spent on estimating ``easy'' integrals. Furthermore,  if $\imax=4$, $B=15$  and the initial sample size is $n_0=128$, such pilot run gives 7680 pairs $(u, F_W^\i(u))$. These can be used to determine starting values for the bisections to find $u_l$, $u_r$ and $u^*$ and they can also be used to estimate the integral in the regions $(0, u_l)$ and $(u_r, 1)$ using a trapezoidal rules with non-equidistant knots. The following algorithm summarizes our procedure, which is implemented in the \R\ function \texttt{dnvmix(, log = TRUE)} of the \R\ package \texttt{nvmix}.

\begin{algorithm}[Adaptive RQMC Algorithm to estimate $\log f(\bx_1),\dots,\log f(\bx_N)$]\label{alg:RQMC:f:x:adaptive}
    Given $\bx_1,\dots,\bx_N$ , $\Sigma$, $\eps$, $\eps_{\text{bisec}}$, $B$, $i_{\max}$, $k_{\text{th}}$, estimate $\log f(\bx_l)$, $l=1,\dots,N$, via:
    \begin{enumerate}
        \item Apply Algorithm~\ref{alg:RQMC:logmu} with at most $\imax$ iterations on all inputs $\bx_l$, $l=1,\dots,N$.  Store all uniforms and corresponding quantiles $F_W^\i(\cdot)$ in a list, say $\mathcal{L}$.
        \item If all estimates $\mu_{\log f(\bx_l)}^{\text{RQMC}}$, $l=1,\dots,N$ meet the error tolerance $\eps$, go to Step~\ref{return}.\\
        If not, we can assume wlog (after reordering) that $\bx_s$, $s=1,\dots,N'$ with $1\leq N' \leq N$ are the inputs whose error estimates did not meet the error tolerance.
        \item For each remaining input $\bx_s$, $s=1,\dots,N'$, do the following:
        \begin{enumerate}
            \item Determine $h_{\max}$ using~\eqref{eq:hmax} and $\eps_{\text{th}}$ using~\eqref{eq:epsth}.
            \item Find maximal $u^{*,\text{l}}$ and minimal $u^{*,\text{r}}$ in the list $\mathcal{L}$ so that $F_W^\i(u^{*, \text{l}}) \leq \bx_s\T \Sigma^{-1} \bx_s / d \leq F_W^\i(u^{*, \text{r}})$ (which implies $u^{*, \text{l}} \leq u^* \leq u^{*, \text{r}}$). Use a bisection algorithm with starting values $u^{*, \text{l}}$ and $u^{*, \text{r}}$ and a tolerance of $\eps_{\text{bisec}}$ to find $u^*$. Add any additional $u$'s and $F_W^\i(u)$'s computed in the bisection to the list $\mathcal{L}$.
            \item Find the largest number $u_l^{(1)}\in\mathcal{L}$ and the smallest number $u_l^{(2)}\in\mathcal{L}$ such that $u_l^{(1)}\leq u_l^{(2)}\leq u^*$,  $h(u_l^{(1)})\leq \eps_{\text{th}}$ and $h(u_l^{(2)})\geq \eps_{\text{th}}$. Then $u_l^{(1)} \leq u_l \leq u_l^{(2)}\leq u^*$. \\
            Similarly, find the largest number $u_r^{(1)}\in\mathcal{L}$ and the smallest number $u_r^{(2)}\in\mathcal{L}$ such that $u^*\leq u_r^{(1)}\leq u_r^{(2)}$, $h(u_r^{(1)})\geq \eps_{\text{th}}$ and $h(u_r^{(2)})\leq\eps_{\text{th}}$. Then $u^*\leq u_r^{(1)} \leq u_r \leq u_r^{(2)}$.\\
            Then use a bisection to find $u_l$ (using starting values $u_l^{(1)}$ and $u_l^{(2)}$) and $u_r$ (using starting values $u_r^{(1)}$ and $u_r^{(2)}$) with a tolerance of $\eps_{\text{bisec}}$. Add any additional $u$'s and $F_W^\i(u)$'s computed in the bisection to the list $\mathcal{L}$.
            \item Approximate $\log \int_{0}^{u_l} h(u)\,\rd u$ using a trapezoidal rule with knots $u_1',\dots,u_m'$ where $u_i'$ are those $u$'s in $\mathcal{L}$ satisfying $u\leq u_l$. Call the approximation $\hat{\mu}_{(0,u_l)}(\bx_s)$.
            \item Approximate $\log \int_{u_r}^{1} h(u)\,\rd u$ using a trapezoidal rule with knots $u_1'',\dots,u_p''$ where $u_i''$ are those $u$'s in $\mathcal{L}$ satisfying $u\geq u_r$. Call the approximation $\hat{\mu}_{(u_r,1)}(\bx_s)$.
            \item Apply Algorithm~\ref{alg:RQMC:logmu} where all uniforms $v\in(0,1)$ are replaced by $v' = u_l + (u_r - u_l) v\in(u_l, u_r)$. Call the output $\hat{\log\mu}$. Then set $\hat{\log\mu}_{(u_l,u_r)} = \log(u_r - u_l) + \hat{\log\mu}$ which estimates $\log \int_{u_l}^{u_r} h(u)\,\rd u$.
            \item Combine
            $$\hat{\mu}_{\log f(\bx_s)}^{\text{RQMC}} = \LSE\left(\hat{\mu}_{(0,u_l)}(\bx_s), \hat{\mu}_{(u_l,u_r)}(\bx_s), \hat{\mu}_{(u_r,1)}(\bx_s)\right)$$
        \end{enumerate}
        \item\label{return} Return $\hat{\mu}_{\log f(\bx_l)}^{\text{RQMC}}$, $l=1,\dots,N$
    \end{enumerate}
\end{algorithm}

\begin{remark}\label{remark:htilde}
Algorithm~\ref{alg:RQMC:f:x:adaptive} can be applied to estimate a slightly larger class of integrals. Let
$$ \mu = \int_0^\infty c w^{-k} \exp\left(m/w\right)\rd F_W(w)=\int_0^\infty \tilde{h}(u)\,\rd u;$$
here, $k,m>0$ are constant and $\tilde{h}(u)= c F_W^\i(u)^{-k} \exp\left(m/F_W^\i(u)\right)$ for $u\in(0,1)$.
A result similar to Lemma~\ref{lemma:propertiesh} applies to $\tilde{h}$ (replace $d$ by $k$ in the formula for $u^*$ in~\eqref{eq:hmax}). Thus, after only slight adjustments to Algorithm~\ref{alg:RQMC:f:x:adaptive}, the latter can be used to estimate $\log(\mu)$ efficiently. This will be useful in Section~\ref{sec:fitnvmix}.
\end{remark}

\section{Fitting multivariate normal variance mixtures}\label{sec:fitnvmix}

In this section, we derive an expectation-maximization (EM)-like algorithm
whose distinctive feature is that it can estimate the parameters of any given normal variance mixture. Its design is inspired by the ECME algorithm used for fitting multivariate $t$ models, but is appropriately modified to allow for a general mixing variable $W$. This requirement means that our approach must be able to handle the case where the density $f_{\bX}(\bx)$ may not exist in closed form, and must therefore be approximated. The fact that ECME-type algorithms break the optimization part into two steps---and thus handle the parameters $\bnu$ of $W$'s distribution separately from $\bmu$ and $\Sigma$---meshes very well with our assumption that all we may know about $W$ is through access to a ``black-box'' function for its quantile function. That is, since the step to find $\bnu$ is done separately, we can easily make it adaptable to whether or not $W$'s distribution is such that $f_{\bX}(\bx)$ exists in closed form. More precisely, in our \R\ implementation, we assume the user either provides a ``black-box'' function for the quantile function of $W$---in which case $f_{\bX}(\bx)$ is approximated using the algorithm described in the previous section---or specifies that $W$ is constant, inverse-gamma, or Pareto, in which case $f_{\bX}(\bx)$ is evaluated exactly. Examples provided in Section \ref{subsec:pnvmixexamples} demonstrate that the versatility of our algorithm, which we now explain, does not come at the cost of decreased accuracy.

Assume $\bX_1,\dots,\bX_n\isim\NVM_d(\bmu, \Sigma, F_W)$ with unknown location vector $\bmu$ and unknown scale matrix $\Sigma$ where $F_W$ has quantile function $F_W^\i(u; \bm{\nu})$ with unknown parameter vector $\bm{\nu}$ of length $p_{\bnu}$. For notational convenience, let  $\btheta = (\bm{\nu}, \bmu , \Sigma^{-1})$ and denote by $\btheta_k$ the current value of $\btheta$ in iteration $k$.

Before deriving our algorithm, we need some notation. The original log-likelihood is given by
\begin{align*}
\log L^{\text{org}}(\bnu, \bmu, \Sigma; \bX_1,\dots,\bX_n) &= \sum_{i=1}^n \log f_{\bX}(\bX_i;\bnu, \bmu, \Sigma)
\end{align*}
and the complete log-likelihood $\log L^{\text{c}}$ can be written as
\begin{align}\label{eq:logLc}
    \log L^{\text{c}}(\btheta; \bX_1,\dots,\bX_n,W_1,\dots,W_n) &= \sum_{i=1}^n \log f_{\bX,W}(\bX_i,W_i;\btheta)\nonumber \\
    &= \sum_{i=1}^n \log f_{\bX|W}(\bX_i\,|\,W_i;\bmu,\Sigma) + \sum_{i=1}^n \log f_{W}(W_i;\bnu),
\end{align}
where $W_1,\dots,W_n$ are (unobserved) iid copies of $W$. Note that the first sum contains the log-likelihood contributions of $\N_d(\bmu,W_i\Sigma)$ and thus is almost the log-likelihood of a normal distribution apart from potentially different $W_i$ (expected, for example, if $W$ is continuously distributed on the whole positive real line).
The expected value of the complete log-likelihood given the (observed) data $\bX_1,\dots,\bX_n$ and current estimate $\btheta_k$ is then
\begin{align}\label{eq:Qtheta}
Q(\btheta;\btheta_k) =  \E(\log L^{\text{c}}(\btheta; \bX_1,\dots,\bX_n,W_1,\dots,W_n)\,|\,\bX_1,\dots,\bX_n;\btheta_k).
\end{align}

As mentioned earlier, rather than trying to maximize $Q(\btheta;\btheta_k)$
over $\btheta$ as a classical EM algorithm would do, we instead
employ an ECME algorithm as developed in \cite{liurubin1994}; see also references therein for more details on variations of the EM algorithm.
In this way, and as explained below, optimization is broken into two steps, which respectively deal with $(\bmu,\Sigma)$ and $\bnu$.

The basic structure of our algorithm is as follows:

\begin{algorithm}[ECME Algorithm for fitting normal variance mixtures: Main idea]\label{alg:fitnvmixmainidea}
    Given iid data $\bX_1,\dots,\bX_n$, estimate $\bmu, \Sigma, \bnu$ via:
    \begin{enumerate}
        \item Obtain an initial estimate $\btheta_0=(\bnu_0, \bmu_0, \Sigma_0^{-1})$
        \item For $k=1,\dots$, repeat until convergence
        \begin{enumerate}
            \item\label{step:updatemusigma1} Update $\bmu_k$ and $\Sigma_k$ by maximizing $Q(\btheta;\btheta_k)$ with respect to $\bmu$ and $\Sigma$ with $\bnu = \bnu_{k-1}$ held fixed.
            \item Update $\bnu_k$ by maximizing $\log L^{\text{org}}(\bnu, \bmu_k, \Sigma_k; \bX_1,\dots,\bX_n)$ with respect to $\bnu$.
        \end{enumerate}
    \end{enumerate}
\end{algorithm}

That is, in the $k$'th iteration, we first update $\bmu$ and $\Sigma$ by maximizing the expected complete log-likelihood conditional on the observed data and then update $\bnu$ by maximizing the original likelihood with respect to $\bnu$ with $\bmu$ and $\Sigma$ set to their current estimates. This is an ECME algorithm as we either maximize the expected complete log-likelihood or the original likelihood; see also \cite{liurubin1995} for a discussion of an ECME algorithm for the multivariate $t$ distribution.

 Let
$$ \xi_{ki}=\E(\log W_i\,|\,\bX_i;\btheta_k)\quad\text{and}\quad \delta_{ki} =\E(1/W_i\,|\,\bX_i;\btheta_k),\quad i=1,\dots,n.$$
We calculate $Q(\btheta;\btheta_k)$ from~\eqref{eq:Qtheta} in the following lemma:

\begin{lemma}\label{lemma:Qtheta}
$Q(\btheta;\btheta_k)$ from~\eqref{eq:Qtheta} allows for the decomposition $Q(\btheta;\btheta_k) = Q_{\bX|W}(\bmu,\Sigma^{-1}; \btheta_k) + Q_W(\bnu;\btheta_k)$
where
\begin{align*}
Q_{\bX|W}(\bmu,\Sigma^{-1}; \btheta_k)&=-\frac{1}{2}\biggl(nd\log(2\pi)-n\log(\det(\Sigma^{-1}))+\sum_{i=1}^n(D^2(\bX_i;\bmu,\Sigma)\delta_{ki}+d\xi_{ki})\biggr),\\
Q_W(\bnu;\btheta_k) &=  \sum_{i=1}^n \E(\log f_{W}(W_i;\bnu)\,|\,\bX_i;\btheta_k).
\end{align*}
\end{lemma}

\begin{proof}
Starting from~\eqref{eq:Qtheta} and using~\eqref{eq:logLc} we obtain
{\allowdisplaybreaks
 \begin{align*}
    Q(\btheta;\btheta_k)&= \E(\log L^{\text{c}}(\btheta; \bX_1,\dots,\bX_n,W_1,\dots,W_n)\,|\,\bX_1,\dots,\bX_n;\btheta_k)\\
    &= \sum_{i=1}^n \E(\log f_{\bX|W}(\bX_i\,|\,W_i;\bmu,\Sigma)\,|\,\bX_1,\dots,\bX_n;\btheta_k)\\
    &\phantom{=+} + \sum_{i=1}^n \E(\log f_{W}(W_i;\bnu)\,|\,\bX_1,\dots,\bX_n;\btheta_k)\\
    &= \sum_{i=1}^n \E(\log f_{\bX|W}(\bX_i\,|\,W_i;\bmu,\Sigma)\,|\,\bX_i;\btheta_k) + \sum_{i=1}^n \E(\log f_{W}(W_i;\bnu)\,|\,\bX_i;\btheta_k)\\
    &= Q_{\bX|W}(\bmu,\Sigma^{-1};\btheta_k) + Q_W(\bnu;\btheta_k)
 \end{align*}}%
where the first expectation is taken with respect to $W_1,\dots,W_n$ for given $\bX_1,\dots,\bX_n$ and $\btheta_k$, and the last line of the displayed equation is understood as the definition of $Q_{\bX|W}$ and $Q_W$. Using that $\bX\mid W \sim \N_d(\bmu, W\Sigma)$,
it is easily verified that
\begin{align*}
    Q_{\bX|W}(\bmu,\Sigma^{-1};\btheta_k) &= \sum_{i=1}^n \E(\log f_{\bX|W}(\bX_i\,|\,W_i;\bmu,\Sigma)\,|\,\bX_i;\btheta_k)\\
    &= -\frac{1}{2}\biggl(nd\log(2\pi)-n\log(\det(\Sigma^{-1}))+\sum_{i=1}^n(D^2(\bX_i;\bmu,\Sigma)\delta_{ki}+d\xi_{ki})\biggr).
 \end{align*}
\end{proof}

\noindent
With $Q(\btheta; \btheta_k)$ at hand, we show in the following lemma how $\bmu$ and $\Sigma$ are updated in Step~\ref{step:updatemusigma1} of Algorithm~\ref{alg:fitnvmixmainidea}.
\begin{lemma}\label{lemma:musigmaupdate}
Maximizing $Q(\btheta;\btheta_k)$ with respect to $\bmu$ and $\Sigma$ in Step~\ref{step:updatemusigma1} of Algorithm~\ref{alg:fitnvmixmainidea} gives the next iterates
 \begin{align}\label{eq:musigmaupdate}
    \bmu_{k+1} = \frac{\sum_{i=1}^n\delta_{ki}\bX_i}{\sum_{i=1}^n\delta_{ki}}\quad\text{and}\quad \Sigma_{k+1} = \frac{1}{n}\sum_{i=1}^n\delta_{ki}(\bX_i-\bmu_k)(\bX_i-\bmu_k)\T.
\end{align}
\end{lemma}
\begin{proof}
By Lemma~\ref{lemma:Qtheta}, $Q(\btheta;\btheta_k) =  Q_{\bX|W}(\bmu,\Sigma^{-1};\btheta_k)+ Q_W(\bnu;\btheta_k)$ and $\bmu$ and $\Sigma$ do not appear in $Q_W(\bnu;\btheta_k)$ so that we only need to maximize $ Q_{\bX|W}(\bmu,\Sigma^{-1};\btheta_k)$. \\
The necessary conditions are $\frac{\partial}{\partial \bmu} Q_{\bX|W}(\bmu,\Sigma^{-1};\btheta_k) = 0$ and $\frac{\partial}{\partial \Sigma^{-1}} Q_{\bX|W}(\bmu,\Sigma^{-1};\btheta_k) = 0$. Using $\frac{\partial}{\partial \bmu} D^2(\bX_i; \bmu, \Sigma) = - 2\Sigma^{-1}(\bX_i-\bmu)$ one obtains that $\frac{\partial}{\partial \bmu}  Q_{\bX|W}(\bmu,\Sigma^{-1};\btheta_k) = 0$ if and only if $\sum_{i=1}^n \delta_{ki} \Sigma^{-1}(\bX_i-\bmu)=0$. Solving for $\bmu$ gives $\bmu_{k+1}$ as given in the lemma.%
For full rank $\Sigma$, it holds that $\frac{\partial}{\partial \Sigma^{-1}} \log\det(\Sigma^{-1}) = \Sigma$. Since $\frac{\partial}{\partial \Sigma^{-1}} D^2(\bX_i; \bmu, \Sigma) = (\bX_i-\bmu)(\bX_i-\bmu)\T$ one gets $\frac{\partial}{\partial \Sigma^{-1}} Q_{\bX|W}(\bmu,\Sigma^{-1};\btheta_k) = 0$ if and only if $n\Sigma - \sum_{i=1}^n\delta_{ki}(\bX_i-\bmu)(\bX_i-\bmu)\T=0$ which, after solving for $\Sigma$, gives the formula for $\Sigma_{k+1}$ as given in the statement.
\end{proof}

Lemma~\ref{lemma:musigmaupdate} indicates that we need to approximate the weights $\delta_{ki}$, $i=1,\dots,n$, in Step~\ref{step:updatemusigma1} of Algorithm~\ref{alg:fitnvmixmainidea}. Note that
\begin{align*}
        \rd F_{W|\bX}(w\,|\,\bx) & = \frac{f_{\bX|W}(\bx\,|\,w)\, \rd F_W(w)}{f_{\bX}(\bx)}= \frac{ \phi(\bx; \bmu, w\Sigma)}{f_{\bX}(\bx)}\,\rd F_W(w),\quad w > 0,
\end{align*}
where $\phi(\bx; \bmu, \Sigma)$ denotes the density of $\N_d(\bmu, \Sigma)$ so that
\begin{align*}
    \delta_{ki}&=\E\left(\frac{1}{W_i}\,\middle\vert\ \bX_i;\btheta_k\right)
    =\int_0^\infty \frac{1}{w}\,\rd F_{W|\bX}(w\,|\,\bX_i)\\
     &=\frac{1}{f_{\bX}(\bX_i; \bmu_k, \Sigma_k, \bnu_k)}\int_0^\infty \frac{\phi(\bX_i; \bmu_k, w\Sigma_k)}{w}\, \rd F_W(w; \bnu_k).
\end{align*}
This yields
\begin{align}
    \log(\delta_{ki}) &= \log\left( \int_0^\infty \frac{\phi(\bX_i; \bmu_k, w\Sigma_k)}{w}\, \rd F_W(w; \bnu_k)\right) - \log f_{\bX}(\bX_i;  \bmu_k, \Sigma_k, \bnu_k)\nonumber \\
    &= \log\left(\int_0^1 \frac{1}{\sqrt{ (2\pi)^d F_W^\i(u; \bnu_k)^{d+2} |\Sigma_k|}} \exp\left(-\frac{D^2(\bX_i;\bmu_k,\Sigma_k)}{2F_W^\i(u; \bnu_k) }\right)\rd u\right) \nonumber\\
    &\phantom{=} -  \log\left(\int_0^1 \frac{1}{\sqrt{ (2\pi)^d F_W^\i(u; \bnu_k)^{d} |\Sigma_k|}} \exp\left(-\frac{D^2(\bX_i;\bmu_k,\Sigma_k)}{2F_W^\i(u; \bnu_k) }\right)\rd u\right). \label{eq:logdeltaki}
\end{align}
Estimation of the latter integral (corresponding to $\log f_{\bX}(\bx)$) was discussed in Algorithm~\ref{alg:RQMC:f:x:adaptive}; the former integral differs from the latter only by a factor of $F_W^\i(u)^{-1}$, and can be estimated similarly; see Remark~\ref{remark:htilde} for details.

Summarizing, the $k$'th iteration of the algorithm consists of approximating the weights $\delta_{ki}$, $i=1,\dots,n$ with $\bnu=\bnu_k$ held fixed (which are then used to update $\bmu$ and $\Sigma$ as in \eqref{eq:musigmaupdate}) and then updating $\bnu$ by maximizing the original likelihood $\log L^{\text{org}}(\btheta; \bX_1,\dots,\bX_n)$ as a function of $\bnu$ with $\bmu$ and $\Sigma$ set to their current estimates, i.e., we set
\begin{align}\label{eq:nuupdate}
\bnu_{k+1} = \argmax\limits_{\bnu} \log L^{\text{org}}(\bnu, \bmu_{k+1}, \Sigma_{k+1}; \bX_1,\dots,\bX_n)
\end{align}
and solve this $p_{\bnu}$-dimensional optimization problem numerically. This optimization problem is the same optimization problem one would solve if $\bmu$ and $\Sigma$ were known (and given by $\bmu_{k+1}$ and $\Sigma_{k+1}$) and is a classical ingredient in ECME algorithms; for more details on rates of convergence of the proposed ECME scheme, see \cite[Section~4]{liurubin1994}. Note that the dimension $p_{\bnu}$ of $\bnu$ is typically small so that this optimization problem is also numerically feasible. In our implementation, we use the \R\ optimizer \texttt{optim()} which by default only relies on function evaluations and works for non-differentiable functions: Derivative-based methods can, due to small estimation errors in the likelihood function, fail to detect a global optimum.

This step is the most costly one as it involves multiple estimation of the likelihood of the data using Algorithm~\ref{alg:RQMC:f:x:adaptive}: Each call to the likelihood function requires the approximation of $n$ integrals. It turns out that estimating the weights $\delta_{ki}$ is faster so that it seems to be fruitful to first update $\bmu$ and $\Sigma$ until convergence (with $\bnu=\bnu_k$ held fixed) and then update $\bnu$. In fact, this can be done efficiently: The weights $\delta_{ki}$ depend on $\bX_i$, $\bmu_k$ and $\Sigma_k$ only through the Mahalanobis distances $D^2(\bX_i;\bmu_k,\Sigma_k)$. Once $\bmu_k$ and $\Sigma_k$ are updated to, say, $\bmu_k'$ and $\Sigma_k'$, (some of) the new weights $\delta_{ki}'$ for the new Mahalanobis distances $D^2(\bX_i;\bmu'_k,\Sigma_k')$ can be obtained by interpolating the already calculated weights $\delta_{ki}$ corresponding to the (old) Mahalanobis distances $D^2(\bX_i;\bmu_k,\Sigma_k)$.

It remains to discuss how a starting value $\btheta_0$ can be found. We suggest using $\bmu_0 = \bar{\bX}_n$, the sample mean vector, as an unbiased estimator for $\bmu$. Denote by $S_n$ the sample covariance matrix (Wishart matrix) of $\bX_1,\dots,\bX_n$. Since $S_n$ is unbiased for $\cov(\bX)$ it follows that $\E(S_n) = \E(W) \Sigma$. The idea is now to maximize the likelihood given $\bmu=\bmu_0$ and given $\Sigma = c\cdot S_n$ with respect to $\bnu$ and $c$ (restricted to $c>0$) which is a $(p_{\bnu}+1)$ dimensional optimization problem. That is, we find
\begin{align}\label{eq:startingvalue}
(\bm{\nu^*}, c^*) = \argmax\limits_{\bnu, c>0} L^{\text{org}}(\bnu, \bmu_0, c S_n; \bX_1,\dots,\bX_n)
\end{align}
numerically (again via \R's \texttt{optim()}) and set $\bnu_0 = \bm{\nu^*}$ and $\Sigma_0 = c^* S_n$ which is just a multiple of the Wishart matrix. As this step is merely needed to obtain a starting value for $\bnu$ and $\Sigma$, this optimization can be done over a subset of the sample $\{\bX_1,\dots,\bX_n\}$ to save run time.

The complete procedure is summarized in Algorithm~\ref{alg:fitnvmix}. As convergence criterion we suggest stopping once the maximal relative difference in parameter estimates is smaller than a given threshold. We define the maximal relative difference by
$$ d(\bnu_k, \bnu_{k+1}) = \max_{i=1,\dots,p_{\bnu}} \frac{ | \bnu_{k,i}-\bnu_{k+1,i} | }{ | \bnu_{k,i} |}, \quad \bnu_k=(\bnu_{k,1},\dots,\bnu_{k,p_{\bnu}})$$
and similarly for $\bmu$ and $\Sigma$.

\begin{algorithm}[ECME algorithm for fitting normal variance mixtures]\label{alg:fitnvmix}
    Given iid input data $\bX_1,\dots,\bX_n$ and convergence criteria $\eps_{\bmu}$, $\eps_\Sigma$ and $\eps_{\bnu}$, estimate $\bmu, \Sigma, \bnu$ via:
    \begin{enumerate}
        \item\label{step:startingvalue} \underline{\emph{Starting value.}}\\
         Set $\bmu_0 = \bar{\bX}_n$ and solve the optimization problem~\eqref{eq:startingvalue} numerically to obtain $\bm{\nu^*}$ and $c^*$. Set $\bnu_0 = \bm{\nu^*}$ and $\Sigma_0 = c^* S_n$.
        \item\label{step:ecmeiteration} \underline{\emph{ECME iteration.}} \\
        For $k=0,1,\dots$, do:
        \begin{enumerate}
            \item\label{step:updatemusigma} \emph{Update $\bmu$ and $\Sigma$.}\\
            Set $\bmu_k^{(1)} = \bmu_k$ and $\Sigma_k^{(1)} = \Sigma_k$. \\
            For $l=1,\dots$, do:
            \begin{enumerate}
                \item\label{step:deltaki} Estimate new weights $\delta_{ki}^{(l+1)} = \E(1/W_i\,|\,\bX_i; \bmu_k^{(l)}, \Sigma_k^{(l)}, \bnu_k)$, $i=1,\dots,n$ using \eqref{eq:logdeltaki} and Algorithm~\ref{alg:RQMC:f:x:adaptive}.
                \item Calculate the new iterates $\bmu_{k}^{(l+1)}$ and $\Sigma_{k}^{(l+1)}$ using \eqref{eq:musigmaupdate} with weights $\delta_{ki}^{(l+1)}$, $i=1,\dots,n$.
                \item If $d(\bmu_k^{(l)}, \bmu_k^{(l+1)})<\eps_{\bmu}$ and $d(\Sigma_k^{(l)}, \Sigma_k^{(l+1)})<\eps_{\Sigma}$, set $\bmu_{k+1}=\bmu_k^{(l+1)}$, $\Sigma_{k+1}=\Sigma_k^{(l+1)}$ and go to Step~\ref{step:updatenu}.
            \end{enumerate}
            \item\label{step:updatenu} \emph{Update $\bnu$.}\\
            Numerically solve the optimization problem~\eqref{eq:nuupdate} to obtain $\bnu_{k+1}$.
            \item If $d(\bnu_k, \bnu_{k+1})<\eps_{\bnu}$, return the MLEs $\bmu^*=\bmu_{k+1}$, $\Sigma^*=\Sigma_{k+1}$ and $\bnu^*=\bnu_{k+1}$.
        \end{enumerate}
    \end{enumerate}
\end{algorithm}

Algorithm~\ref{alg:fitnvmix} is implemented in the function \texttt{fitnvmix()} of our \R\ package \texttt{nvmix}. The mixing variable is specified by providing a function to the argument \texttt{qmix}.
In the special case where $W$ follows an inverse-gamma or Pareto distribution, the density function is
known in closed form which is used by \texttt{fitnvmix()} when called with argument \texttt{qmix = "inverse.gamma"} or \texttt{qmix = "pareto"}.

\section{Numerical Examples}\label{sec:numericalexamples}

In this section we provide a careful numerical analysis of all algorithms presented. The first part discusses the type of mixing distributions used; the second, third and fourth part detail numerical examples for estimating the distribution function using Algorithm~\ref{alg:RQMC:F:a:b} with variable reordering as in Algorithm~\ref{alg:precond}, estimating the log-density function using Algorithm~\ref{alg:RQMC:f:x:adaptive}, and estimating parameters $\nu$, $\bmu$ and $\Sigma$ given a random sample using Algorithm~\ref{alg:fitnvmix}, respectively. The last part provides an application of our methods to a multivariate financial data set.

\subsection{Test Distributions}

For our numerical examples, we consider two distributions for the mixing variable~$W$, an inverse-gamma distribution (so that $\bX$ is multivariate $t$) and a Pareto distribution.

\paragraph{Inverse-gamma mixture} Here $W$ follows an inverse-gamma distribution with shape and scale parameter $\nu / 2$. The resulting distribution is the multivariate $t$ distribution, $\bX\sim \text{MVT}_d(\nu,\bmu,\Sigma)$ with positive degrees of freedom $\nu$; see, for instance, \cite[Chapter~1]{kotznadarajah2004}. Note that if $\nu>1$, $\E(\bX)=\bmu$ and if $\nu>2$, $\cov(\bX)=\frac{\nu}{\nu-2}\Sigma$. The multivariate $t$ distribution has the density
\begin{align}\label{eq:dens:mvt}
 f_{\bX}(\bx) = \frac{ \Gamma( (\nu+d)/2 ) }{ \Gamma(\nu/2) \sqrt{ (\nu\pi)^d |\Sigma|}}\left(1+D^2(\bx; \bmu, \Sigma)/\nu\right)^{-\frac{\nu+d}{2}},\;\;\; \bx\in\mathbb{R}^d.
\end{align}
For the ECME procedure it is useful to calculate the weight $\E(1/W \mid \bX)$. Since
\begin{align*}
f_{W\mid \bX}(w\mid \bx) &\propto f_{\bX\mid W}(\bx\mid w) f_W(w)\propto w^{-\frac{d+\nu}{2}-1} \exp\left(- \frac{ (D^2(\bx; \bmu, \Sigma)+\nu)/2}{w}\right),\quad w>0,
\end{align*}
$W\mid \bX$ follows an inverse-gamma distribution, i.e., $W\mid \bX\sim \operatorname{IG}( (d+\nu)/2, (D^2(\bX; \bmu, \Sigma)+\nu)/2)$. This implies
$$ \E(1/W \mid \bX) = \frac{\nu + d}{\nu + D^2(\bX; \bmu, \Sigma)},$$
so that the weights $\delta_{ki}$ in Step~\ref{step:deltaki} of Algorithm~\ref{alg:fitnvmix} can be calculated analytically in this case.

\paragraph{Pareto mixture}
In order to test our algorithms for a normal variance mixture distribution that has not been studied as extensively as the multivariate $t$ distribution we consider $W\sim\Par(\alpha, x_m)$ with density
$$ f_W(w) = \alpha \frac{x_m^\alpha}{w^{\alpha+1}},\;\;\;w\geq x_m.$$
One can calculate that $\E(W^k)$ exists with $\E(W^k)=\alpha/(\alpha-k)$ if $k<\alpha$. This implies for the resulting normal variance mixture $\bX = \bmu + \sqrt{W} A \bZ$ that $\E(\bX)=\bmu$ for $\alpha>1/2$ and $\cov(\bX)=\frac{\alpha}{\alpha-1}\Sigma$ for $\alpha>1$. The density $f_{\bX}(\bx)=f_{\bX}(\bx; \mu, \Sigma, \alpha, x_m)$ can be determined using~\eqref{eq:densityX:fw}:
\begin{align*}
f_{\bX}(\bx) &= \frac{\alpha x_m^\alpha}{\sqrt{(2\pi)^d |\Sigma|}}\int_{x_m}^\infty w^{-d/2-\alpha-1}\exp\left(-\frac{D^2(\bx; \bmu, \Sigma)}{2w}\right)\rd w\\
&= \frac{\alpha x_m^\alpha}{\sqrt{(2\pi)^d |\Sigma|}}\left(\frac{D^2(\bx; \bmu, \Sigma)}{2}\right)^{-d/2-\alpha}\int_0^{\frac{D^2(\bx; \bmu, \Sigma)}{2x_m}} u^{d/2+\alpha-1} \exp(-u)\,\rd u\\
&= \frac{\alpha x_m^\alpha}{\sqrt{(2\pi)^d |\Sigma|}}\left(\frac{D^2(\bx; \bmu, \Sigma)}{2}\right)^{-d/2-\alpha}\gamma\left(\alpha+\frac{d}{2}; \frac{D^2(\bx; \bmu, \Sigma)}{2 x_m}\right), \;\;\; \bx\in\mathbb{R}^d,
\end{align*}
where $\gamma(z; x) = \int_0^x t^{z-1}e^{-t}\,\rd t$ for $z,x>0$ denotes the (lower) incomplete gamma function. Note that $f_{\bX}(\bx; \mu, \Sigma, \alpha, x_m) = f_{\bX}(\bx; \mu, x_m\Sigma, \alpha, 1)$ so that the scale parameter $x_m$ is redundant as the scaling can be achieved via scaling $\Sigma$. We can thus set $x_m=1$ and obtain
\begin{align}\label{eq:densityX:Paretomix}
f_{\bX}(\bx; \mu, \Sigma, \alpha) = \frac{\alpha}{\sqrt{(2\pi)^d |\Sigma|}}\left(\frac{D^2(\bx; \bmu, \Sigma)}{2}\right)^{-d/2-\alpha}\gamma\left(\alpha+\frac{d}{2}; \frac{D^2(\bx; \bmu, \Sigma)}{2}\right), \;\;\;\bx\in\mathbb{R}^d.
\end{align}
We use the notation $\bX\sim\PNVM(\alpha,\bmu,\Sigma)$ (``Pareto normal variance mixture'') for a random vector $\bX$ with density~\eqref{eq:densityX:Paretomix}.

As in the case of an inverse-gamma mixture, it is possible to derive an expression for $\E(1/W\mid \bX)$ in the Pareto setting. Note that
\begin{align*}
f_{W\mid \bX}(w\mid \bx) &\propto f_{\bX\mid W}(\bx\mid w) f_W(w)\propto w^{-(\alpha+d/2+1)} \exp\left(-D^2(\bx; \bmu, \Sigma)/(2w) \right),\quad w>1,
\end{align*}
so that using the density transformation formula we obtain for $\tilde{W}=1/W$ that
$$ f_{\tilde{W} \mid \bX}(\tilde{w} \mid \bx) \propto \tilde{w}^{\alpha+d/2-1} \exp(-\tilde{w} D^2(\bx; \bmu, \Sigma)/2),\quad\tilde{w}\in(0,1).$$
Therefore, $W^{-1}\mid \bX$ follows a $(0,1)$ truncated gamma distribution with shape $\alpha+d/2$ and scale $2/D^2(\bX; \bmu, \Sigma)$. For more details on truncated gamma distributions, see \cite{coffeymuller2000};  Equation (2.12) therein implies that
$$ \E(1/W\mid \bX) = \frac{F_{\Gamma}(1; \alpha+d/2+1, 2/D^2(\bX; \bmu, \Sigma))}{F_{\Gamma}(1; \alpha+d/2, 2/D^2(\bX; \bmu, \Sigma))}\frac{2\alpha + d}{D^2(\bX; \bmu, \Sigma)}.$$

\subsection{Estimating the distribution function}\label{subsec:pnvmixexamples}
In the case where $\bX\sim \MVT_d(\nu,\bmu,\Sigma)$, Algorithm~\ref{alg:RQMC:F:a:b} combined with the variable reordering Algorithm~\ref{alg:precond} can be used to estimate $F(\ba,\bb)$, and is implemented in the function \texttt{pStudent()} in the \R\ package \texttt{nvmix}. In this case, one can also use the QRSVN algorithm from  \cite{genzbretz2002}, which is implemented in the function \texttt{pmvt()} of the \R\ package \texttt{mvtnorm} (\cite{genzbretzmivamileischscheiplhothorn2019}).
The differences between these two algorithms was explained in Section \ref{sec:mcrqmc}.
Furthermore, our implementation relies on C code, whereas \texttt{pmvt()} internally calls Fortran code.

\subsubsection{Error behaviour as a function of the sample size}

\begin{figure}[!t]
\centering
  \includegraphics[width=0.32\linewidth]{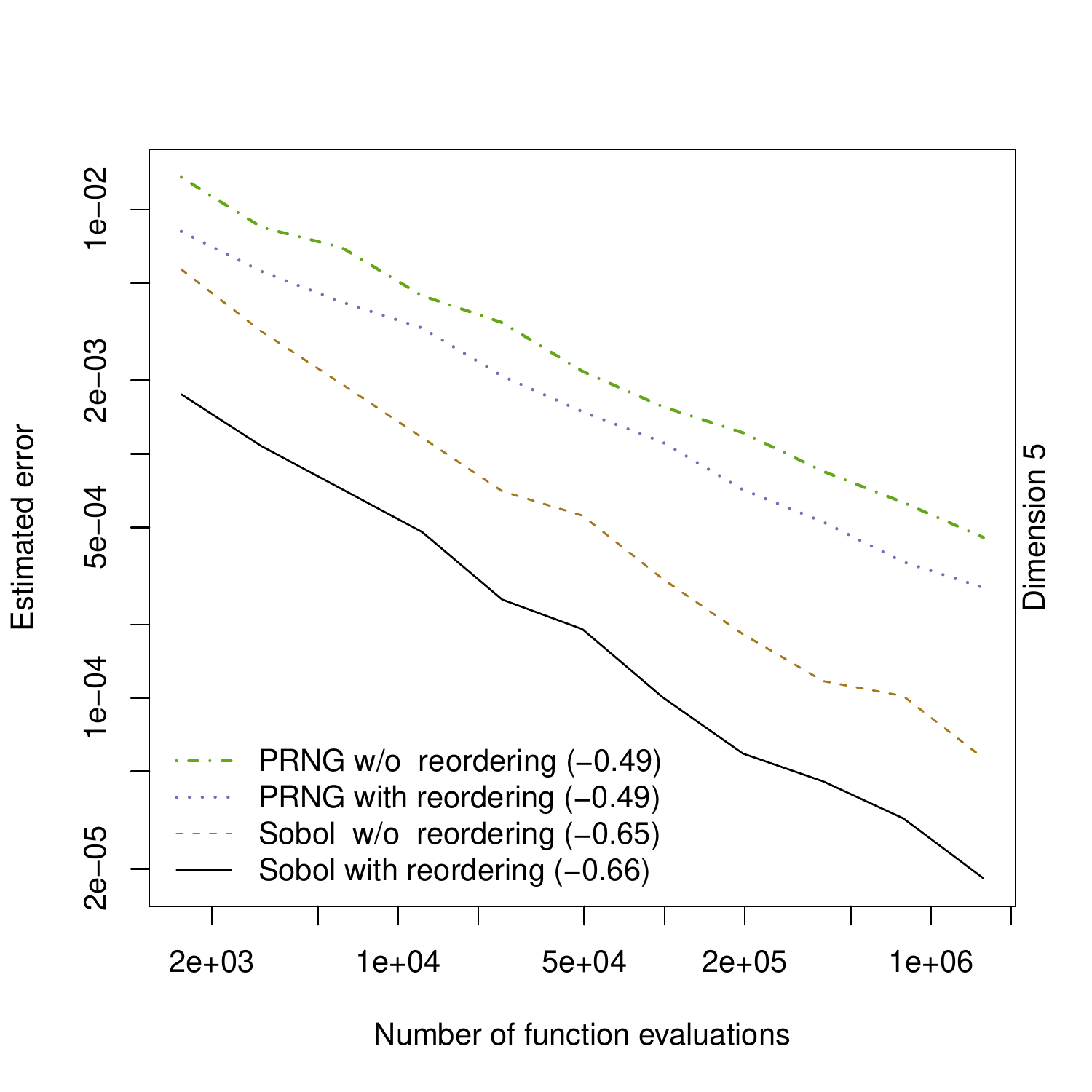}
  \includegraphics[width=0.32\linewidth]{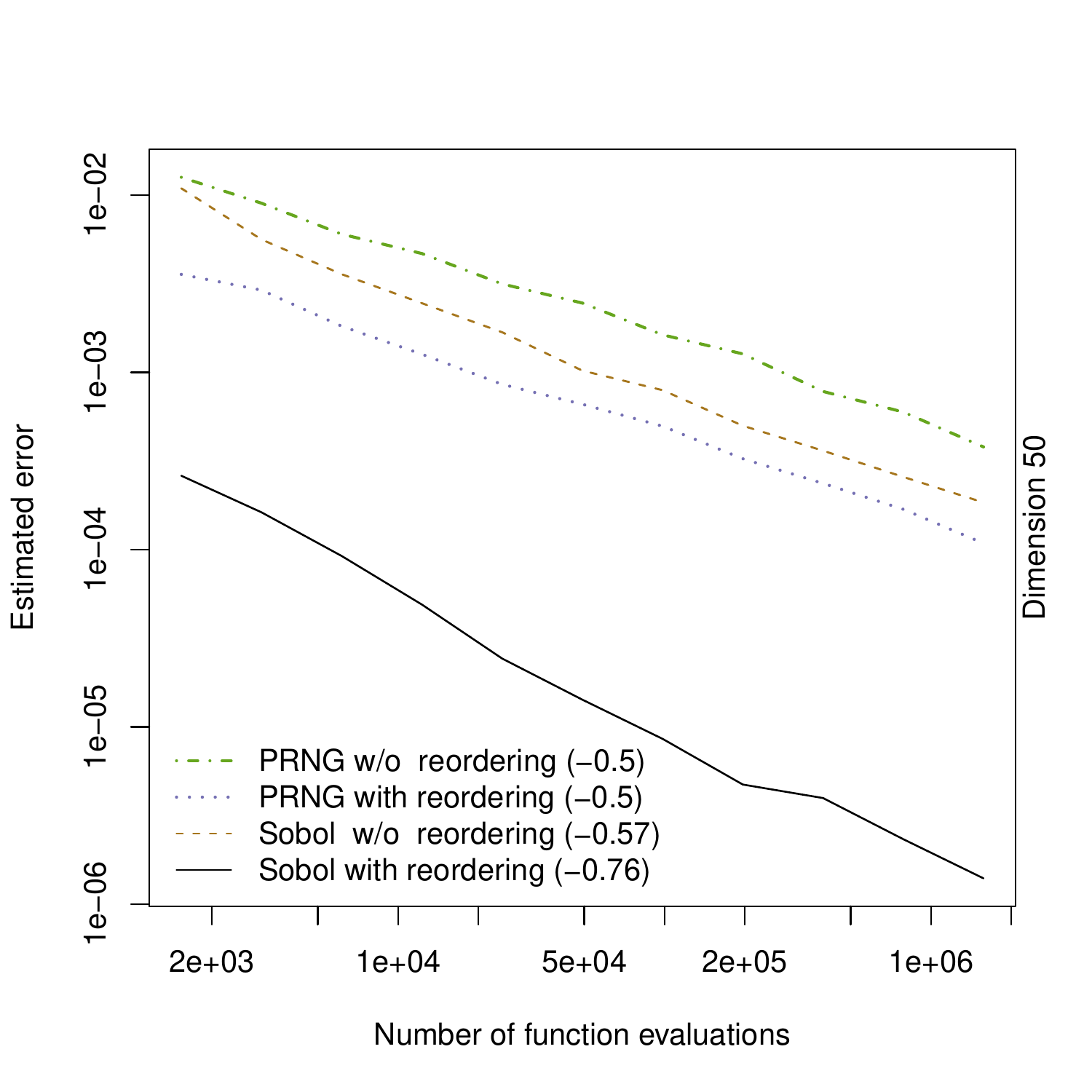}
  \includegraphics[width=0.32\linewidth]{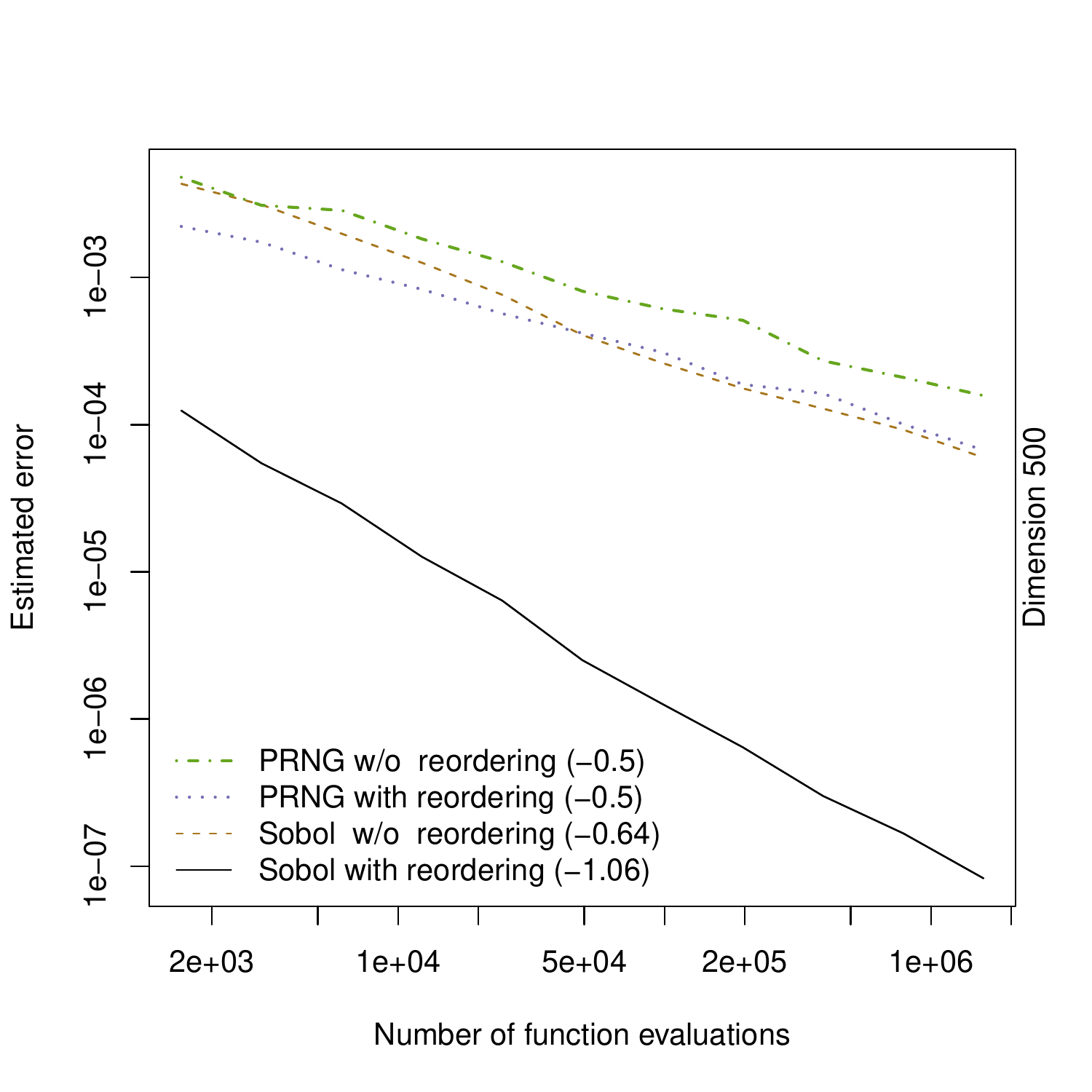}
  \caption{Average absolute errors of different estimators for $F_{\bX}(\bx)$ as a function of $n$ for $\bX\sim\MVT_d(2, \bzero, \Sigma)$, where for each $n$, 15 different settings for $\Sigma$ and $\bx$ are randomly chosen. Regression coefficients are in parentheses in the legends.}
  \label{fig:meanerrors.t}
\end{figure}

In order to assess the performance of our algorithm let us first consider estimated absolute errors as a function of the number of function evaluations. Four settings are considered: (pure) MC with and without reordering and RQMC (using a randomized Sobol' sequence) with and without reordering. In Figures~\ref{fig:meanerrors.t} and~\ref{fig:meanerrors.Pareto}, estimated absolute errors (estimated as in Algorithm~\ref{alg:RQMC:F:a:b} via $\hat{\varepsilon}$ in Step 4.3)) are reported for different sample sizes $n$ (which refer to the total number of function evaluations) in different dimensions using the four aforementioned methods for the multivariate $t$ case and the Pareto mixture. For each dimension and for each $n$ we report the average estimated absolute error for 15 different parameter settings. In each parameter setting, an upper limit is randomly chosen via $\bb\sim\U(0,3\sqrt{d})^d$ and a correlation matrix $R$ is sampled as a standardized Wishart matrix via the function \texttt{rWishart()} in \R\. The lower limit is set to $\ba=(-\infty,\dots,-\infty)$. The degrees of freedom $\nu$ in the $\MVT$ setting and the shape parameter $\alpha$ in the $\PNVM$ setting are set to 2.

It is evident that RQMC methods yield lower errors than their MC counterparts. We also report the convergence speed (as measured by the regression coefficient $\alpha$ of $\log\hat{\varepsilon} = \alpha \log n + c$ displayed in the legend): Variable reordering does not have an influence on the convergence speed $1/\sqrt{n}$ of MC methods; however, it does speed up the RQMC methods. A possible explanation is that variable reordering can reduce the effective dimension. This is discussed below in more detail.

\begin{figure}[!b]
\centering
  \includegraphics[width=0.32\linewidth]{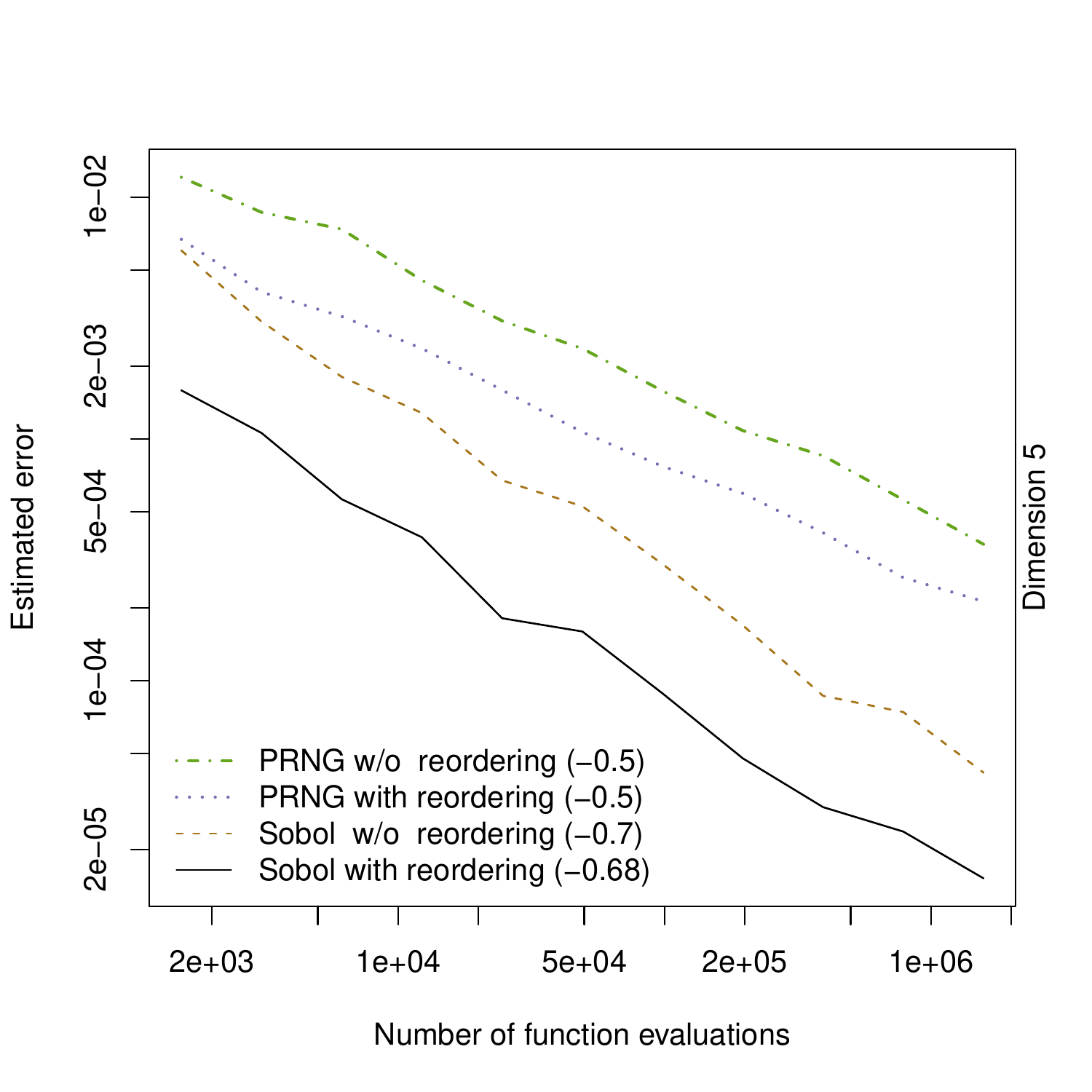}
  \includegraphics[width=0.32\linewidth]{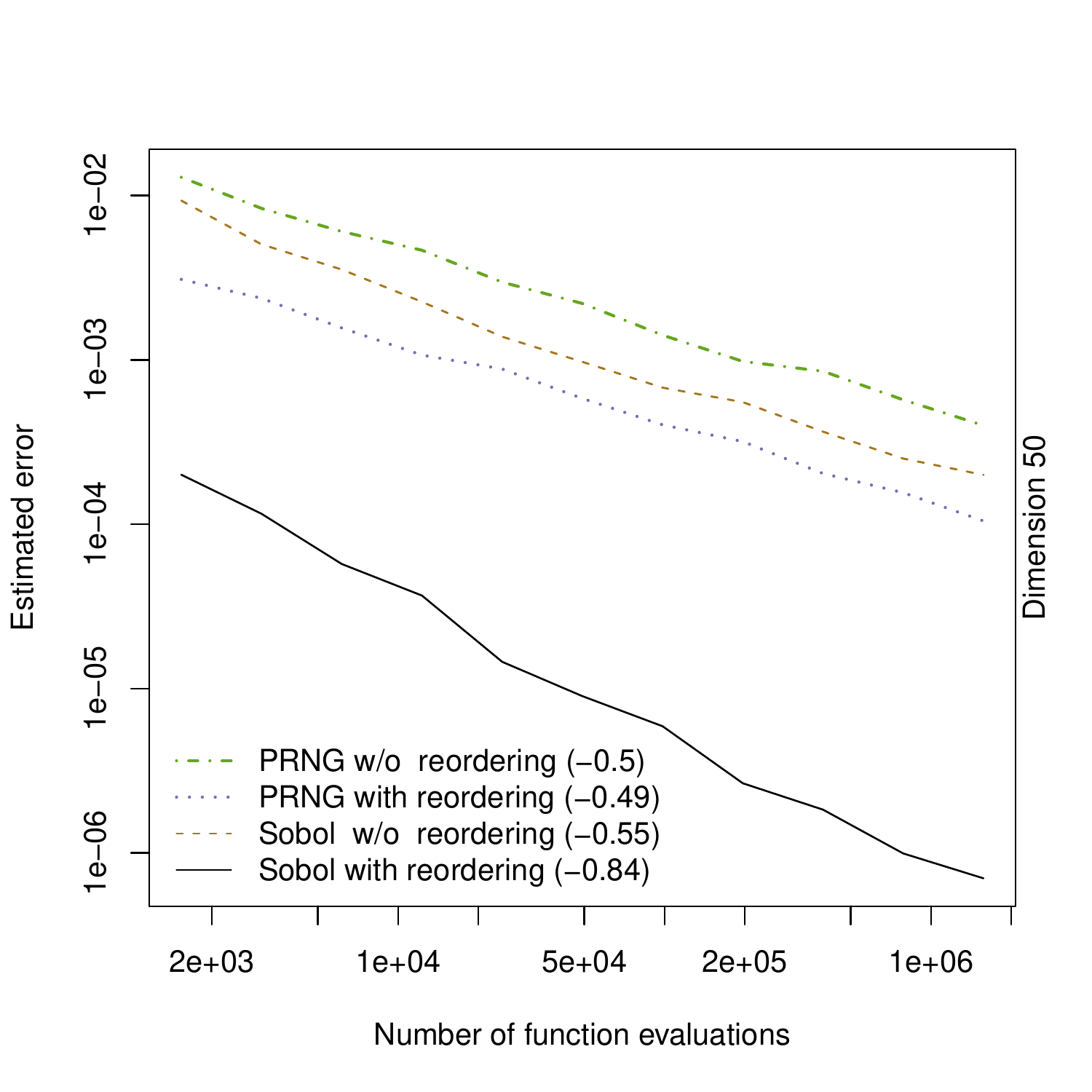}
  \includegraphics[width=0.32\linewidth]{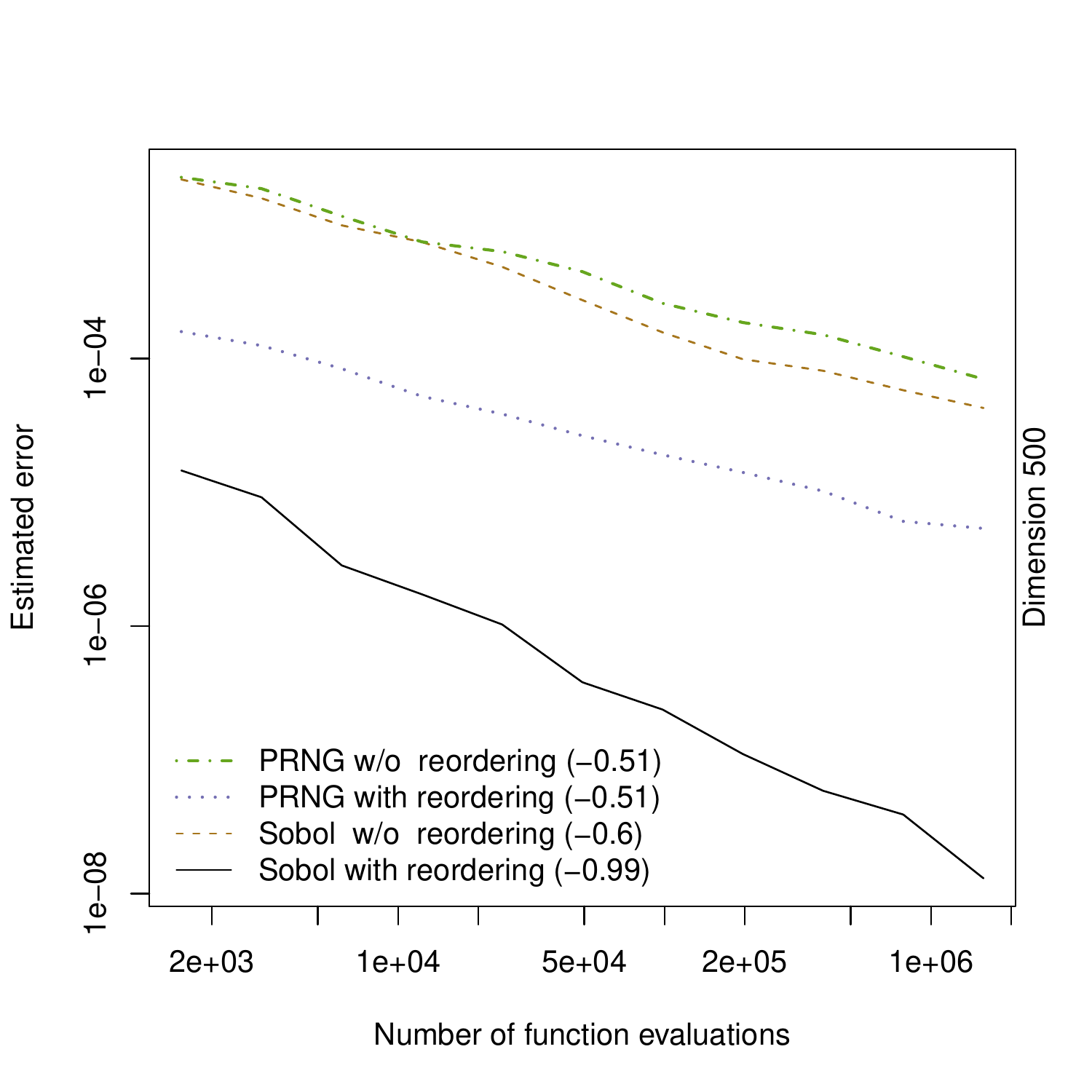}
  \caption{Average absolute errors of different estimators for $F_{\bX}(\bx)$ as a function of $n$ for $\bX\sim\PNVM_d(2, \bzero, \Sigma)$, where for each $n$, 15 different settings for $\Sigma$ and $\bx$ are randomly chosen. Regression coefficients are in parentheses in the legends.}
  \label{fig:meanerrors.Pareto}
\end{figure}

\subsubsection{The effect of variable reordering}
\label{subsubsec:VarReorder}

\begin{figure}[!t]
\centering
\begin{minipage}{.45\textwidth}
  \centering
  \includegraphics[width=\linewidth]{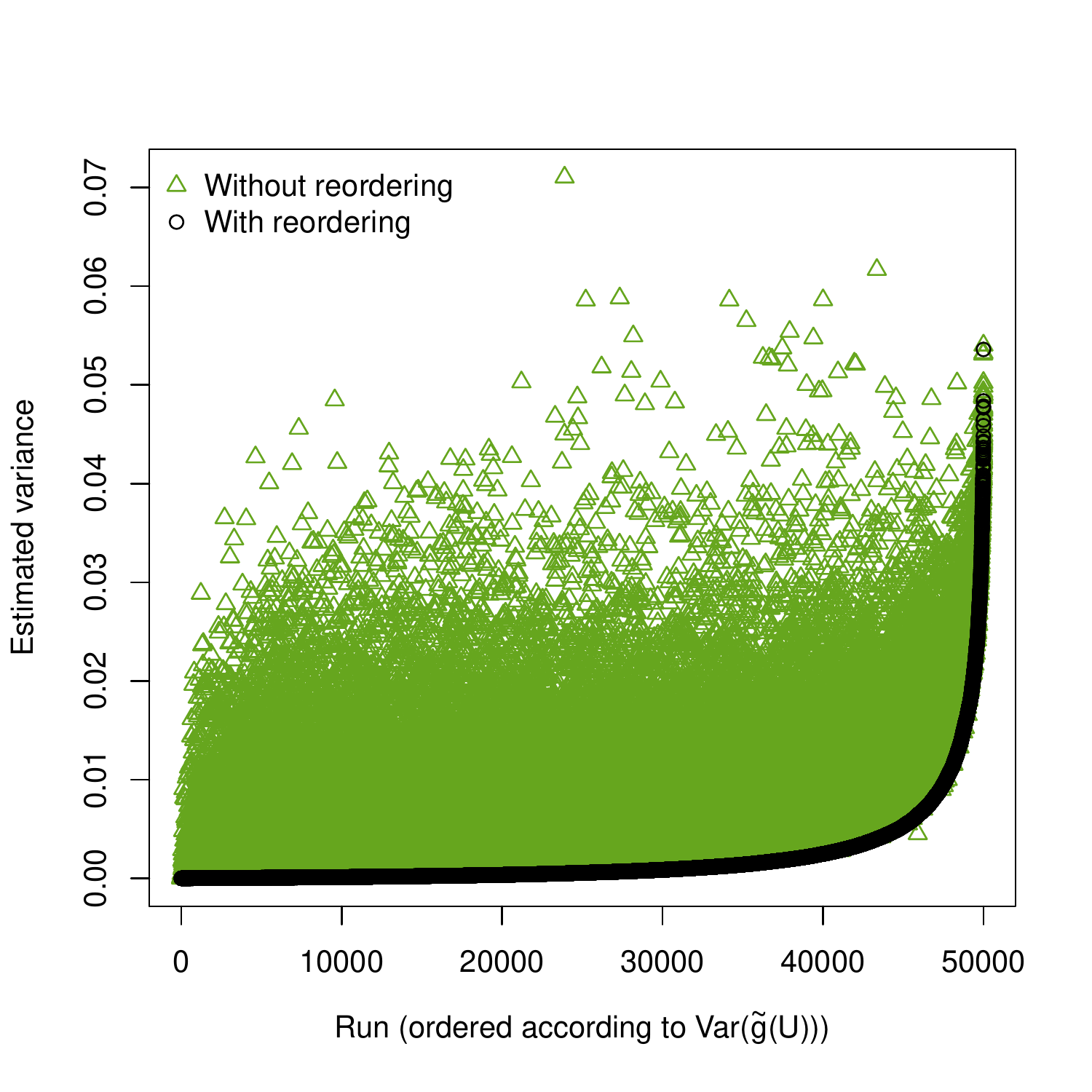}
\end{minipage}
\begin{minipage}{.45\textwidth}
  \centering
  \includegraphics[width=\linewidth]{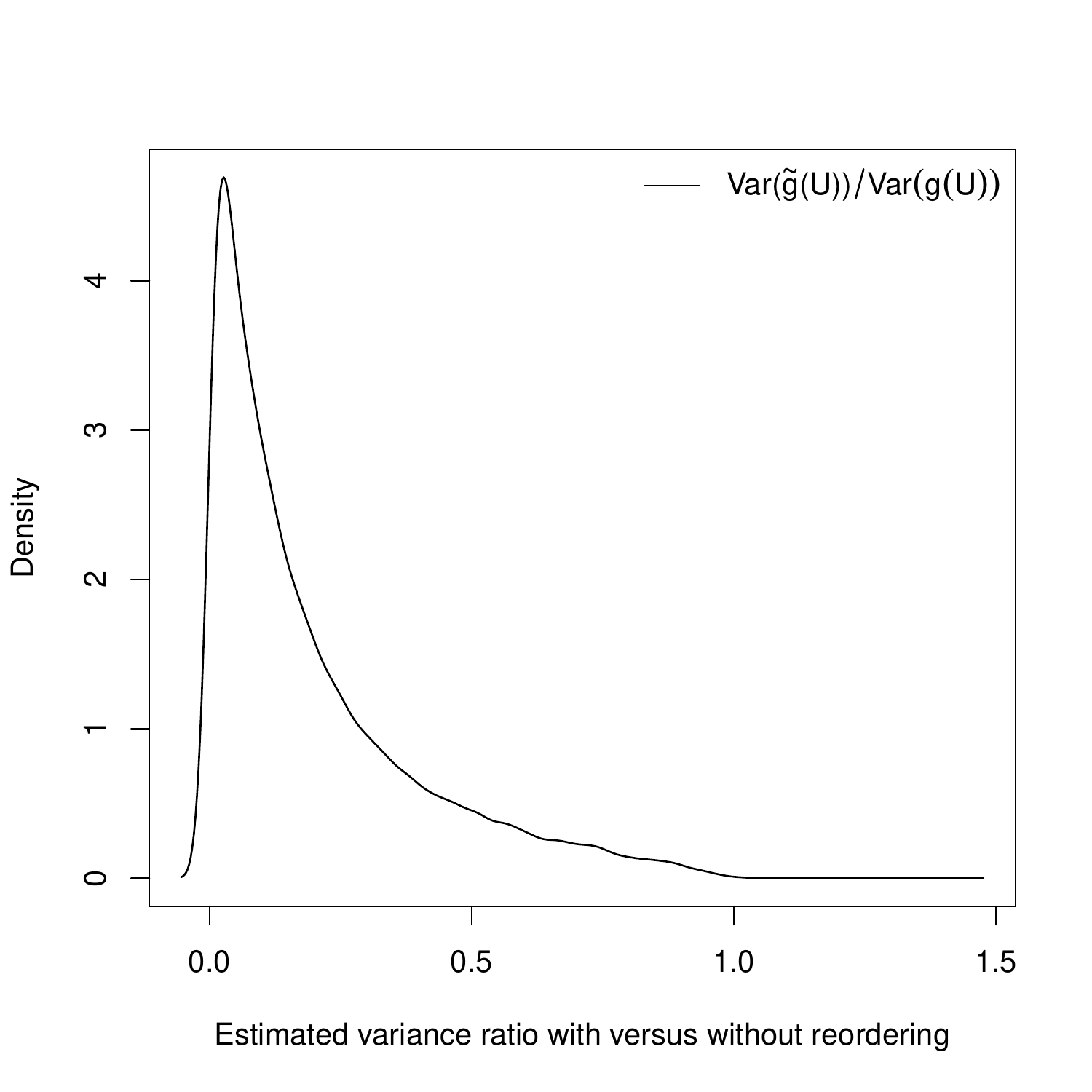}
\end{minipage}%
  \caption{Left: Variance of the integrand $\var(g(\bU))$ with and without variable reordering.
  Right: Density plot of estimated variance ratios.}
  \label{fig:vars}
\end{figure}

\paragraph{Investigating the variance of the integrand}
It is interesting to further investigate the effect of variable reordering as detailed in Section~\ref{subsec:precond}. To this end, the variance of the integrand $g$ from \eqref{eq:gi} given by
$$\var(g(\bU)) = \int_{[0,1]^d} g^2(\bu)\rd\bu - \left(\int_{[0,1]^d}g(\bu)\rd\bu\right)^2$$
is estimated, once with the original $g$ without reordering, and once with $\tilde{g}$ which is the integrand $g$ after applying Algorithm~\ref{alg:precond} to the inputs $\ba, \bb, \Sigma$. We use a randomized experiment and do the following 50\,000 times for an inverse-gamma mixture: Sample $d\sim\U(\{5,\dots,500\})$, $\nu\sim\U(0.1, 5)$ and $\ba$, $\bb$, $\Sigma$ are randomly chosen as in the previous section. The variance of the integrand is then estimated via the sample variance of $g(\bU_1),\dots,g(\bU_N)$ for $N=10\,000$. Results can be found in Figure~\ref{fig:vars}: On the left, variances have been ordered according to the ordering of the variances when variable reordering is employed (for better visibility of the reordering effect). On the right, a density plot of the ratios $\var(\tilde{g}(\bU))/\var(g(\bU))$ is shown.
It can be confirmed that in the vast majority of cases, variable reordering substantially decreases the variance of the integrand. In only 12 of the 50\,000 runs did the estimated variance after reordering exceed the variance without reordering.

\paragraph{Effective dimension of the integrand}
As was seen in Figures~\ref{fig:meanerrors.t} and~\ref{fig:meanerrors.Pareto}, reordering improves both MC and RQMC methods; the effect is however stronger for RQMC methods. A possible explanation for this is that the variable reordering not only reduces the overall variance of the integrand, $\sigma^2=\var(g(\bU))$, as seen in the previous part, but also the \emph{effective dimension} of the integrand, to be defined later. (R)QMC methods often work better if only a small number of variables are important, see \cite{wangfang2003} and references therein for a discussion and examples. Variable reordering, as explained in Section~\ref{subsec:precond}, was derived in a way such that the first components are the most important ones.

Sensitivity indices, such as Sobol' indices, can help understand the importance of different variables of an integrand. Following \cite[Ch. 6.3]{lemieux2009} and \cite{sobol2001}, we consider the ANOVA decomposition of a (square integrable) function $g:(0,1)^d\rightarrow\mathbb{R}$ given by
$$ g(\bu) = \sum_{I\subseteq\{1,\dots,d\}}g_I(\bu)$$
where
$$g_I(\bu) = \int_{[0,1]^{d-k}} g(\bu)\,\rd\bu_{-I} - \sum_{J\subset I}g_J(\bu),\quad g_{\emptyset}(\bu)= \int_{[0,1]^d} g(\bu)\ \rd\bu;$$
here, $k=|I|$ and $\bu_{-I}$ is the vector $\bu$ without components $k\in I$. The $g_I$'s only depend on variables $i\in I$ and are orthogonal; if $I\not=\emptyset$, $g_I$ has mean zero. The overall variance of the integrand can then be decomposed as $ \sigma^2 = \var(g(\bU)) = \sum_{I\subseteq\{1,\dots,d\}} \sigma_I^2$
where $\sigma_I^2 = \var(g_I(\bU)) = \int_{[0,1]^d}g_I(\bu)^2\,\rd\bu$. The number
$$S_I = \frac{\sigma_I^2}{\sigma^2}\in [0,1]$$
is called \emph{Sobol' index} of $I$. It explains the fraction of the overall variance of the integrand explained by the variables in $I$; if this number is close to 1, it means that most of the variance is explained by $g_I$ and therefore by the variables in $I$. If $I=\{l\}$ is a singleton, $S_I = S_{l}$ is called a \emph{first order index}.

Another useful sensitivity index is the \emph{total effect index} of variable $l\in\{1,\dots,d\}$ given by
$$ S_{T_l} = \frac{1}{\sigma^2} \sum_{I\subseteq\{1,\dots,d\}:l\in I} \sigma_{I}^2$$
which measures the relative impact of component $l$ and all its interactions. Care must be taken when interpreting this value as $\sum_{i=1}^d S_{T_i}\geq 1$ in general since interactions are counted several times. For instance, $\sigma_{\{1,2\}}^2$ is contained in $S_{T_1}$ as well as in $S_{T_2}$.

Finally, the \emph{effective dimension in the superposition sense} in proportion $p\in(0,1]$ is the smallest integer $d_S$ so that
$$ \frac{1}{\sigma^2} \sum_{I:|I|\leq d_S}\sigma_I^2 \geq p.$$
If the effective dimension is $d_S$, the integrand can be well approximated by functions of at most $d_S$ variables; see \cite[Sec. 3.6.1]{lemieux2009}.

The indices $S_{\{l\}}$ and $S_{T_l}$ for $l\in\{1,\dots,d\}$ can be estimated using \cite{owen2013}'s method which is implemented in the function \texttt{sobolowen()} in the \R\ package \texttt{sensitivity}; see \cite{sensitivity}. Figure~\ref{fig:sobolindices} shows estimated Sobol' indices in two settings: In each setting, $W\sim\text{IG}(1/2, 1/2)$ (so that $\bX$ follows a multivariate $t$ distribution with 1 degrees of freedom) and $d=10$. The upper limit $\bb$ and the scale matrix $\Sigma$ were found by trial \& error so that there is either a substantial variance reduction (top figure) achieved by reordering or an increase in variance (bottom figure). In order to be consistent with the definition of the integrand $g$ in~\eqref{eq:gi}, variables are called $0,\dots,d-1$ so that they correspond to $u_0,\dots, u_{d-1}$. For instance, in the top figure, one can read that $S_{(0)}\approx 0.52$ after reordering so that $52\%$ of the variance of $g$ can be explained by a function $g_{\{0\}}(u_0)$.

\begin{figure}[!t]
\centering
  \includegraphics[width=0.49\textwidth]{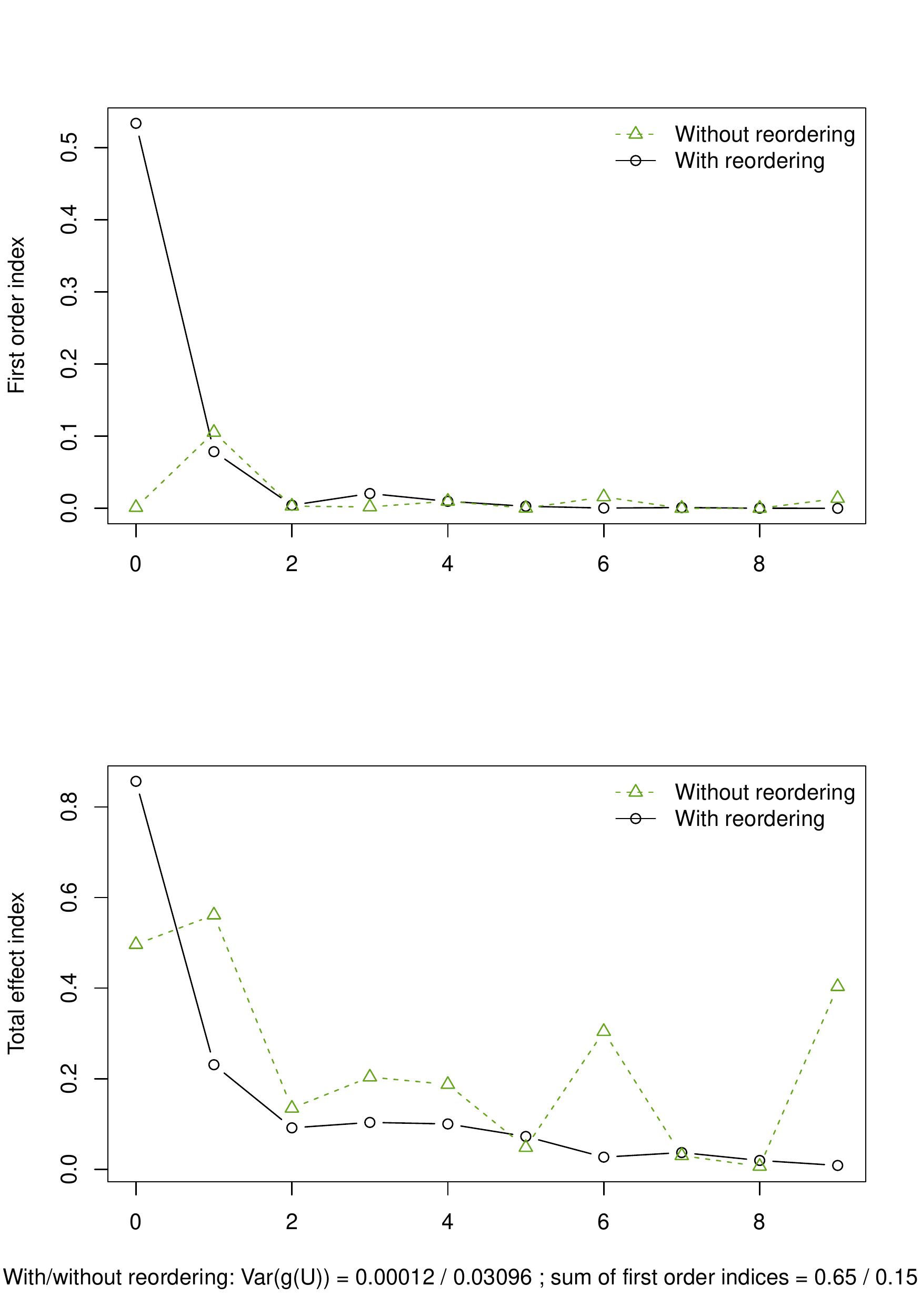}
  \includegraphics[width=0.49\textwidth]{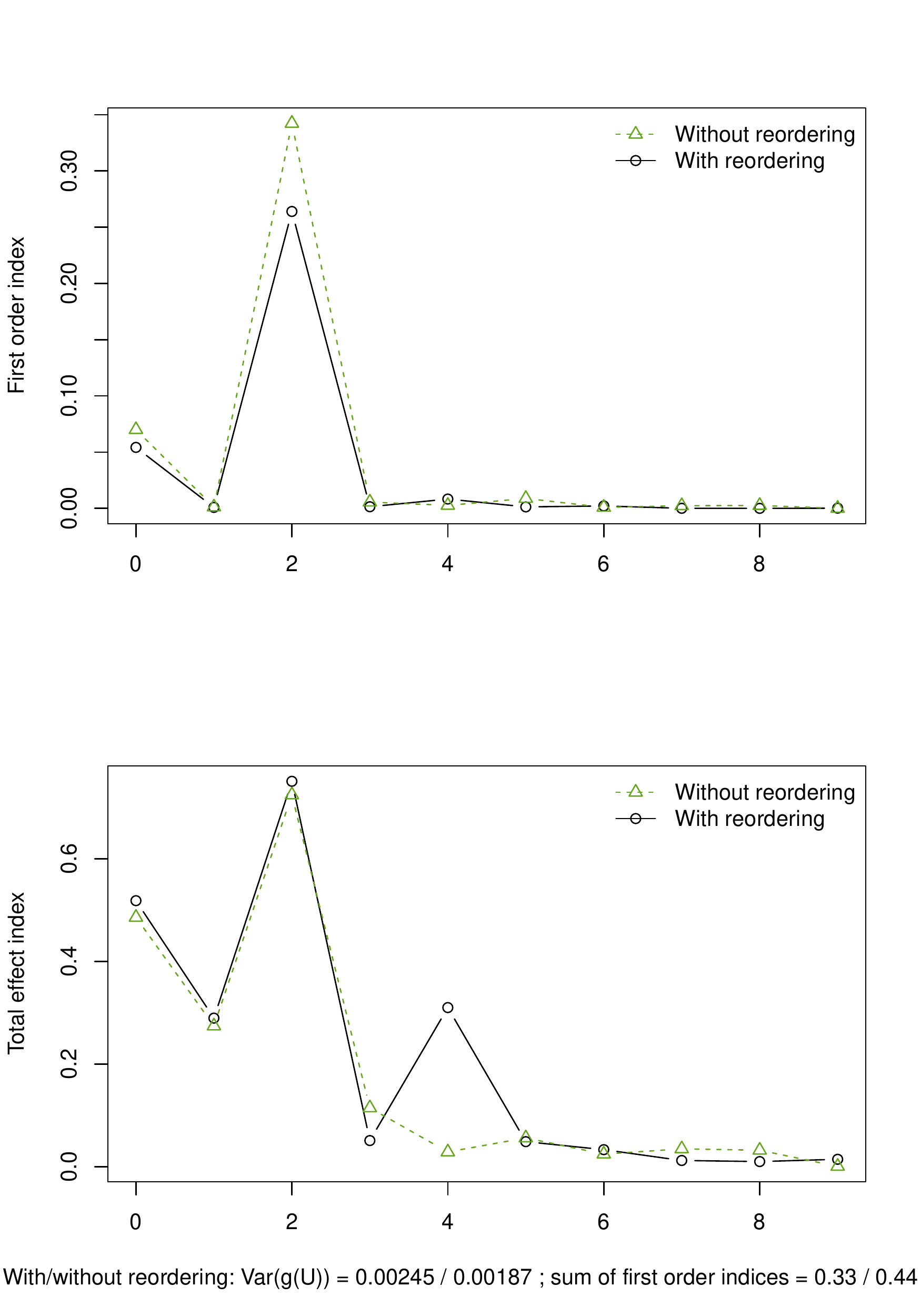}
  \caption{Estimated first order and total effect indices with and without reordering for an inverse-gamma
  mixture in a setting with high variance reduction (top) and increase in variance (bottom).}
  \label{fig:sobolindices}
\end{figure}

Inspecting the top figures where variable reordering led to a decrease in variance of approximately 99\% reveals that both first order and total effect indices are decreasing in the dimension after variable reordering was performed. Also, the figure label includes the sum of the first order indices. After reordering, 65\% (as opposed to 15\%) of the overall variance of the integrand is explained by components $g_{I}$ of $g$ of exactly one variable, hinting at the fact that the effective dimension decreased: The effective dimension in the superposition sense in proportion 65\% decreased to 1 after reordering.

There are rare cases when variable reordering leads to an increase in variance: In the bottom figures, the relative increase is about 31\%. Here, the new ordering is clearly not optimal and indices are not decreasing with the dimension. Given the nature of the greedy procedure it is expected that in some cases, no improvement is achieved.

\noindent
\subsubsection{Run times}
In this part we take a brief look at the run-times of Algorithm~\ref{alg:RQMC:F:a:b} combined with the variable reordering Algorithm~\ref{alg:precond}. We restrict our attention to the important multivariate $t$ case and compare run times of our implementation in \texttt{pStudent()} with the run times of the above mentioned QRSVN algorithm described in \cite{genzbretz2002} and provided by the function \texttt{pmvt()} in the \R\ package \texttt{mvtnorm}.

In order to get meaningful estimates of the
CPU time, for each dimension $d$, the following is done 15 times: Sample
$\bb$ and $\Sigma$ as before when estimating $\var(g(\bU))$, set $\ba=(-\infty,\dots,-\infty)$ and $\nu=2$. Then call
$\texttt{pmvt()}$ and $\texttt{pStudent()}$ three times each and average their CPU
times obtained using the package \texttt{microbenchmark} of \cite{mersmann2015}. The above procedure is done for an absolute error tolerance
$\eps=0.001$ and the maximum number of function evaluations is chosen such that both algorithms always terminate with the correct precision.

Figure
\ref{fig:times_all} shows the run times obtained. The symbols represent the
corresponding means whereas the lines show the largest/smallest CPU time measured for that dimension. Note that \texttt{pmvt()} only works for dimensions up to 1\,000. Figure
\ref{fig:times_all} shows that our implementation significantly outperforms the existing standard which takes up to 8 times more run time.

\begin{figure}[!t]
\centering
  \includegraphics[width=0.45\textwidth]{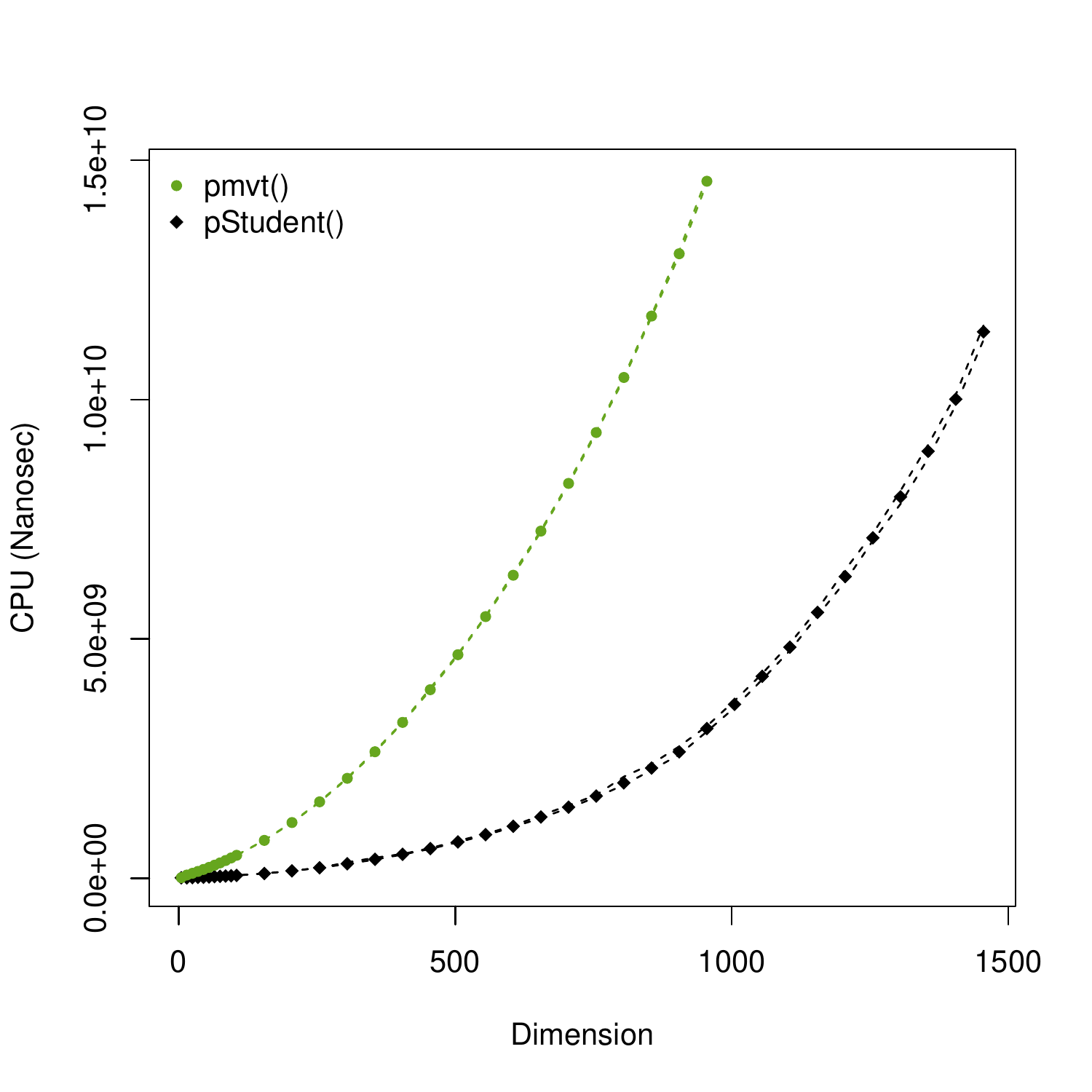}
  \includegraphics[width=0.45\textwidth]{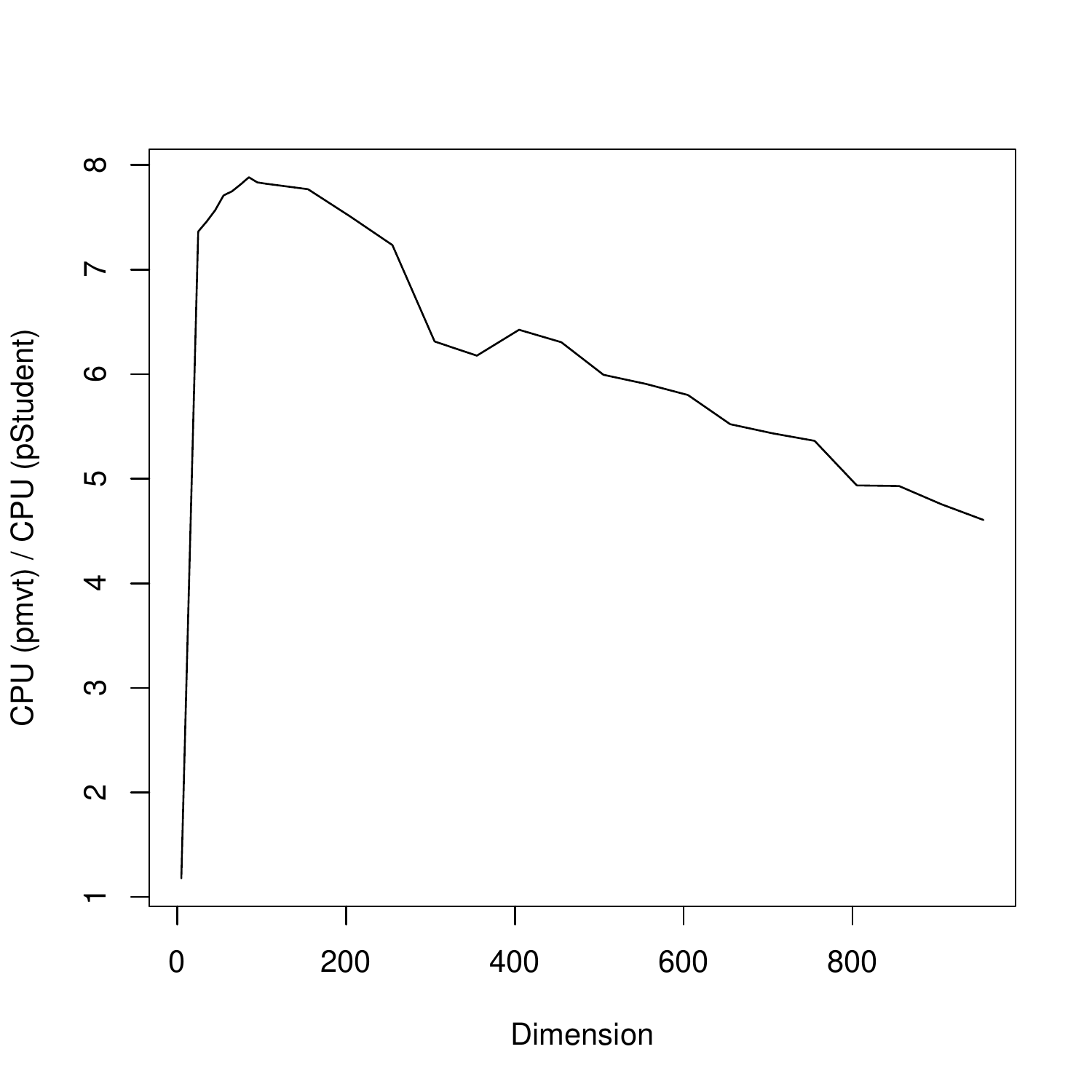}
  \caption{Run times based on three replications of 15 randomly chosen inputs $\bb$ and $\Sigma$ in each dimension (left); run-time ratios relative to \texttt{pStudent()} (right).}
  \label{fig:times_all}
\end{figure}

\subsection{Estimating the density function}

In this section we test the performance of Algorithm~\ref{alg:RQMC:f:x:adaptive} to estimate the log-density of $\bX\sim \MVT_d(\nu,\bmu,\Sigma)$ and $\bX\sim \PNVM_d(\alpha,\bmu,\Sigma)$. Note that the density is known in either case and given in~\eqref{eq:dens:mvt} and~\eqref{eq:densityX:Paretomix} so that estimated and true log-density values can be compared.

We sample $n=1\,000$ points from $\bX\sim \text{MVT}_d(\nu=1,\bzero,I_d)$ in dimension $d=10$ and evaluate the density of $\text{MVT}_d(\nu=4,\bzero,I_d)$ at the sampled points. The Pareto case is done similarly. Figure~\ref{fig:dnvmix} displays results obtained by the adaptive algorithm (Algorithm~\ref{alg:RQMC:f:x:adaptive}) and by the crude (non-adaptive) Algorithm
\ref{alg:RQMC:logmu}; %
the true log-density and the probability $\P(D^2(\bX, \bzero, I_d)> m^2)$ are also plotted. The latter probability gives an idea of how likely it is to see a sample point $\bx$ with Mahalanobis distance greater than $m$.
For small Mahalanobis distances, both algorithms perform well. For larger ones the problem becomes harder as the underlying integrand becomes more difficult to integrate (recall Figure~\ref{fig:density_integrand} and the discussion thereafter) and the crude, non-adaptive version gives highly biased results. The adaptive version, however, is able to accurately estimate the log-density for any Mahalanobis distance and is furthermore much faster (it takes only approximately 1~second for a total of $n=1\,000$ log-density estimations).

By inspecting the axes in Figure~\ref{fig:dnvmix}, one can see that our procedure performs well even for very large Mahalanobis distances that would rarely been observed. For likelihood-based methods, such as Algorithm~\ref{alg:fitnvmix}, it is, however, crucial to be able to evaluate the density function for a wide range of inputs. For instance, consider the problem where a sample $\bX_1,\dots,\bX_n\isim\MVT_d(\nu,\bzero,I_d)$ for unknown $\nu$ is given. It is then necessary to evaluate the log-density of $\bX_1,\dots,\bX_n$ at a range of values of $\nu$ in order to find the maximum likelihood estimator. In fact, this was the motivation for performing the experiments undertaken to produce Figure~\ref{fig:dnvmix}: The sample is coming from a heavy-tailed multivariate $t$ distribution and the log-density function of a less heavy tailed multivariate $t$ distribution is evaluated at that sample. The same intuition lies behind the experiment to produce the plot on the right of Figure~\ref{fig:dnvmix}.

\begin{figure}[!h]
\centering
  \includegraphics[width=0.45\textwidth]{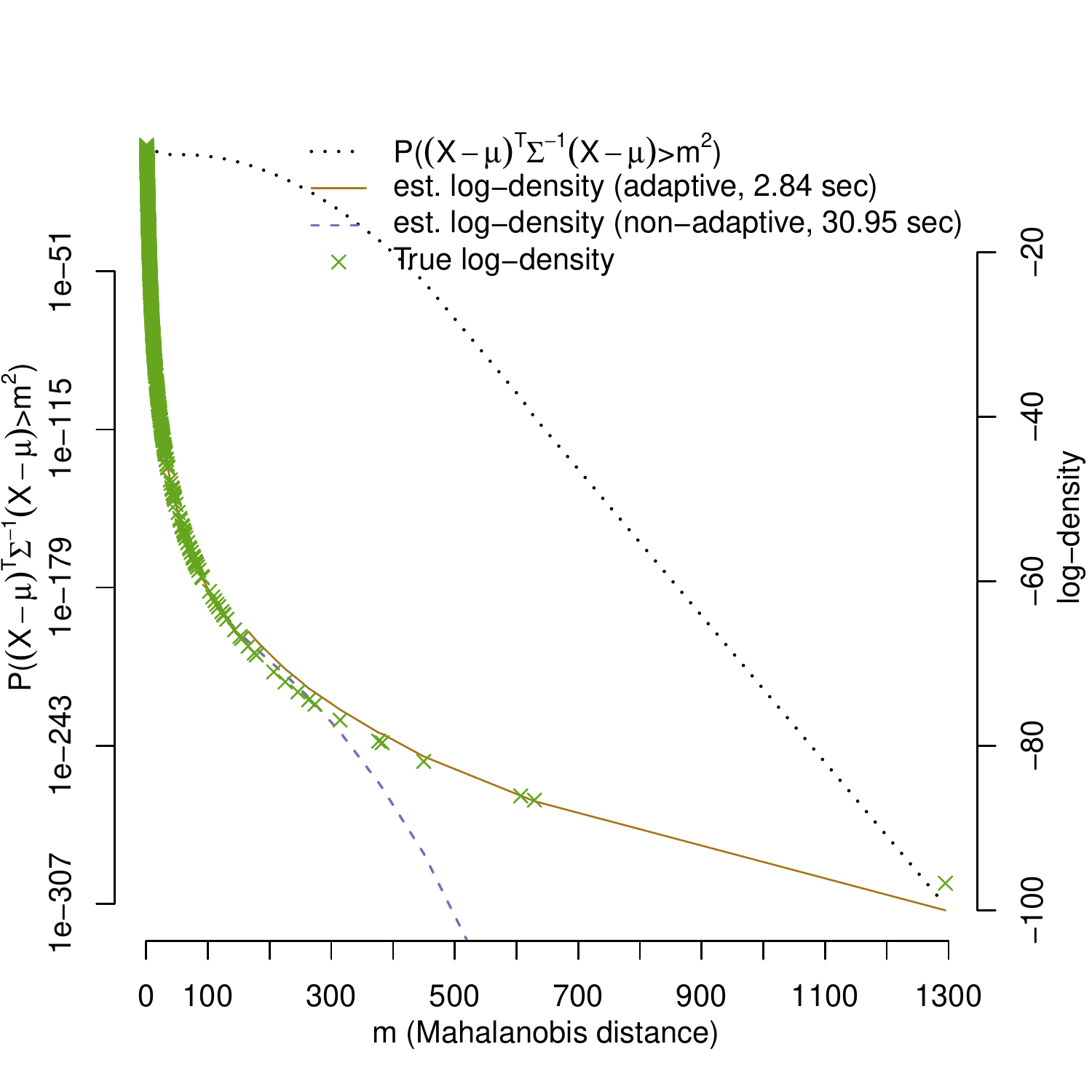}
  \includegraphics[width=0.45\textwidth]{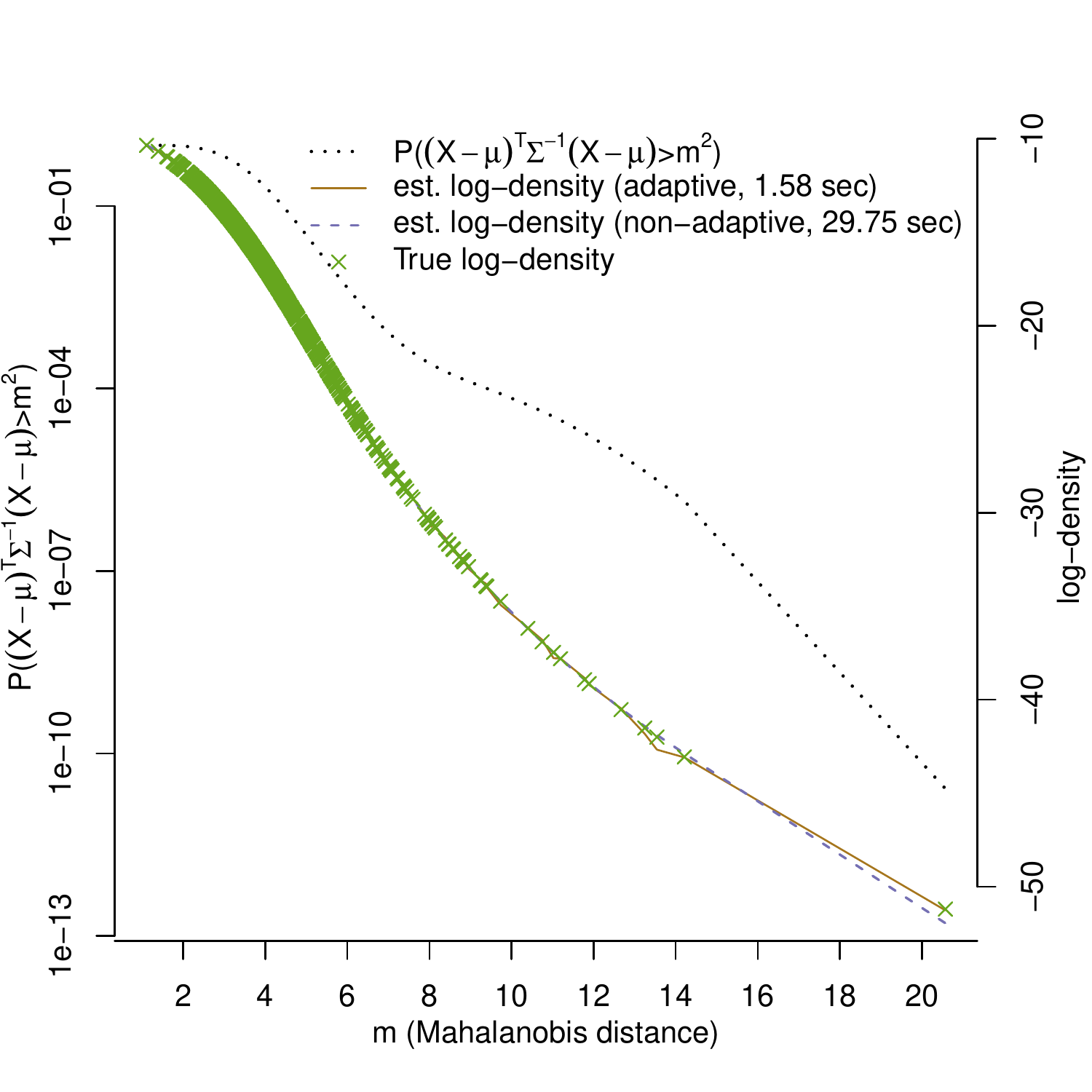}
  \caption{Estimated log-density of $\MVT_d(\nu=4, \bzero, I_d)$ (left) and $\PNVM_d(\alpha=6, \bzero, I_d)$ (right) in $d=10$ evaluated at $n=1\,000$ points sampled from $\MVT_d(\nu=1, \bzero, I_d)$ (left) and $\PNVM(\alpha=2, \bzero, I_d)$ (right).}
  \label{fig:dnvmix}
\end{figure}

\subsection{Fitting normal variance mixture distributions}

In this section we provide examples for our fitting procedure Algorithm~\ref{alg:fitnvmix}. While in the special case where $W$ follows an inverse-gamma distribution (i.e., $\bX\sim\MVT_d(\nu, \bmu, \Sigma)$ for which the joint density function is available in closed form), ECME methods described in \cite{liurubin1995} and \cite{nadarajahkotz2008} can be applied directly (implemented, for instance, in the function \texttt{fit.mst()} in the \R\ package \texttt{QRM}; see \cite{pfaffmcneil2016}), this is not the case for a general normal variance mixture distribution where the density function may not be available in closed form. In the latter case, we do rely on Algorithm~\ref{alg:fitnvmix} in combination with our adaptive procedure described in Algorithm~\ref{alg:RQMC:f:x:adaptive} to estimate the log-density function. This is all done automatically in the function \texttt{fitnvmix()} which merely needs a specification of the mixing distribution in the form of its quantile function.

As in the previous section, we consider an inverse-gamma and a Pareto mixture as test cases. We chose these two distributions where the density function is known in closed form so that we are able to investigate if
optimizing the log-likelihood estimated via Algorithm~\ref{alg:RQMC:f:x:adaptive} (as opposed to using a closed formula for the log-likelihood) has a significant effect on parameter estimates. In a practical setting where the density function is not known in closed form (as is the case for the inverse-Burr mixture considered in the data analysis) such comparison is not possible.

\begin{figure}[!t]
\centering
\includegraphics[width = 0.45\textwidth]{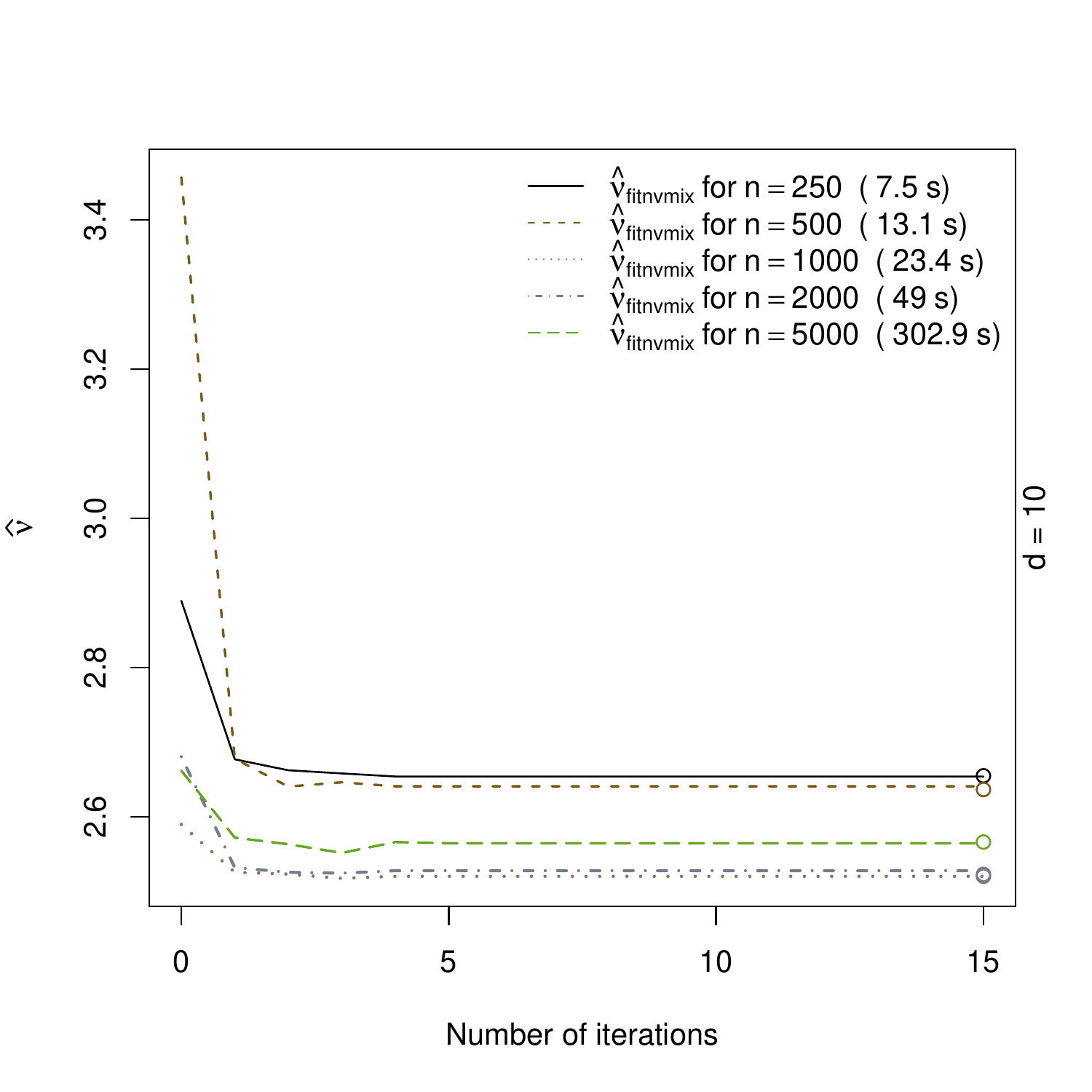}
\includegraphics[width = 0.45\textwidth]{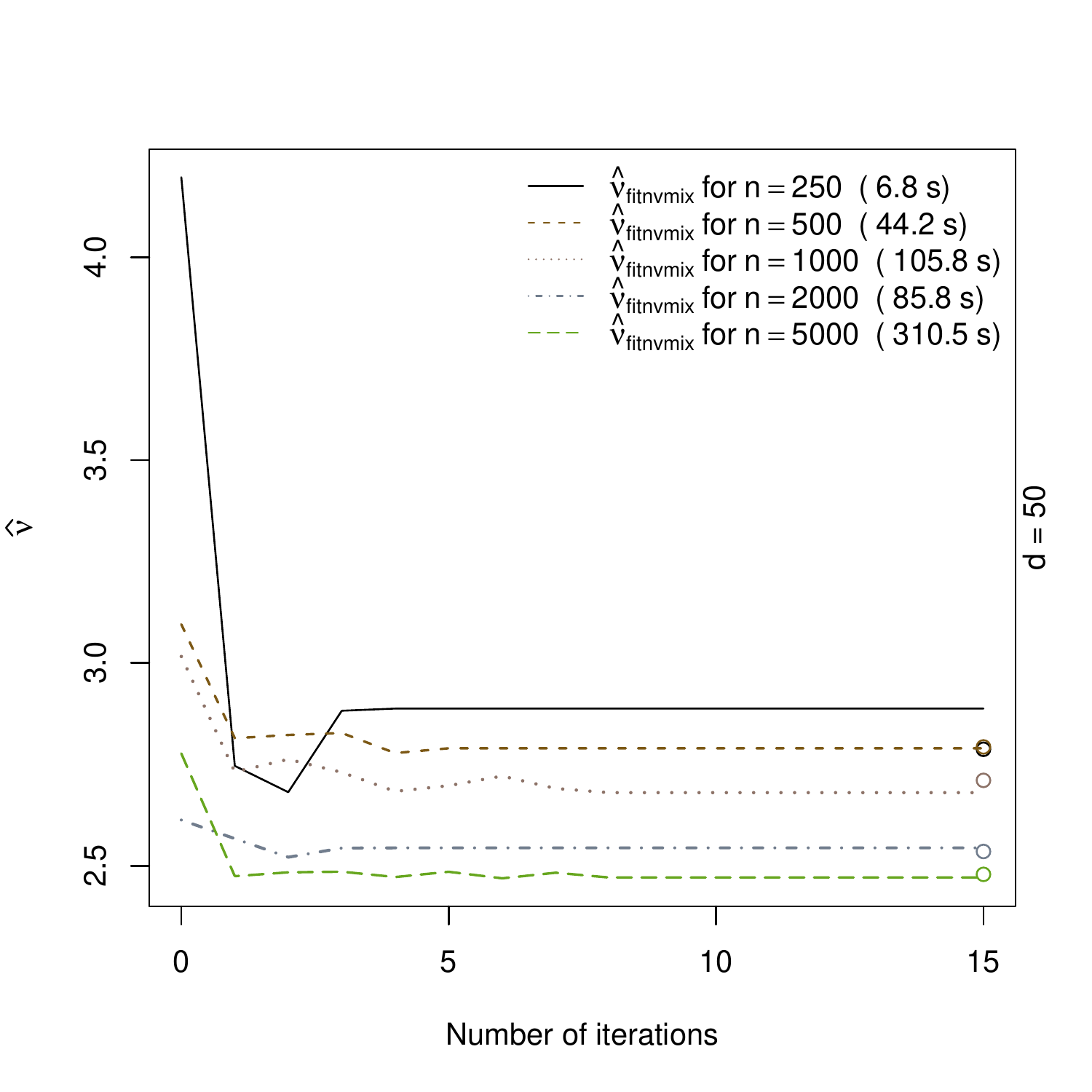}
\caption{Estimates $\hat{\nu}$ computed by Algorithm~\ref{alg:fitnvmix} as a function of the number of ECME iterations for multivariate $t$ distributions of different sample sizes and dimensions. The symbols at the end of each curve denote the maximum likelihood estimator of $\nu$ as found by the ECME algorithm with analytical weights and densities.}
\label{fig:fitnvmix.t}
\includegraphics[width = 0.45\textwidth]{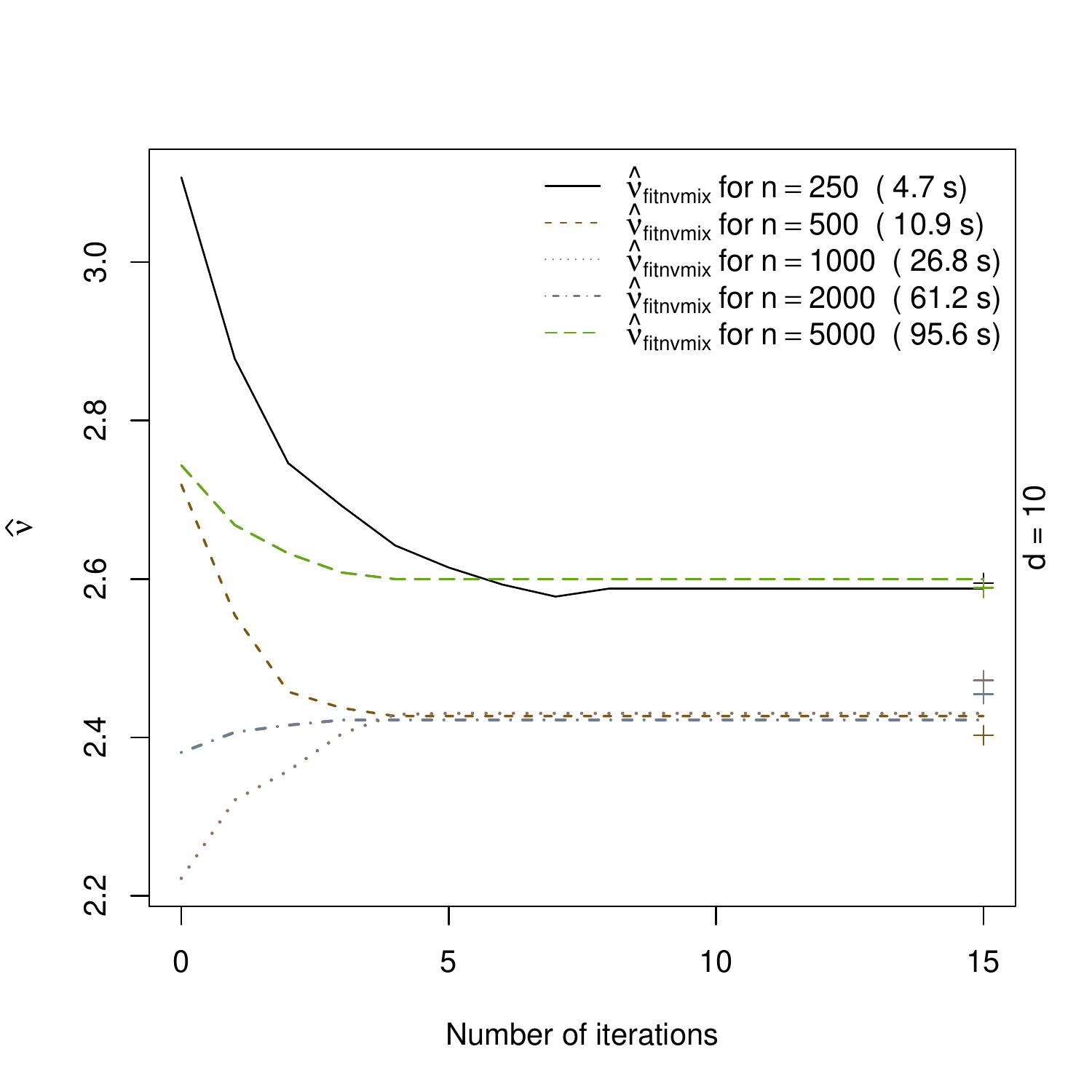}
\includegraphics[width = 0.45\textwidth]{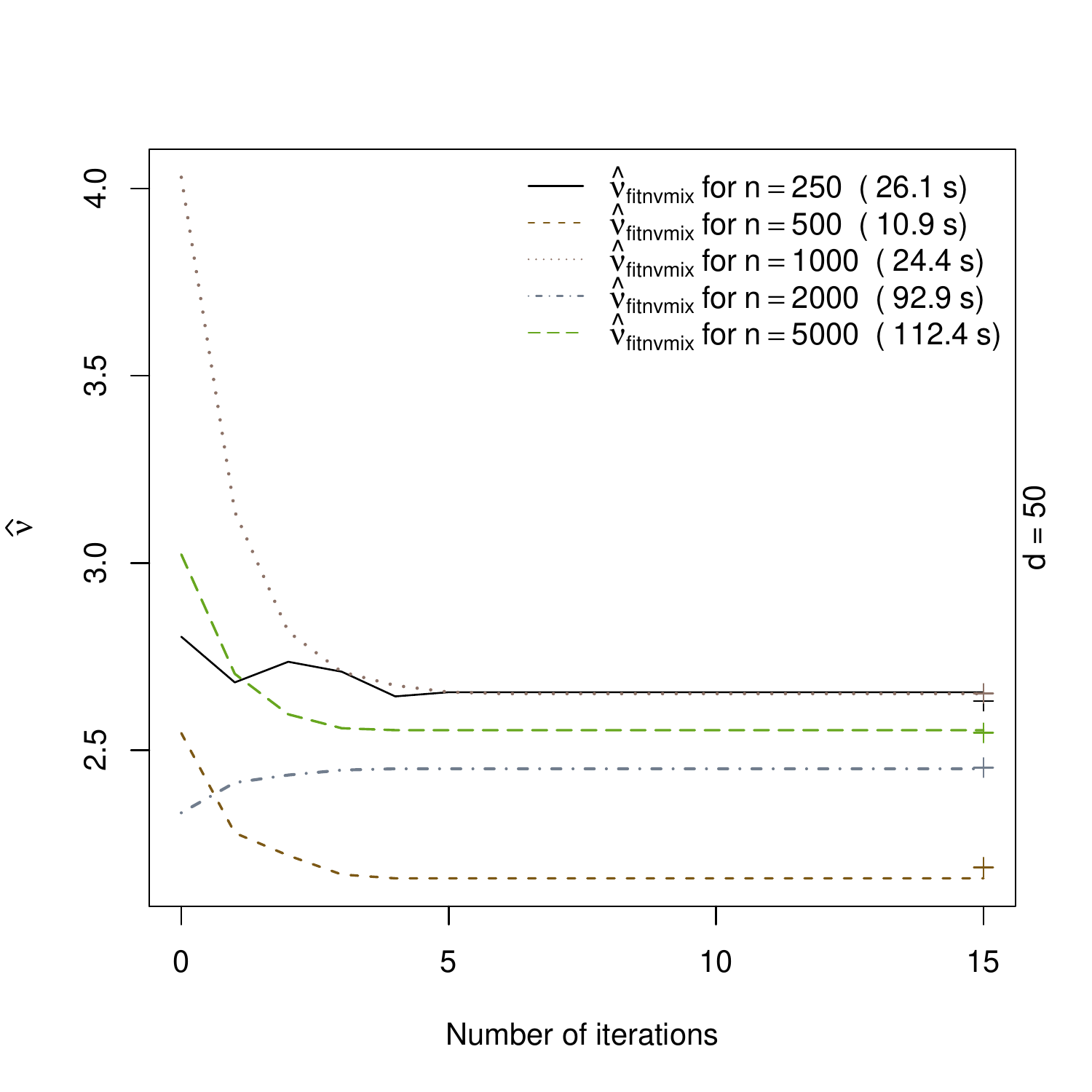}
\caption{Estimates $\hat{\nu}$ computed by Algorithm~\ref{alg:fitnvmix} as a function of the number of ECME iterations for Pareto mixture distributions of different sample sizes and dimensions. The symbols at the end of each curve denote the maximum likelihood estimator of $\nu$ as found by the ECME algorithm with analytical weights and densities.}
\label{fig:fitnvmix.Pareto}
\end{figure}

Our algorithm is tested in dimensions $d\in\{10, 50\}$ for sample sizes $n$ between 250 and 5\,000.  In each setting, $n$ random vectors $\bX_1,\dots,\bX_n\isim\MVT_d(\nu=2.5, \bzero, \Sigma)$ are sampled and then Algorithm~\ref{alg:fitnvmix} is used to estimate the parameters. We randomly choose $\Sigma$ as $D R D$ where $R$ is a random Wishart matrix and $D$ is diagonal with entries $D_{ii}\isim\U(2,5)$ for $i=1,\dots,d$. Results are displayed in Figure~\ref{fig:fitnvmix.t} where the estimate $\hat{\nu}$ of $\nu$ is plotted as a function of the number of ECME iterations (see Step~\ref{step:ecmeiteration} of Algorithm~\ref{alg:fitnvmix}). The optimizations in Steps~\ref{step:startingvalue} and~\ref{step:updatenu} of Algorithm~\ref{alg:fitnvmix} are based on the estimated log-likelihood function via Algorithm~\ref{alg:RQMC:f:x:adaptive}.

As mentioned earlier, an ECME procedure for estimating parameters of a multivariate $t$ distribution is available in the function \texttt{fit.mst()}. The symbols at the end of the curves in Figure~\ref{fig:fitnvmix.t} denote estimates obtained from this function. It can be confirmed that not only does our procedure converge to the correct maximum likelihood estimate in the given examples, but also that run times are reasonably small for this challenging problem. Note that only few iterations are needed until convergence is detected.

A similar experiment is performed for the Pareto-mixture case, see Figure~\ref{fig:fitnvmix.Pareto}. Here, the symbols at the end of each line display results obtained from Algorithm~\ref{alg:fitnvmix} using analytical weights and densities, obtained by calling our function \texttt{fitnvmix()} with \texttt{qmix = "pareto"}.

The run times displayed in Figures~\ref{fig:fitnvmix.t} and~\ref{fig:fitnvmix.Pareto} may seem counter-intuitive; however, several factors influence run time: The larger the sample size $n$, the more integrals need to be approximated and the higher the probability of observing extreme Mahalanobis distances. Furthermore, the problem of estimating the log-density and the weights becomes harder the larger the Mahalanobis distance of the input. However, larger sample sizes can also lead to a quicker convergence of the weights in Step~\ref{step:updatemusigma} of Algorithm~\ref{alg:fitnvmix} and also to faster convergence of the estimates of the mixing variable in Step~\ref{step:updatenu} of Algorithm~\ref{alg:fitnvmix}.
Overall, as there are numerical approximations involved at many levels, it will depend on the sample at hand how long the algorithm takes. This explains why run times are not monotone in the sample size $n$.

\subsection{Example application}
This section demonstrates an application of all our methods presented to a real financial data set. We consider daily return data  of 5 constituents of the SP500 index between 2007-01-03 and 2009-12-31 ($n=755$ data points in $d=5$). The dataset \texttt{SP500} is obtained from the \R\ package \texttt{qrmdata}, see \cite{hoferthornik2016}, and the stocks considered are \texttt{AAPL} (Apple), \texttt{ADBE} (Adobe), \texttt{INTC} (Intel), \texttt{ORCL} (Oracle) and \texttt{GOOGL} (Google). We first fit marginal $\operatorname{ARMA}(1,1)-\operatorname{GARCH}(1,1)$ models and then fit normal variance mixture models to the standardized residuals (``innovations").

Four normal variance mixture models are considered: The multivariate $t$ (an inverse-gamma mixture), a Pareto-mixture, an inverse-Burr mixture and the multivariate normal, where $\bX$ follows an inverse-Burr mixture if $F_W^\i(u, \bnu)= (u^{-1/\nu_2}-1)^{-1/\nu_1}$ (which is the quantile function of $1/\tilde{W}$ where $\tilde{W}\sim\operatorname{Burr}(\nu_1,\nu_2)$ has distribution function $F_{\tilde{W}}(\tilde{w})=1-(1+\tilde{w}^{\nu_1})^{-\nu_2}$ for $\tilde{w}>0$ and $\nu_1,\nu_2>0$). We highlight that in the inverse-Burr mixture case, neither the density of the resulting mixture nor weights for our estimation procedure are available in closed form, so that in this case, we indeed rely on our adaptive estimation procedure Algorithm~\ref{alg:RQMC:f:x:adaptive} to estimate the log-density function. As such, we supply aforementioned quantile function as a ``black box'' to our fitting procedure via \code{fitnvmix(, qmix = function(u, nu) (u^(-1/nu[2])-1)^(-1/nu[1]))}. We re\-mark that the multivariate normal case is trivial from an estimation point of view, as the maximum likelihood estimators for
$\bmu$ and $\Sigma$ are merely the sample mean and the sample variance,
respectively; this case is included for the sake of comparison.

We fit the aforementioned distributions to the stock data using Algorithm~\ref{alg:fitnvmix}. For the inverse-gamma and Pareto-mixtures we find $\hat{\nu}=5.65$ and $\hat{\nu}=1.64$, respectively, when using the closed form densities and weights; if weights and densities are estimated, we found $\hat{\nu}=5.62$ (20 sec) and $\hat{\nu}=1.61$ (13 sec), respectively. Overall it is reassuring that the estimates obtained from analytical and estimated weights and densities only differ slightly; given the difficulty of the problem the run times also seem reasonable. For the inverse-Burr mixture, we found $\hat{\bnu}=(2.15, 3.61)$ after 30 seconds run-time.

\begin{figure}[!t]
\centering
\includegraphics[width = 0.24\textwidth]{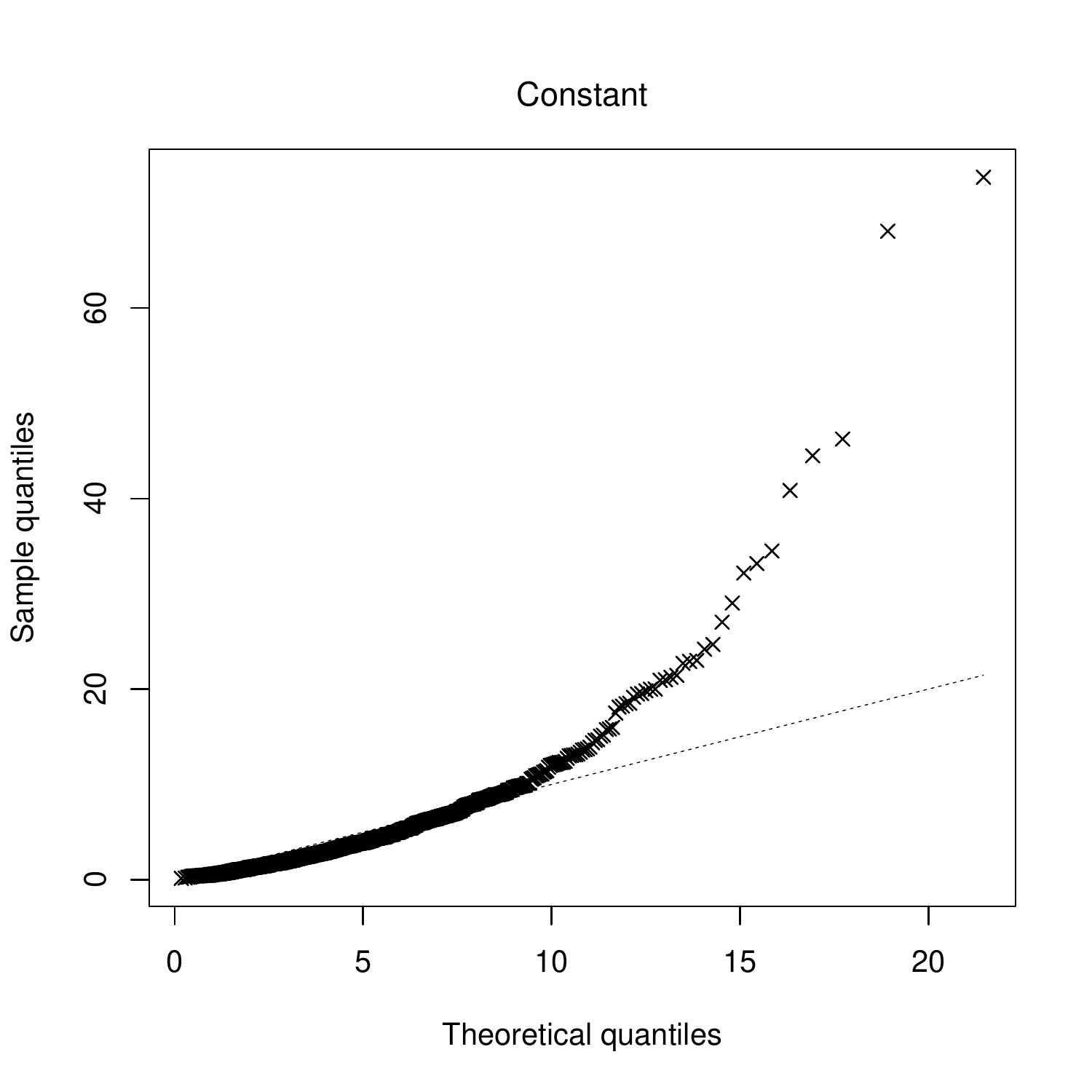}
\includegraphics[width = 0.24\textwidth]{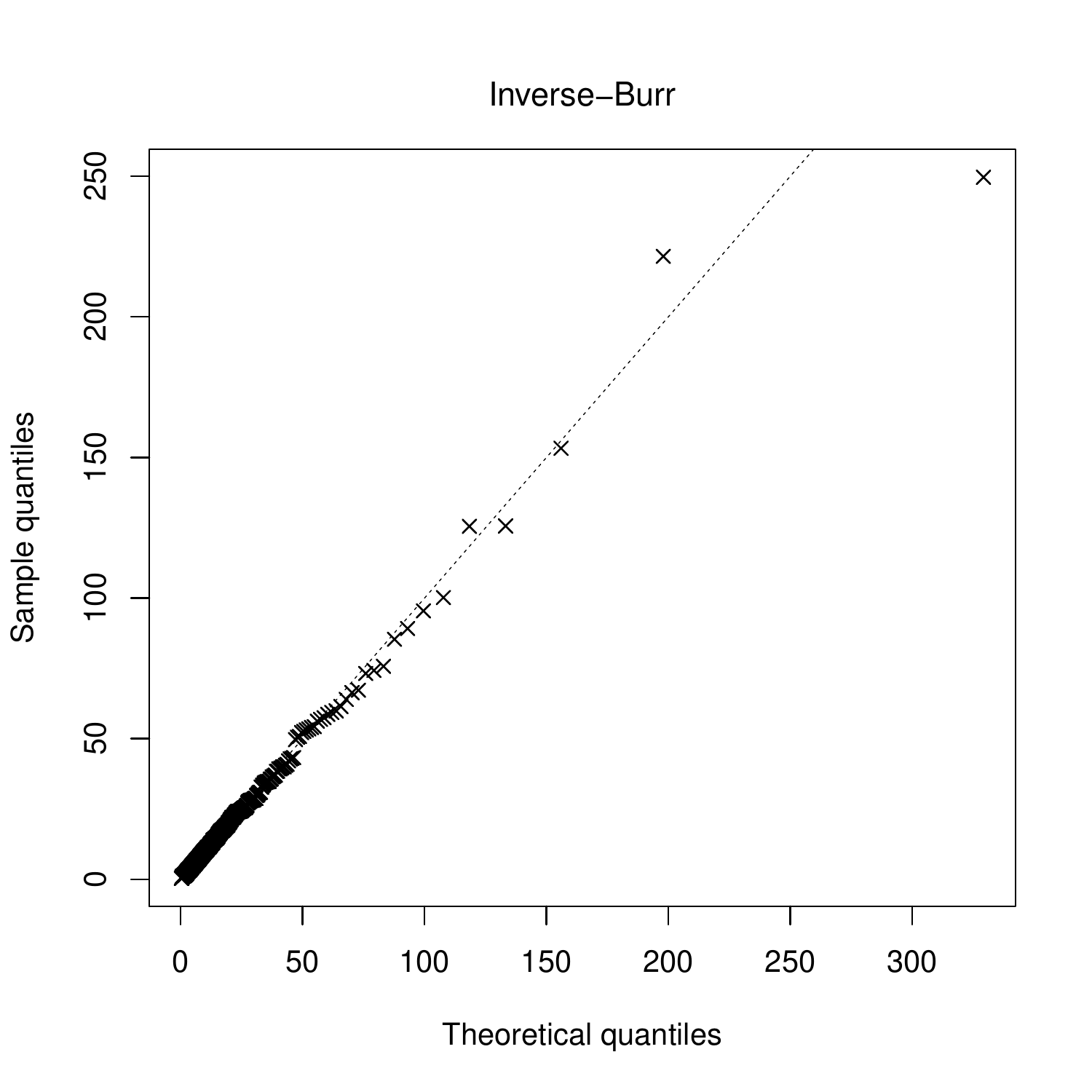}
\includegraphics[width = 0.24\textwidth]{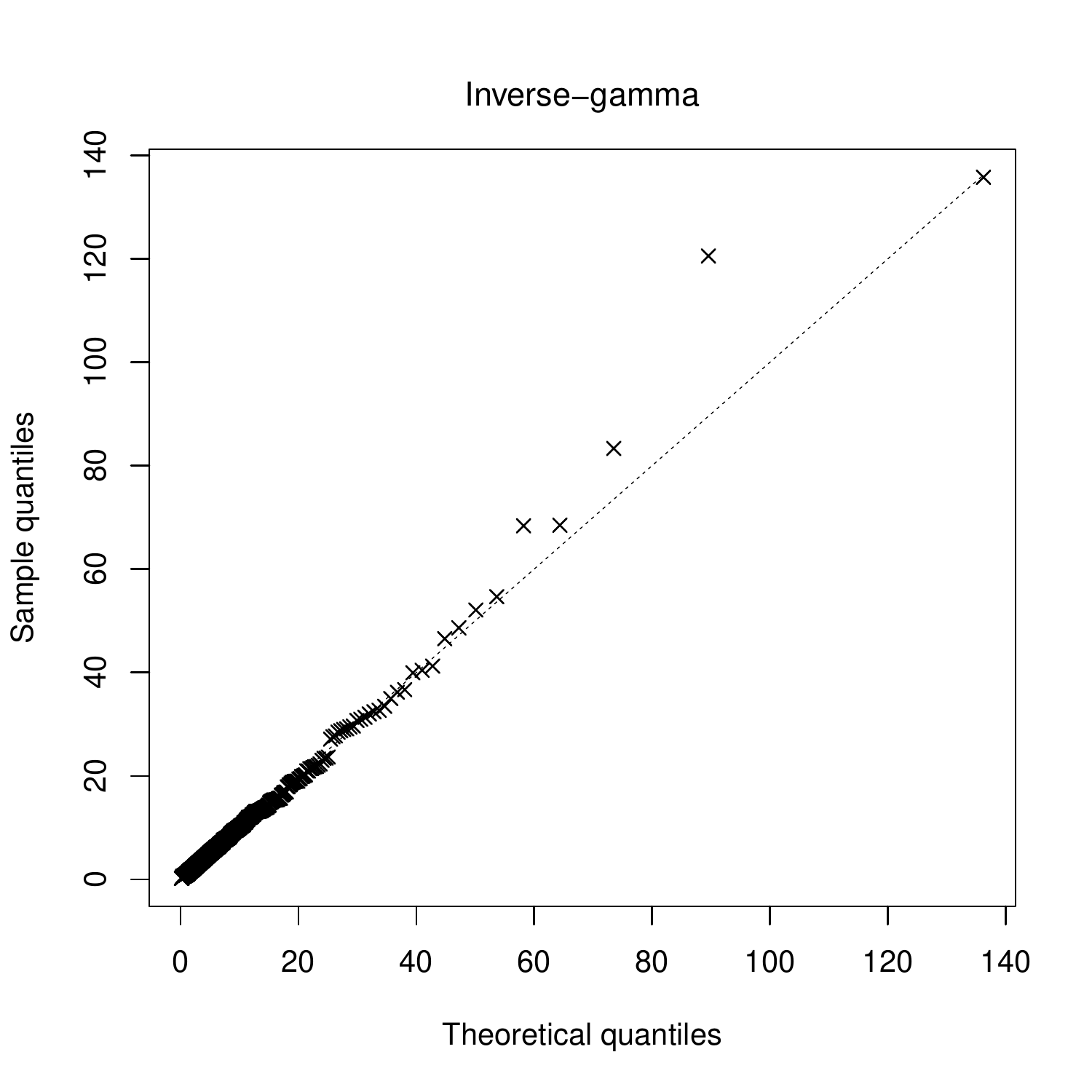}
\includegraphics[width = 0.24\textwidth]{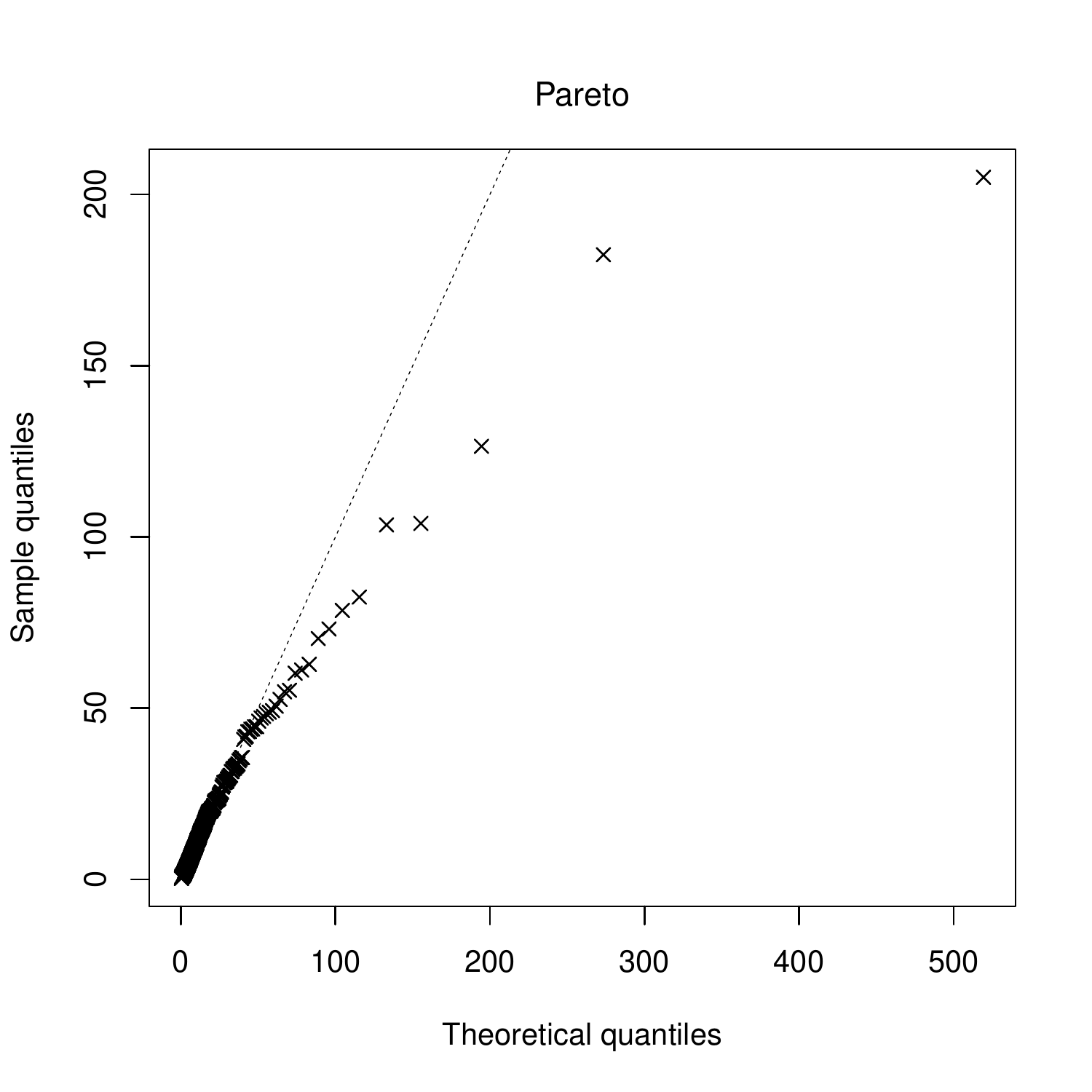}
\caption{Q-Q Plots of the empirical quantiles of the Mahalanobis distances $D^2(\bx_i, \hat{\bmu}, \hat{\Sigma})$, $i=1,\dots,n$, versus their theoretical quantiles for different models using a 5 stock portfolio with data from the \texttt{SP500} data set.}
\label{fig:qqplotssp500}
\end{figure}

Figure~\ref{fig:qqplotssp500} displays Q-Q Plots of $D^2(\bx_i,\hat{\bmu}, \hat{\Sigma})$, $i=1,\dots,n$ as a graphical goodness-of-fit test. Theoretical quantiles are estimated using the methods described in Section~\ref{sec:gammamix}. Clearly, the multivariate normal distribution (corresponding to constant $W$) provides a poor fit to the data as the tail is heavily underestimated. Both the inverse-gamma mixture and the inverse-Burr mixture provide an excellent fit to the data; the Pareto-mixture however shows too heavy tails. These plots confirm our main motivation outlined in the introduction: The multivariate normal is poorly suited for heavy-tailed return-data; normal variance mixtures, however, are more flexible in that they allow for heavier joint tails, often giving a better fit.

Finally, we use Algorithm~\ref{alg:RQMC:F:a:b} to estimate the joint quantile shortfall probability
$$ Q(u) := \P( X_1 \leq F_{X_1}^\i(u),\dots, X_d \leq F_{X_d}^\i(u))$$
for $u\in(0,1)$. In our context this is the probability that each of the 5 stocks yields a return smaller than its respective $u$ quantile; for small $u$, $Q(u)$ is the probability of a joint large loss and a rare event. This quantity is often considered in risk management to quantify the risk associated with joint extreme events. Since the margins are continuous, $Q(u)$ is the underlying copula evaluated at $(u,\dots,u)$.
In Figure~\ref{fig:shortfallprobs} we plot the estimated quantile shortfall probability $Q(u)$ for a range of values of $u$ for each fitted model separately. The figure on the right-hand-side shows the same probabilities $Q(u)$ standardized by the corresponding normal probability. The plots show again that the Pareto-mixture is significantly more heavy tailed than the multivariate $t$ distribution: It yields significantly higher shortfall probabilities. Furthermore these plots exemplify that our Algorithm~\ref{alg:RQMC:F:a:b} is also capable of estimating small probabilities despite the increasing numerical difficulty when moving outwards in the joint tail.

\begin{figure}[!t]
\centering
\includegraphics[width = 0.45\textwidth]{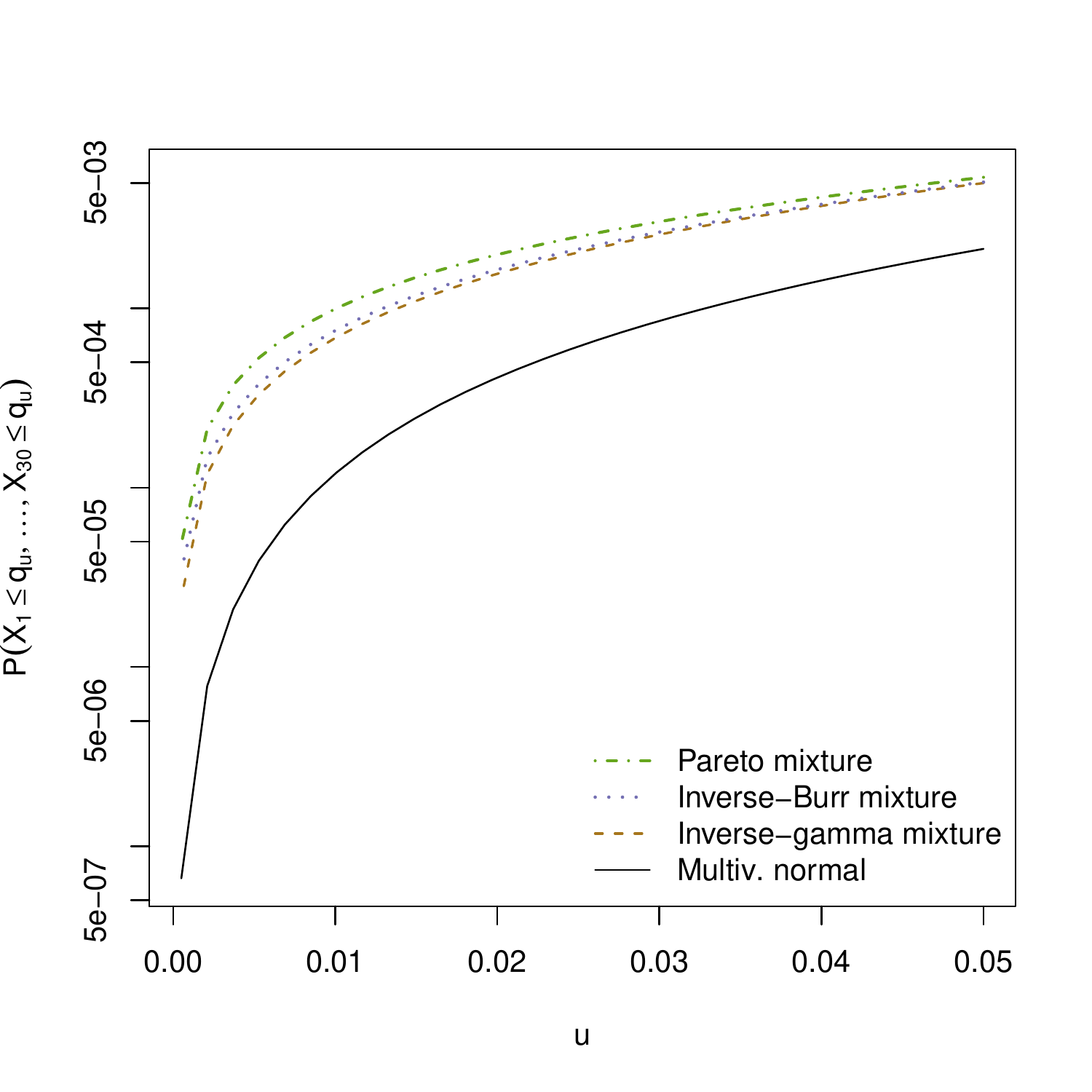}
\includegraphics[width = 0.45\textwidth]{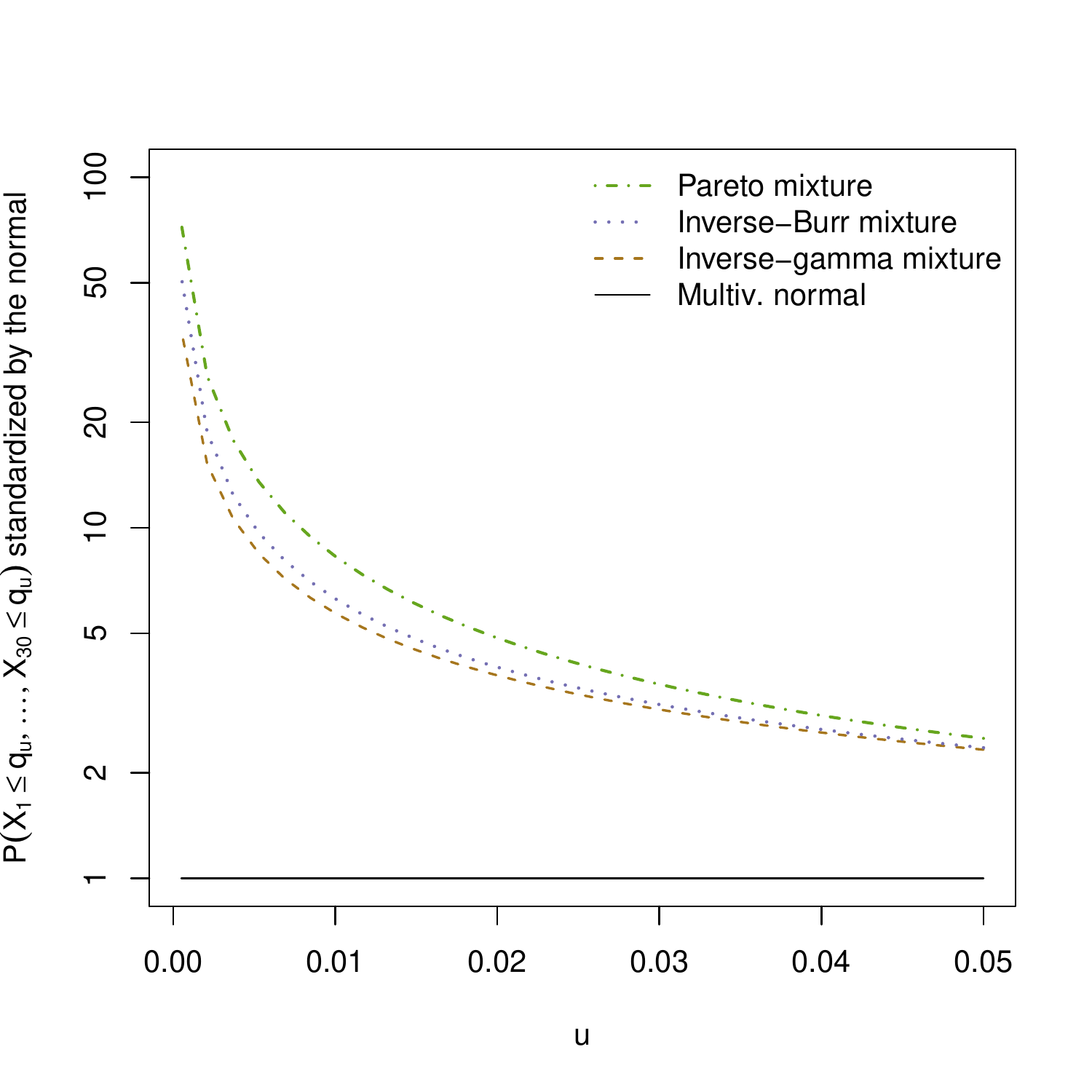}
  \caption{Estimated shortfall probabilities for different models for a 5 stock portfolio with data from the \texttt{SP500} data set (left);
  same probabilities standardized by the normal case (right).}
  \label{fig:shortfallprobs}
\end{figure}

\section{Conclusion}
We introduced efficient algorithms to perform the four main tasks for multivariate normal variance mixtures: Estimating the distribution function, the log-density function, sampling and estimating parameters for a given data set when only the quantile function of the mixing variable $W$ is available.
Due to the importance of multivariate normal variance mixtures for disciplines such as actuarial science or quantitative risk management, these algorithms are also widely applicable in practice.

We saw that the distribution function and the log-density function of normal variance mixtures can be accurately and quickly estimated even in high dimensions using RQMC algorithms. The algorithm for the distribution function relies on a generalization of methods that were used for estimating multivariate normal and $t$ probabilities in the past, including an efficient variable reordering algorithm. The algorithm for the log-density is based on an adaptive RQMC procedure that samples only in important regions. We also saw that it is possible to fit multivariate normal variance mixtures using an ECME algorithm in such generality where all involved quantities such as log-densities and weights need to be estimated via RQMC methods. Numerical results validate our methods. An implementation of all methods is provided in the \R\ package \texttt{nvmix}.

We remark that our work also exemplifies the superiority of RQMC methods even in very high dimensions over MC methods for this class of problems.

Another application of our methods is related to normal variance mixture copulas, the implicit copulas derived from normal variance mixture distributions. These copulas can be used to build flexible models with different joint and marginal behaviours. The methods presented here can be used directly to evaluate the distribution and log-density function and for sampling; corresponding methods are implemented in the \R\ package \texttt{nvmix}, too.

A possible limitation of our methods is the assumption of a computationally tractable quantile function of the mixing variable $W$. For more complicated distributions such quantile function may not be available so that an avenue for future research could be to modify our methods so that they work with a random number generator (RNG) for $W$ (for instance, based on acceptance-rejection algorithms). While sampling and estimating the distribution function is possible when instead of the quantile function of $W$ a RNG for $W$ is provided, this is not the case for estimating the log-density (and thus for the fitting procedure) as our methods are adaptive and thus require sampling in certain low-probability subregions of the support of $W$.

We also demonstrated via a few examples how variable reordering affects Sobol' indices of the integrand and therefore the effective dimension; given by how much the reordering improves the performance of our RQMC estimator for the distribution function, we believe it would be interesting to explore if this idea can be exploited in other problems as well.
\clearpage

 \appendix

 \section{Evaluation of singular normal variance mixtures}\label{sec:app:singular}

If $\Sigma\in\IR^{d\times d}$ is positive semidefinite with rank $1\leq r<d$, the resulting singular normal variance mixture can be estimated by applying results described in \cite{genzkwong2000}, who developed an accurate method to evaluate the distribution function of a multivariate normal distribution with singular correlation matrix $\Sigma$, see also \cite[Section 5.2]{genzbretz2009} for more details.

Let $\Sigma = CC\T$ with $C_{ij} = 0$ for $j>r$, $i=1,\dots,d$, that is, $C$ is lower triangular with some diagonal elements being zero; see \cite{healy1968} for an algorithm to compute such $C$ which uses a numerical tolerance to determine zero-entries. After permutations and scalings (that must also be applied to $\ba$ and $\bb$), $C$ shall have the following form where ``$\ast$'' denotes an entry that can be zero or non-zero:
$$ C = \left( \begin{matrix} 	1	  & 0         & 0         & \dots   & \dots   & \dots  & \dots    & 0 \\
					\vdots & \vdots 	& \vdots & \vdots & \vdots & \vdots & \vdots & \vdots \\
					1  	  & 0		& \dots   & \dots   & \dots   & \dots  & \dots   & 0 \\
					\ast  	  & 1		& 0         & \dots   & \dots   & \dots  & \dots   & 0 \\
					\vdots & \vdots 	& \vdots & \vdots & \vdots & \vdots & \vdots & \vdots \\
					\ast  	  & 1		& 0         & \dots   & \dots   & \dots  & \dots   & 0 \\
					\vdots &            &             &           &            &           &            & \vdots \\
					\ast  	  & \ast		& \dots    & \ast        & 1        & 0        & \dots   & 0 \\
					\vdots & \vdots 	& \vdots & \vdots & \vdots & \vdots & \vdots & \vdots \\
					\ast  	  & \ast		& \dots    &\ast         & 1        & 0        & \dots   & 0 \end{matrix} \right)
					\begin{matrix} 1 \\ \vdots \\ k_1 \\ 1 \\ \vdots \\ k_2 \\ \vdots \\ 1 \\ \vdots \\ k_r \end{matrix}$$

Note that $\sum_{j=1}^r k_j = d$. Define $m_i = \sum_{j=1}^{i-1}k_j$ with $m_1 = 0$.
As demonstrated in \cite{genzkwong2000},  $\Phi_{\Sigma}$ can then be written in a similar fashion as in \eqref{eq:normalcdf}:
\begin{align}\label{eq:normalcdfsingular}
    \Phi_\Sigma(\ba, \bb) = \int_{\ba < C\by \leq \bb} \phi(y_1)\dots\ \phi(y_d) \,\rd \by = \int_{\ta_1}^{\tb_1}\phi(y_1)\dots  \int_{\ta_r}^{\tb_r}\phi(y_r)\,\rd y_r\dots \rd y_1
 \end{align}
Note that the $r$-dimensional integral still has $d$ active constraints: For variable $l$, the $k_l$ constraints $\ba_j < C_j\T \by \leq \bb_j$ for $j\in\{m_l+1,\dots,m_{l+1}\}$ need to be satisfied simultaneously so that the limits in \eqref{eq:normalcdfsingular} are given by
\begin{align*}
    \ta_l = \max_{ m_l < i \leq m_{l+1}} \left\{ a_i - \sum_{j=1}^{l-1} C_{i, j} y_j\right\}\quad\text{and}\quad \tb_l = \min_{ m_l < i \leq m_{l+1}} \left\{ b_i - \sum_{j=1}^{l-1} C_{i, j} y_j\right\}
\end{align*}
for $l=1,\dots,r$.

This idea can be generalized to singular normal variance mixtures. Proceeding as in Section~\ref{sec:reformulationFab} one obtains
$$ F(\ba, \bb) = \int_{ (0,1)^r} g(\bu)\,\rd u,\quad  g(\bu) = \prod_{l=1}^r \left(  e_l(u_0,\dots,u_{l-1}) - d_l(u_0,\dots,u_{l-1})\right) $$
with
\begin{align*}
    d_l(u_0,\dots, u_{l-1}) &= \Phi\left(\max_{ m_l < i \leq m_{l+1}} \left\{ \frac{a_i}{\sqrt{F_W^\i(u_0)}} - \sum_{j=1}^{l-1} C_{i, j} \Phi^{-1}(d_j + u_j(e_j-d_j))\right\}\right),\\
    e_l(u_0,\dots, u_{l-1}) &= \Phi\left(\min_{ m_l < i \leq m_{l+1}} \left\{ \frac{b_i}{\sqrt{F_W^\i(u_0)}} - \sum_{j=1}^{l-1} C_{i, j} \Phi^{-1}(d_j + u_j(e_j-d_j))\right\}\right)
\end{align*}
for $l=1,\dots,r$.  The RQMC methods described in Section~\ref{sec:mcrqmc} can be applied to the problem in this form to estimate $F(\ba,\bb)$. The main difference is that the dimension of the problem in the singular case is given by the rank $r$ as opposed to the dimension $d>r$ of the normal variance mixture.

\section{Gamma Mixture Models}\label{sec:gammamix}

For statistical purposes it is often interesting to study the distribution of
the squared Mahalanobis distance of $\bX\sim\NVM_d(\bmu,\Sigma, F_W)$ given by $D^2(\bX; \bmu, \Sigma) = (\bX-\bmu)\T \Sigma^{-1} (\bX-\bmu)$. We write $D^2 := D^2(\bX; \bmu, \Sigma)$ if there is no confusion.

It follows readily from the stochastic representation~\eqref{eq:nvm:stoch:rep} of $\bX$ that, in distribution,
$$ D^2 = W\, X^2$$
where $X^2\sim \chi^2_d$. This immediately gives rise to a sampling algorithm to generate random variates from $D^2$. Since a $\chi^2$ distribution is a special case of a gamma distribution, it follows that $ D^2\mid W \sim \Gamma(d/2, 2W)$
where $\Gamma(\alpha, \beta)$ denotes a gamma distribution with shape $\alpha>0$ and scale $\beta>0$ which admits the density
$ f_{\Gamma(\alpha, \beta)}(x) = ( \beta^\alpha \Gamma(\alpha))^{-1}  x^{\alpha-1} e^{-x/\beta}$, $x>0$, and distribution function
$ F_{\Gamma}(x; \alpha,\beta)=\int_0^x f_{\Gamma(\alpha, \beta)}(t)\,\rd t$
for $x>0$. The function $\Gamma(z)=\int_0^\infty t^{z-1}e^{-t}\,\rd t$, $z>0$ denotes the gamma function.

In the special case where $W=1$ almost surely, $D^2\sim\chi^2_d$; if $W$ follows an inverse-gamma distribution so that $\bX$ follows a multivariate $t$ with $\nu>0$ degrees of freedom, it can be easily seen that $D^2 / d \sim \text{F}(d, \nu)$. For the general case where only $F_W^\i$ is available, we can use methods similar to the ones developed so far to approximate the density and the distribution function of $D^2$.

\paragraph{Estimating the distribution function of $D^2$}
Using a conditioning argument similar to the normal variance mixture case, we obtain that
$$ F_{D^2}(x) = \P(D^2\leq x) = \E\left( F_{\Gamma(d/2,2)}\left(\frac{x}{W}\right)\right),\;\;\; x \geq 0.$$
This univariate integral can be approximated directly using an RQMC approach similar to Algorithm~\ref{alg:RQMC:F:a:b}. An implementation can be found in the function \texttt{pgammamix()} in the \R\ package \texttt{nvmix}.

\paragraph{Estimating the density function of $D^2$}
In a similar fashion as in the derivation of Equation~\eqref{eq:densityX}, the density of $D^2$ can be calculated as $ f_{D^2}(x) = \int_0^1\tilde{h}(u)\,\rd u$ for $x>0$, where
$$\tilde{h}(u) = \frac{1}{\Gamma(d/2) (2 F_W^\i(u))^{d/2}} x^{d/2-1}\exp\left(-\frac{x}{2F_W^\i(u)}\right), \;\;\; u\in(0,1).$$

The functions $\tilde{h}$ and $h$ from Equation~\eqref{eq:integrandh} differ only in constants with respect to $u$, the functional form is identical. Algorithm~\ref{alg:RQMC:f:x:adaptive}  can then, with some slight modifications, be used to estimate the density $ f_{D^2}(x)$ (or $\log f_{D^2}(x)$); see also Remark~\ref{remark:htilde}. This is implemented in the function \texttt{dgammamix()} in the \R\ package \texttt{nvmix}.

\paragraph{Estimating the quantile function of $D^2$}
Many applications, such as graphical goodness-of-fit assessment or random variate generation, rely on the quantile function of $D^2$. Note that both the density and the distribution
function of $D^2$ can be estimated as discussed above; the quantile function can then be estimated by
numerically solving the equation $F_{D^2}(q_u) - u= 0$ for $q_u$ where $u\in(0,1)$ is given. We suggest using
Newton's method: In iteration $k\geq 1$, given a current iterate $q_u^{(k)}$, the next iterate is given by
\begin{align*}
 q_u^{(k+1)} &= q_u^{(k)} - \frac{ F_{D^2}(q_u^{(k)}) - u}{f_{D^2}(q_u^{(k)})}\\
 &=q_u^{(k)} - \operatorname{sign}(F_{D^2}(q_u^{(k)}) - u_i) \exp\left\{ \log\left(|F_{D^2}(q_u^{(k)}) - u_i|\right) - \log f_{D^2}(q_u^{(k)})\right\}.
\end{align*}

The second line is a numerically more stable version of the first. We remark that (potentially) many calls to $F_{D^2}(\cdot)$ and $f_{D^2}(\cdot)$ are necessary until convergence takes place. We also note that in most applications, the quantile function has to be evaluated at multiple inputs, say $u_1,\dots,u_n$. In order to reduce run time, one can sort the inputs $u_i$ in increasing order and also store all calls to $F_{D^2}(\cdot)$ and $f_{D^2}(\cdot)$. These values can be used as starting values for the next quantile calculation. If they are reasonably close to the true quantile, the procedure enjoys local quadratic convergence so that only a few calls to $F_{D^2}(\cdot)$ and $f_{D^2}(\cdot)$ are needed. Furthermore, $F_{D^2}(\cdot)$ and $f_{D^2}(\cdot)$ can be estimated simultaneously using the same realizations of $W$, and all those realizations can also be stored so that they do not need to be generated more often than necessary. This is implemented in the function \texttt{qgammamix()} in the \R\ package \texttt{nvmix}; the same idea can be exploited to estimate the
quantile function of univariate normal variance mixtures which is implemented in the function \texttt{qnvmix()}.

\section{Algorithms}

\begin{algorithm}[RQMC Algorithm to estimate $\log\mu$ where $\mu=\int_{(0,1)^d} g(\bu)\;\rd \bu$.]\label{alg:RQMC:logmu}
Given $\eps$, $B$, $n_0$, $\imax$, estimate $\log\mu=\log(\int_{(0,1)^d} g(\bu)\;\rd \bu)$ via:
    \begin{enumerate}
    \item Set $n=n_0$, $i=1$, and compute $\hmurqmc_{b,n,\log}=\hmurqmc_{b,0,n_0,\log}$ for $b=1,\dots,B$ and $\hmurqmc_{n,\log}$ from~\eqref{eq:hmurqmcn:log} and~\eqref{eq:hmurqmcbn:log}.
            \item Set $\hat{\eps}=3.5 \hsigma_{\hmurqmc_{n,\log}}$ with $\hsigma_{\hmurqmc_{n,\log}}$ as in~\eqref{eq:sigma:rqmc}.
        \item While $\hat\eps>\eps$ and $i\leq \imax$ do:
        \begin{enumerate}
            \item Set $n=n+n_0$, compute $\hmurqmc_{b,in_0,(i+1)n_0,\log}$, $b=1,\dots,B$ and update $\hmurqmc_{b,n,\log} = -\log(i+1) + \LSE(i\hmurqmc_{b,n},\hmurqmc_{b,in_0,(i+1)n_0})$ for $b=1,\dots,B$.
           \item Update $\hmurqmc_{n,\log}=-\log(B)+\LSE(\hmurqmc_{1,n,\log},\dots,\hmurqmc_{B,n,\log})$ and update $\hat{\eps}=3.5 \hsigma_{\hmurqmc_{n,\log}}$
            \item Set $i=i+1$.
        \end{enumerate}
        \item Return $\hmurqmc_{n,\log}$.
 \end{enumerate}
 \end{algorithm}

\begin{algorithm}[Variable reordering]\label{alg:precond}
    \begin{enumerate}
        \item Start with given $\ba, \bb$ and $\Sigma$.
        \item Calculate or approximate $\mu_{\sqrt{W}} = \E(\sqrt{W})$.
        \item
        \begin{enumerate}
            \item[a)] Choose the first integration variable as
            $$i = \argmin_{j\in\{1,\dots,d\}} \left\{ \Phi\left( \frac{b_j}{\mu_{\sqrt{W}}  \sqrt{\Sigma_{jj}}}  \right) -\Phi\left( \frac{a_j}{\mu_{\sqrt{W}} \sqrt{\Sigma_{jj}}}  \right) \right\}.$$
            Swap components 1 and $i$ of $\ba$ and $\bb$ and interchange both rows and columns of $\Sigma$ corresponding to the variables $i$ and 1.
            \item[b)] Update $C_{11}=\sqrt{\Sigma_{11}}$ and $C_{j1}=\Sigma_{j1} / C_{11}$ for $j=1,\dots,d$. Set
            $$ y_1 = \frac{ \int_{\hat{a}_1 }^{\hat{b}_1 }s\phi(s)ds}{\Phi(\hat{b}_1 )-\Phi(\hat{a}_1) }$$ as expected value for $u_1$, where
            $$ \hat{a}_1 = \frac{a_1}{\mu_{\sqrt{W}} C_{11}}\quad \text{and} \quad \hat{b}_1 = \frac{b_1 }{\mu_{\sqrt{W}} C_{11}}.$$
            This is the same as $\E(Z\mid Z\in [\hat{a}_1 ,\hat{b}_1] )$ for $Z\sim N(0,1)$.
        \end{enumerate}
        \item For $j=2,\dots,d$,
        \begin{enumerate}
            \item[a)] Choose the $j$th integration variable as
            $$ i = \argmin_{ l\in\{j,\dots,d\}} \left\{  \Phi\left( \frac{\frac{b_l}{\mu_{\sqrt{W}}} -
            \sum_{k=1}^{j-1} C_{lk}y_k}{\sqrt{\Sigma_{l,l}-\sum_{k=1}^{j-1}C_{lk}^2}} \right) -
            \Phi\left( \frac{\frac{a_l}{\mu_{\sqrt{W}}} -
            \sum_{k=1}^{j-1} C_{lk}y_k}{\sqrt{\Sigma_{l,l}-\sum_{k=1}^{j-1}C_{lk}^2}} \right)\right\}.$$
            Swap components $i$ and $j$ of $\ba$ and $\bb$ and interchange both rows
            and columns of $\Sigma$ corresponding to variables $i$ and $j$ and interchange rows $i$ and $j$ in $C$.
            \item[b)] Update $C_{jj}=\sqrt{\Sigma_{jj}-\sum_{k=1}^{j-1}C_{jk}^2}$ and $C_{lj} = \frac{1}{C_{jj}}\left( \Sigma_{lj}-\sum_{k=1}^{j-1}C_{jk}C_{lk}\right)$ for $l=j+1,\dots,d$ and set
        $$y_j = \frac{ \int_{\hat{a}_j }^{\hat{b}_j }s\phi(s)ds}{\Phi(\hat{b}_j )-\Phi(\hat{a}_j ) }$$
        where
        $$ \hat{a}_j = \frac{\frac{a_j}{\mu_{\sqrt{W}}} - \sum_{k=1}^{j-1} C_{jk}y_k }{C_{jj}}\quad \text{and}\quad \hat{b}_j = \frac{\frac{b_j}{\mu_{\sqrt{W}}} - \sum_{k=1}^{j-1} C_{jk}y_k }{C_{jj}}.$$
       \end{enumerate}
\end{enumerate}
\end{algorithm}

\clearpage
\printbibliography[heading=bibintoc]
\end{document}

%
%
%
%